\documentclass[11pt,tbtags,reqno]{amsart}

\RequirePackage{amsthm,amsmath}
\RequirePackage[numbers,square]{natbib}

\RequirePackage{lscape}

\RequirePackage{calligra}
\RequirePackage{pstricks,graphicx}
\RequirePackage{tabmac}
\RequirePackage{mathrsfs}
\RequirePackage{latexsym}
\RequirePackage{amsthm,amsopn,amsfonts,amssymb,epsfig,stmaryrd}
\RequirePackage[active]{srcltx}

\newtheorem{theorem}{Theorem}
\newtheorem{lemma}[theorem]{Lemma}
\newtheorem{proposition}[theorem]{Proposition}

\newtheorem{conjecture}[theorem]{Conjecture}

\newtheorem{corollary}[theorem]{Corollary}

\newtheorem{definition}[theorem]{Definition}

\theoremstyle{remark}
\newtheorem{remark}[theorem]{Remark}

\def\la{{\lambda}}

\def\cal L{{\mathcal L}}

\def\cd{\circledast}
\def\deg{ {\rm {deg}}}

\let\L\langle
\let\R\rangle

\let\g\gamma
\let\s\sigma
\let\y\infty

\def\cd{{\circledast}}

\def\lrw{\leftrightarrow}

\def \uw { \uparrow}\def \dw { \downarrow}

\def\s{\sigma}

\def\Lt{{\tilde \La}}
\let\n\noindent
\let\De\Delta

\def\K{{\mathcal K}}

\def\B{{\mathcal B}}

\let\la\lambda
\let\La\Lambda
\let\a\alpha

\let\Om\Omega

\let\ta\theta

\let\rw\rightarrow
\let\Rw\Rightarrow

\newcommand{\tcercle}[1]{\ensuremath{\setlength{\unitlength}{1ex}\begin{picture}(2.8,2.8)\put(1.4,1.4){\circle{2.8}\makebox(-5.6,0){#1}}\end{picture}}}

\let\n\noindent

\def\K{{\mathcal K}}

\def\B{{\mathcal B}}

\let\la\lambda
\let\La\Lambda
\let\Om\Omega

\let\ta\theta

\let\rw\rightarrow

\newcommand{\LL}{\ensuremath{\langle\!\langle}}
\newcommand{\RR}{\ensuremath{\rangle\!\rangle}}

\def\cd{{\circledast}}

\def\lrw{\leftrightarrow}

\def\B{\mathcal{B}}

\def\S{\mathcal{S}}

\def\Ke{\mathsf{K}}
\def\beq{\begin{equation}}
\def\eeq{\end{equation}}
\def\inv{\text{inv}}

\begin{document}

\title[Macdonald polynomials in superspace]{Macdonald polynomials in superspace as eigenfunctions of commuting
operators}
\thanks{This article is dedicated to Professor Adriano Garsia on the occasion of his 84th birthday.  The authors are extremely grateful to Alain Lascoux for spending time and effort to provide them a detailed outline of the proof of 
Proposition~\ref{quasip}.}

\thanks{This work was supported by FQRNT, NSERC,  FONDECYT \#1090034 and \#1090016,   and   by CONICYT's  anillo ACT56.}

\author{O. Blondeau-Fournier}
\address{D\'epartement de physique, de g\'enie physique et
d'optique, Universit\'e Laval,  Qu\'ebec, Canada,  G1V 0A6.}
\email{olivier.b-fournier.1@ulaval.ca}
\author{P. Desrosiers}

\address{Instituto de Matem\'atica y F\'{\i}sica, Universidad de
Talca, Casilla 747, Talca, Chile.}
\email{patrick@inst-mat.utalca.cl}

\author{L. Lapointe}
\address{Instituto de Matem\'atica y F\'{\i}sica, Universidad de
Talca, Casilla 747, Talca, Chile.}
\email{lapointe@inst-mat.utalca.cl }
\author{P. Mathieu}
\address{D\'epartement de physique, de g\'enie physique et
d'optique, Universit\'e Laval,  Qu\'ebec, Canada,  G1V 0A6.}
\email{pmathieu@phy.ulaval.ca}

\begin{abstract}
A generalization of the Macdonald polynomials depending upon both commuting and anticommuting variables has been introduced recently. The construction relies on certain orthogonality and triangularity relations. Although many 
 superpolynomials were constructed as solutions of a highly over-determined system, the existence issue was left open. This is resolved here: we demonstrate that the underlying construction has a (unique) solution. The proof uses, as a starting point, the definition of the Macdonald superpolynomials in terms of  the Macdonald non-symmetric polynomials via a non-standard (anti)symmetrization and a suitable dressing by anticommuting monomials. This relationship naturally suggests the form of two families of commuting operators that have the defined superpolynomials as their common eigenfunctions. These eigenfunctions are then shown to 
   be triangular and orthogonal. Up to a normalization, these two conditions  uniquely characterize these superpolynomials. Moreover, the Macdonald superpolynomials are found to be orthogonal with respect to a second (constant-term-type) scalar product, and its norm is evaluated. 
   The latter is shown to match (up to a $q$-power) the conjectured norm with respect to the original scalar product.
   Finally, we recall the super-version of the Macdonald positivity conjecture  and present two new conjectures which both provide a remarkable relationship between the new $(q,t)$-Kostka coefficients and the usual ones.
\end{abstract}

\date{December 2012}

\maketitle
\newpage 
\tableofcontents
\normalsize
\newpage

\section{Introduction}

\subsection{Macdonald superpolynomials and related positivity conjectures}

A candidate for the superspace extension of the Macdonald polynomials -- Macdonald superpolynomials for short -- has been obtained in \cite{BDLM}. 
Such an extension involves the anticommuting variables $\ta_1,\theta_2,\dots$ (with 
$\ta_i^2=0$), as well as the usual commuting variables $x_1,x_2,\dots$
The superspace approach turns out to be a very restrictive framework: each variable $x_i$ is considered to be  paired with an anticommuting variable $\ta_i$, so that  symmetric superpolynomials are required to be invariant under the interchange of pairs $(x_i,\ta_i)\lrw(x_j,\ta_j)$ \cite{DLMnpb}.

The construction in \cite{BDLM} is presented  as a conjecture (a point developed shortly). But the first exploration of the resulting superpolynomials 
revealed a very rich structure. As expected,  this two-parameter $(q,t)$ family of superpolynomials contains, in the appropriate limit ($q=t^\a, \, t\rw 1$), the Jack superpolynomials \cite{DLMcmp2}. But, what was totally unexpected a priori, is that  it contains two versions of the Hall-Littlewood superpolynomials (for $q=0$ and $q\rw\y$). Moreover, from each version of the latter, we can define a natural extension of the Schur polynomials (for $t=0$ and $t\rw\y$ respectively). Both types of Schurs have a positive integral decomposition into the monomial superpolynomials (a property that is not verified for the Jack superpolynomials at $\a=1$). But what is more remarkable is that, using the $q=t=0$ Schurs, we could conjecture a generalization of the Macdonald positivity conjecture (reviewed below).

The present article, although completely independent, is a continuation of our previous work \cite{BDLM}. It addresses the following issues: the existence  of the Macdonald superpolynomials, their relation with the non-symmetric Macdonald polynomials, and their characterization as an eigenvalue problem. As an aside, new conjectures for the Kostka coefficients are presented.

\subsection{The main result: an existence proof}
The conjectural
construction of \cite{BDLM} 
is proved here to be a valid  characterization of the Macdonald polynomials
in superspace.  To clarify this point, we first recall the definition 
of the ordinary Macdonald polynomials \cite{Mac}.

The Macdonald polynomials $P_\la=P_\la(x;q,t)$, in the variables 
$x=x_1,x_2,\dots$, are characterized by the two conditions:
 \begin{equation}\label{co12}\begin{array}{lll} 1)& P_{\lambda} =
m_{\lambda} + \text{lower terms},\\
2)&\LL  P_{\la}| P_{\mu} \RR_{q,t} =0\quad\text{if}\quad \la\ne\mu \, .
\end{array}\end{equation}
The triangular decomposition refers to the usual dominance order on partitions (see \eqref{ordre} below) and the $m_\la$'s are the monomial symmetric
functions. 
The  orthogonality relation is defined in the power-sum basis $p_\la=p_{\la_1}\cdots p_{\la_\ell}$, with { $p_r=\sum_{i\geq 1} x^r_i$ }, as
\begin{equation}\label{ortoqt}  \LL {p_\la}|
{p_\mu}\RR_{q,t}= z_\la\, \delta_{\la\mu}\,
\prod_{i=1}^{\ell(\la)}\frac{1-q^{\la_i}}{1-t^{\la_i}}\, , \qquad
z_{\la}=\prod_{i \geq 1} i^{n_{\la}(i)} {n_{\la}(i)!}\, ,
\end{equation}
$n_{\la}(i)$  being the number of parts in $\la$ equal to $i$. 
Since the dominance ordering is partial (for degrees $\geq 6$), 
the orthogonality constraint leads to an overdetermined system at each degree $\geq 6$. It is 
therefore necessary to show that there exists such a family 
of polynomials.  
This is generally done through an eigenvalue-problem characterization.

The brute-force approach followed in \cite{BDLM} 
was to look for a suitable deformation of the scalar product \eqref{ortoqt} that allows for nontrivial solutions to these systems. In this way, a
 candidate scalar product was identified and  a large number of Macdonald superpolynomials were constructed. The correctness of the construction was corroborated  by various conjectural properties that provide natural extensions to superspace of classical results on Macdonald polynomials.
However, establishing the existence of the Macdonald superpolynomials remained an open problem,  the corresponding eigenvalue problem being still missing.

The existence issue is resolved here: it is demonstrated that the superspace extension of the criteria \eqref{co12} has a solution. This is our main result (the notation that follows is explained in full detail in Section \ref{sps}).
\begin{theorem}\label{theo1}
Given a superpartition $\La=(\La^a;\La^s)$ of fermionic degree $m$,
there is a unique symmetric superpolynomial  $P_\La=P_\La(x,\ta;q,t)$,
with $x=(x_1,x_2,\dots)$ and $\theta=(\theta_1,\theta_2,\dots)$, such that: 
\begin{equation}\label{mac1}
\begin{array}{lll} 1)& P_{\Lambda} =
m_{\Lambda} + \text{lower terms},\\
2)&\LL  P_{\La}| P_{\Om} \RR_{q,t} =0\quad\text{if}\quad \La\ne\Om,
\end{array}\end{equation}
where $m_\La$ is a monomial superpolynomial and where lower terms refer to 
the dominance  ordering on superpartitions (see (\ref{eqorder1})). The scalar product is 
defined by 
\begin{equation}\label{newsp}
\LL {p_\La}|{p_\Om}\RR_{q,t}=(-1)^{\binom{m}2}\,z_\La(q,t)\delta_{\La\Omega}\, ,
\end{equation}where\begin{equation}\label{newz}
\qquad z_\La(q,t)
= z_{\La^s} 
\, q^{|\La^a|} \prod_{i=1}^{\ell(\La^s)}
\frac{1-q^{\La^s_i}}{1-t^{\La^s_i}}\, .\, \end{equation}  
\end{theorem}

\subsection{A key relationship: the connection with Macdonald non-symme\-tric polynomials}
The existence proof proceeds indirectly, via an alternative definition of the Macdonald superpolynomials, this one relying on a suitable (anti)symmetri\-zation of the non-symme\-tric Macdonald polynomials. Such a construction is akin to that of the Jack superpolynomials in terms of the non-symmetric polynomials worked out in \cite[Sect. 9]{DLMcmp2}. However, the present construction turns out to be trickier than in the Jack case. Indeed,  Macdonald polynomials with prescribed symmetry -- a priori, expected building blocks -- are obtained by $t$-(anti)symmetrization of some subset of variables, where the role of symmetric group is played by the Hecke algebra \cite{BDF,Ba}.  However, anticommuting variable cannot be $t$-antisymmetrized, and
the construction has to incorporate both the usual antisymmetrization and the $t$-symmetrization, the two operations being applied to distinct set of variables. This is made explicit in eq. \eqref{PvsE}. 

This definition of the Macdonald superpolynomials in terms of non-symmetric Macdonald polynomials is a crucial step toward the characterization of the former in terms of an eigenvalue problem. Recall that these non-symmetric polynomials are eigenfunctions of the Cherednik operators \cite{Che}. The mere symmetrization procedure that defines the Macdonald superpolynomials indicates how to dress symmetric combinations of the Cherednik operators in order to generate operators whose eigenfunctions are the Macdonald superpolynomials. Because we have two sets of variables (commuting and anticommuting), we need two families of commuting operators to obtain a non-degenerate eigenvalue characterization (see e.g. \cite{DLMcmp2} for the Jack case). We thereby construct two generating functions of commuting operators. These are the natural extension of the Sekiguchi operators  characterizing the Jack superpolynomials and introduced in \cite[Sect. 3]{DLM0}. To fully characterize the eigenvalue problem, it is sufficient to consider the simplest representative of each family. 

The proof of the existence of the superpolynomials defined by the two conditions in Theorem \ref{theo1} proceeds along standard arguments. But a crucial and difficult result that needs to be established is the self-adjoint property of the two eigenoperators. The outline of this proof was provided to us by Alain Lascoux \cite{Laspv}.

The relationship between Macdonald superpolynomials and the non-symme\-tric polynomials has a further direct consequence: it implies a second orthogonality relation, where in this case, the scalar product is a constant-term expression. We show that when the number of variables tends to infinity, the norm calculated from the constant-term scalar product is equal (up to a $q$-power) to the norm, 
conjectured in \cite{BDLM},
with respect to the scalar product defined in Theorem \ref{theo1}.

\subsection{Organization  of the article}
The outline of the article is the following. In Section \ref{sps}, we recall some basic definitions  related to superpolynomials and superpartitions. Section \ref{nonS} is also devoted to a review of known results, here pertaining to non-symmetric Macdonald polynomials. The definition of the Macdonald superpolynomials in terms of non-symmetric Macdonald polynomials is presented in Section \ref{smac}. In
Section \ref{2com}, we introduce two generating functions for commuting operators
and show that they have the Macdonald superpolynomials as their common eigenfunctions.
{This eigenfunction characterization allows us to demonstrate in Section \ref{tri} that the Macdonald superpolynomials have a triangular decomposition  in the monomial basis. In preparation for the orthogonality proof, the eigenvalue problem defining the Macdonald superpolynomials is simplified in Section  \ref{simp}, where it is shown to be sufficient  to consider two eigenoperators to get a non-degenerate characterization.}
The proof of the orthogonality with respect to the scalar product \eqref{newsp}, which solves the existence issue in Theorem~\ref{theo1},
is worked out in Section \ref{ortho}. 
{At first, the problem is reformulated in terms of action of the two  eigenoperators on the kernel. The long proof of the self-adjoint property of the Macdonald-type operators is worked out in Section \ref{self}. As a natural extension of these orthogonality results, a non-trivial duality relation is established in Section \ref{dual}. Given this, all the tools for the demonstration of the evaluation and norm conjectures of \cite{BDLM} are available, proof that would follow that in  \cite{DLMeva} for the Jack superpolynomials. Finally, a} second orthogonality relation is demonstrated in Section \ref{asp}.  The corresponding norm is evaluated, and shown to be equal, up to a power of $q$, to the conjectural norm of the Macdonald
superpolynomials with respect to the original scalar product.
The combinatorial identity on which this claim relies is demonstrated in the Appendix.

The last section (Section \ref{Kos}), devoted to the generalized Kostka coefficients, is somewhat off the streamline of this article but should be viewed in the context of the continuation of \cite{BDLM}.
The superspace version of the Macdonald  positivity conjecture 
is recalled in Section \ref{gkos}. Two symmetry properties of the generalized Kostka coefficients are presented. Although these are rather direct extensions of the usual  symmetry relations given in \cite{Mac}, one of these involve a new combinatorial number that is specific to superpartitions. This new data illustrates well the kind of novelties brought in by the introduction of anticommuting variables and the richness of the combinatorics of superpartitions.  Tables of Kostka coefficients are appended to this section.

Two new conjectures concerning the Kostka coefficients are given in Section \ref{2new}. Both results exhibit a different relationship between the $(q,t)$-Kostka coefficients in the $m=1$ sector with the ordinary (i.e., $m=0$) 
$(q,t)$-Kostka coefficients. In particular, Conjecture \ref{kos1} expresses the usual $(q,t)$-Kostka coefficients for partitions of degree $n$ as a 
sum of $(q,t)$-Kostka coefficients for superpartitions of degree $n-1$
and $m=1$. In other words, the super-version of the Macdonald positivity conjecture provides a refinement of the usual ones. Let us digress briefly 
and  point out that this result naturally poses the question: to which extent could we get information on the usual $(q,t)$-Kostka coefficients from the perspective of the Macdonald superpolynomials?   We intend to investigate this point elsewhere from a particular angle. Such a connexion, at this stage, could rightly be called
  ``Science fiction and Macdonald's superpolynomials'' \cite{BG}.


\section{Symmetric polynomials in superspace}\label{sps}

A polynomial {in} superspace, or equivalently, a superpolynomial, is
a polynomial in the usual $N$ variables $x_1,\ldots ,x_N$  and the $N$ anticommuting variables $\ta_1,\ldots,\ta_N$ over a certain field, which will be taken throughout this article to be $\mathbb Q(q,t)$.  
A superpolynomial $P(x,\theta)$,
with $x=(x_1,\ldots,x_N)$ and $\theta=(\theta_1,\ldots,\theta_N)$, is said to be symmetric if the following is satisfied
\cite{DLMnpb}:
\begin{equation}
\mathcal{K}_{\sigma}P(x,\theta)=P(x,\theta) \qquad {\rm for~all~} \sigma \in S_N 
\, ,
\end{equation}
where
\begin{equation}
\mathcal{K}_{\sigma}=\kappa_{\sigma}K_{\sigma},
\qquad\text{with}\quad\begin{cases}
&K_{\sigma}\,:\, (x_1,\dots,x_N) \mapsto (x_{\sigma(1)}, \dots,
x_{\sigma(N)})
\\
&\kappa_{\sigma}\,\;:\, (\theta_1,\dots,\theta_N) \mapsto (\theta_{\sigma(1)}, \dots,
\theta_{\sigma(N)}).
\end{cases}\end{equation}
The space of symmetric superpolynomials in $N$ variables
over the field $\mathbb Q(q,t)$ will
be denoted $\mathscr R_N$, and its inverse limit by $\mathscr R$ (loosely speaking, the number of variables is considered infinite in $\mathscr R$).

Before defining superpartitions, we  recall some definitions
related to partitions \cite{Mac}.
A partition $\lambda=(\lambda_1,\lambda_2,\dots)$ of degree $|\lambda|$
is a vector of non-negative integers such that
$\lambda_i \geq \lambda_{i+1}$ for $i=1,2,\dots$ and such that
$\sum_i \lambda_i=|\lambda|$.  The length $\ell(\lambda)$
of $\lambda$ is the number of non-zero entries of $\lambda$.
Each partition $\lambda$ has an associated Ferrers diagram
with $\lambda_i$ lattice squares in the $i^{th}$ row,
from the top to bottom. Any lattice square in the Ferrers diagram
is called a cell (or simply a square), where the cell $(i,j)$ is in the $i$th row and $j$th
column of the diagram.  
The conjugate $\lambda'$ of  a partition $\lambda$ is represented  by
the diagram
obtained by reflecting  $\lambda$ about the main diagonal.
Given a cell $s=(i,j)$ in $\lambda$, we let 
\begin{equation} \label{eqarms}
a_{\lambda}(s)=\lambda_i-j\, , \qquad {\rm and} \qquad l_{\lambda}(s)=\lambda_j'-i \,  .
\end{equation}
The quantities $a_{\lambda}(s)$ and $l_{\lambda}(s)$ 
are respectively called the arm-length and leg-length.
We say that the diagram $\mu$ is contained in $\la$, denoted
$\mu\subseteq \la$, if $\mu_i\leq \la_i$ for all $i$.  Finally,
$\la/\mu$ is a horizontal (resp. vertical) $n$-strip if $\mu \subseteq \lambda$, $|\lambda|-|\mu|=n$,
and the skew diagram $\la/\mu$ does not have two cells in the same column
(resp. row). 

Symmetric superpolynomials  are naturally indexed by superpartitions \cite{DLMnpb}. A superpartition $\Lambda$ of
degree $(n|m)$ and length $\ell$
  is a pair $(\Lambda^\circledast,\Lambda^*)$ of partitions
$\Lambda^\circledast$ and $\Lambda^*$ such
 that \cite{DLMeva}:
 \begin{enumerate} \item $\Lambda^* \subseteq \Lambda^\circledast$;
 \item the degree of $\Lambda^*$ is $n$;
 \item the length of $\Lambda^\circledast$ is $\ell$;
 \item the skew diagram $\Lambda^\circledast/\Lambda^*$
is both a horizontal and a vertical $m$-strip.\footnote{Some authors call such a diagram an $m$-rook strip.}
 \end{enumerate}
We refer to  $m$ and $n$ respectively as the fermionic degree 
and total degree 
of $\La$.
 Obviously, if
$\Lambda^\circledast= \Lambda^*=\lambda$,
then $\Lambda=(\lambda,\lambda)$ can be interpreted as the partition
$\lambda$.

We will also need another characterization of a superpartition.
 A superpartition $\La$ is 
a pair of partitions $(\La^a; \La^s)=(\La_{1},\ldots,\La_m;\La_{m+1},\ldots,\La_\ell)$, 
 where $\La^a$ is a partition with $m$ 
distinct parts (one of them possibly  equal to zero),
and $\La^s$ is an ordinary partition.   The correspondence 
between $(\Lambda^\circledast,\Lambda^*)$ and 
$(\Lambda^a; \Lambda^s)$ is given explicitly as follows:
given 
$(\Lambda^\circledast,\Lambda^*)$, the parts of $\Lambda^a$ correspond to the
parts of $\Lambda^*$ such that $\Lambda^{\circledast}_i\neq 
\Lambda^*_i$, while the parts of $\Lambda^s$ correspond to the
parts of $\Lambda^*$ such that $\Lambda^{\circledast}_i=\Lambda^*_i$.

The conjugate of a superpartition 
$\Lambda=(\Lambda^\circledast,\Lambda^*)$ is $\Lambda'=((\Lambda^\circledast)',(\Lambda^*)')$.
A diagrammatic representation of $\La$ is given by 
the Ferrers diagram of $\La^*$ with circles added in the cells corresponding
to $\Lambda^{\circledast}/\Lambda^*$.
For instance, if $\La=(\Lambda^a;\Lambda^s)
=(3,1,0;2,1)$,  we have $\Lambda^\circledast=(4,2,2,1,1)$ and $\Lambda^*
=(3,2,1,1)$, so that 
{\small
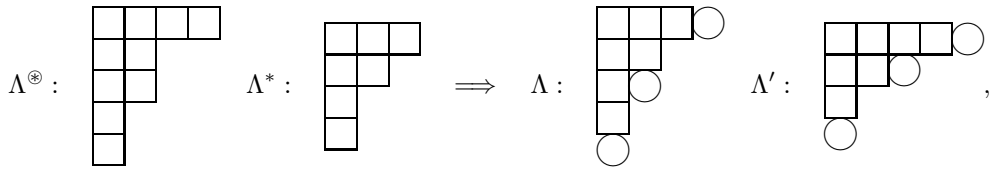
\begin{equation*} \label{exdia}
     \La^\cd:\quad{\tableau[scY]{&&&\\&\\&\\\\ \\ }} \quad
         \La^*:\quad{\tableau[scY]{&&\\&\\ \\ \\ }} \quad
 \Longrightarrow\quad      \La:\quad {\tableau[scY]{&&&\bl\tcercle{}\\&\\&\bl\tcercle{}\\ \\
    \bl\tcercle{}}} \quad    \La':\quad {\tableau[scY]{&&&&\bl\tcercle{}\\&&\bl\tcercle{}\\ \\
    \bl\tcercle{}}},
\end{equation*}}
\hspace{-0.3cm} where the last diagram illustrates the conjugation operation that corresponds, as usual, to replacing rows by columns.

The extension of the dominance ordering
to superpartitions  is \cite{DLMeva}:
\begin{equation} \label{eqorder1}
 \Omega\leq\Lambda \quad \text{iff}\quad
 \deg(\La)=\deg(\Om) ,
 \quad \Omega^* \leq \Lambda^*\quad \text{and}\quad
\Omega^{\circledast} \leq  \Lambda^{\circledast}.
\end{equation}
Note that comparing two superpartitions amounts to comparing  two pairs of ordinary partitions,
($\Omega^*$, $\La^*$) and
($\Omega^{\circledast} $, $ \La^{\circledast}$), 
with respect to the usual  dominance ordering: 
\begin{equation}\label{ordre}
   \mu \leq \la\quad\text{ iff }\quad |\mu|=|\la|\quad\text{
and }\quad \mu_1 + \cdots + \mu_i \leq \lambda_1 + \cdots + \lambda_i\quad \forall i \, . \end{equation}

Two simple bases  of the space of 
symmetric polynomials in superspace (with commuting indeterminates $x_1,\ldots ,x_N$ and anticommuting  inderterminates $\ta_1,\ldots,\ta_N$) will be particularly relevant to our work:\begin{enumerate}\item 
the extension of the monomial symmetric 
functions, $m_\La=m_\La(x,\theta)$, defined by
\begin{equation}
m_\La={\sum_{\sigma \in S_N} }' \mathcal{K}_\sigma \left(\theta_1\cdots\theta_m x_1^{\La_1}\cdots x_\ell^{\La_\ell}\right),
\end{equation}
where the sum is over the  permutations of $\{1,\ldots,N\}$ that produce distinct terms;
\item 
 the generalization of the power-sum 
symmetric
functions, $p_\La=p_\La(x,\theta)$, defined by
\begin{equation}\label{spower}
p_\La=\tilde{p}_{\La_1}\cdots\tilde{p}_{\La_m}p_{\La_{m+1}}\cdots p_{\La_{\ell}}\, ,\end{equation}
{where}\begin{equation} \tilde{p}_k=\sum_{i=1}^N\theta_ix_i^k\qquad\text{and}\qquad p_r=
\sum_{i=1}^Nx_i^r \, ,
\end{equation}  
with $k\geq 0$ and $r \geq 1$.
\end{enumerate}


\section{The non-symmetric Macdonald polynomials}\label{nonS}

The ordinary Macdonald polynomials can be defined by the conditions (1) and (2) in \eqref{co12}.
But they could alternatively be defined directly in terms of the so-called non-symmetric Macdonald    polynomials by a suitable symmetrization process \cite{Mac1,Che} (see also \cite{Mac2,Mar}). As will be shown in the following section, this can also be done for their superspace extension.  But since this result uses a fair amount of notations and definitions, it is convenient to summarize these here. 

The non-symmetric Macdonald polynomials are defined in terms of an eigenvalue problem formulated in terms of the Cherednik operators \cite{Che}. They are constructed from the 
operators $T_i$ 
defined as
\begin{equation}T_i=t+\frac{tx_i-x_{i+1}}{x_i-x_{i+1}}(K_{i,i+1}-1),\quad i=1,\ldots,N-1,\end{equation}
and
\begin{equation}
 {T_0=t+\frac{qtx_N-x_1}{qx_N-x_1}(K_{1,N}\tau_1\tau_N^{-1}-1)}\, ,
\end{equation}
where we recall that $ {K_{i,j}}$ exchanges the variables $x_i$ and ${x_{j}}$.
Note that for $t=1$, $T_i$ reduces to $K_{i,i+1}$. 
The $T_i$'s satisfy the  {affine} Hecke algebra relations  {($0\leq i\leq N-1$)}:
\begin{align}\label{Hekoalg}&(T_i-t)(T_i+1)=0\nonumber\\
&T_iT_{i+1}T_i=T_{i+1}T_iT_{i+1}\nonumber\\
&T_iT_j=T_jT_i \, ,\quad i-j \neq \pm 1 \mod N
\end{align}
where the indices are taken modulo $N$.
To define the Cherednik operators, we also need to introduce the
  $q$-shift operators 
  \begin{equation} \tau_i:\begin{cases}x_i\mapsto qx_i,\\ x_j\mapsto x_j\;
{\rm~if~} j\neq i, \end{cases}
\end{equation}
and the operator $\omega$ defined as:  
\begin{equation}\omega=K_{N-1,N}\cdots K_{1,2} \, \tau_1.
\end{equation}
We note that $\omega T_i=T_{i-1}\omega$ for $i=2,\dots,N-1$.

We are now in position to define the Cherednick operators:
\begin{equation}Y_i=t^{-N+i}T_i\cdots T_{N-1}\omega T_1^{-1}\cdots T_{i-1}^{-1},\end{equation}
where 
$ T_j^{-1}$ (also denoted  $\bar T_j$ below) is
\begin{equation}\label{Tinv}
 T_j^{-1}=t^{-1}-1+t^{-1}T_j,
 \end{equation} 
which follows from the quadratic relation \eqref{Hekoalg} of the  Hecke algebra.
 These  operators satisfy the following  relations  \cite{Che, Kiri} 
: \begin{align} \label{tsym1}
T_i \, Y_i&= Y_{i+1}T_i+(t-1)Y_i\nonumber \\
T_i \, Y_{i+1}&= Y_{i}T_i-(t-1)Y_i\nonumber \\
T_i Y_j & = Y_j T_i \quad {\rm if~} j\neq i,i+1.
\end{align}
It can be easily deduced from these relations that
\begin{equation}\label{TYi}
(Y_i+Y_{i+1})T_i= T_i (Y_i+Y_{i+1}) \qquad {\rm and } \qquad (Y_i Y_{i+1}) T_i =
T_i (Y_i Y_{i+1}). 
\end{equation}
But more importantly, the $Y_i$'s commute among each others, $[Y_i,Y_j]=0$,
and can therefore be simultaneously diagonalized. Their eigenfunctions are the
 (monic) non-symmetric Macdonald polynomials (labeled by compositions).
To be more precise,  the non-symmetric Macdonald polynomial $E_\eta$ is 
the 
unique polynomial with rational coefficients in $q$ and $t$ 
that is triangularly related to the monomials (in the Bruhat ordering on compositions)
\begin{align}\label{defEtrian}
E_\eta=x^\eta+\sum_{\nu\prec\eta}b_{\eta\nu}x^\nu
\end{align}
and that satisfies, for all $i=1,\dots,N$, 
\begin{align}Y_i E_\eta=\bar \eta_iE_\eta,\qquad\text{where}\qquad  \bar\eta_i =q^{\eta_i}t^{-\bar l_\eta(i)} \label{eigenvalY}
\end{align}
with $\bar l_\eta(i)=\# \{k<i|\eta_k\geq \eta_i\}+\# \{k>i|\eta_k> \eta_i\}$.
 The Bruhat order on compositions is defined as follows:
 \begin{align}\nu\prec\eta\quad \text{ {iff}}\quad \nu^+<\eta^+\quad \text{or} \quad \nu^+=\eta^+\quad \text{and}\quad w_\eta < w_\nu,
 \end{align}
 where $\eta^+$ is the partition associated to $\eta$ and 
$w_{\eta}$ is the unique permutation of minimal length such 
that $\eta = w_{\eta} \eta^+$ ($w_{\eta}$ permutes the entries of $\eta^+$).
In the Bruhat order on the symmetric group, $w_\eta {<} w_\nu$ iff
$w_{\eta}$ can be obtained as a  {proper} subword  of $w_{\nu}$.

The following two properties of the non-symmetric Macdonald polynomials will be needed below. 
The first one expresses the stability of the polynomials $E_\eta$ with respect to the number of variables  (see e.g. \cite[eq. (3.2)]{Mar}):
\begin{equation} \label{property1}
E_\eta (x_1,\dots,x_{N-1},0) =
\left \{ 
\begin{array}{ll}
E_{\eta_-} (x_1,\dots,x_{N-1})
& {\rm if~} 
\eta_N = 0\, , \\
0 & {\rm if~} \eta_N \neq 0\, .
\end{array} \right.
 \end{equation}
where $\eta_{-}=(\eta_1,\ldots, \eta_{N-1}) $.   
The second one gives the action of the operators $T_i$ on $E_\eta$ (see e.g. \cite[eq. (8)]{Ba} and \cite{BF}): 
\begin{equation} \label{property2}
T_i E_{\eta} = \left\{ 
\begin{array}{ll}
\left(\frac{t-1}{1-\delta_{i,\eta}^{-1}}\right) E_\eta + t E_{s_i \eta} & {\rm if~} 
\eta_i < \eta_{i+1} \, ,  \\
t E_{\eta} &  {\rm if~} 
\eta_i = \eta_{i+1} \, ,\\
\left(\frac{t-1}{1-\delta_{i,\eta}^{-1}}\right) E_\eta + \frac{(1-t{\delta_{i,\eta}})(1-t^{-1}\delta_{i,\eta})}{(1-{\delta_{i,\eta}})^2} E_{s_i \eta} & {\rm if~} 
\eta_i > \eta_{i+1} \, ,
\end{array} \right. 
\end{equation}
where $\delta_{i,\eta}=\bar \eta_i/\bar \eta_{i+1}$ {and} $s_i \eta=(\eta_1,\dots,\eta_{i-1},\eta_{i+1},\eta_i,\eta_{i+2},\dots,\eta_N)$.

Finally, we introduce the $t$-symmetrization and $t$-antisymmetrization operators of variables $x_{1},\ldots,x_{N}$ { \cite{Mac1}}:
\begin{equation}\label{upm}
 U^+_N=\sum_{\sigma\in S_N}T_\sigma\qquad\text{and}\qquad  U^-_N=\sum_{\sigma\in S_N}(-t)^{-\ell(\sigma)}T_\sigma \end{equation}
where
\begin{equation}T_\sigma=T_{i_1}\cdots T_{i_\ell}
\qquad  \text{if}\quad \sigma=s_{i_1}\cdots s_{i_\ell}.\end{equation}
Note that for any polynomial $f$ in the variables $x_1,\ldots,x_N$, we have 
$K_{i,i+1} U^+_Nf=U^+_Nf$,  but $ K_{i,i+1} U^-_Nf\neq -U_N^-f$ since \cite[eq.(2.26)]{Mar}
\begin{equation}\label{UvsA}
U^-_N\,f=t^{-\binom{N}{2}}\frac{\Delta^t_{N}}{\Delta_{N}}A_N\,{f},
\end{equation} 
where  
\begin{equation}\label{defdel} 
A_N=\sum_{\sigma \in S_N}(-1)^{\ell(\sigma)}K_\sigma\, , \quad  \Delta^t_N=\prod_{1\leq i<j\leq N}(tx_i-x_j)\, ,\quad  
\Delta_{N}=\Delta^1_N\, . 
\end{equation}
Note that $A_N$ is the usual antisymmetrization operator.
Below, we will designate by $S_m$ and $S_{m^c}$ the group of permutations of the variables $x_1,\ldots,x_m$ and $x_{m+1},\ldots,x_N$ respectively. For instance, $U^-_{m}$ and $U^+_{m^c}$ are defined as in \eqref{upm} but with $S_N$ replaced by $S_{m}$ and $S_{m^c}$ respectively. Similarly, we will frequently use the notation $\Delta^t_m$ which is defined as in \eqref{defdel} but with $N$ replaced by $m$.


\section{Macdonald superpolynomials} \label{PvsEs}

\subsection{Definition of the Macdonald superpolynomials}\label{smac}
We are now in position to define the Macdonald superpolynomials
in terms of the non-symmetric Macdonald polynomials.  We will prove later that 
the Macdonald superpolynomials defined in the next definition
do in fact provide a solution to the existence problem in Theorem \ref{theo1}.

\begin{definition}  The Macdonald superpolynomials $P_{\Lambda}={  P}_\La(x,\theta;q,t)$ are defined as
 \begin{equation}\label{PvsE} 
 {  P}_\La=
\frac{(-1)^{\binom{m}{2}}}{f_{\La^s}(t)\, t^{\inv(\La^s)} }
\sum_{\sigma \in S_N/(S_m\times S_{m^c})}\mathcal{K}_\sigma \theta_1\cdots\theta_m A_{m}U^+_{m^c}
E_{\Lambda^R} \, ,\end{equation}
where 
\begin{equation}\label{cla}
f_{\La^s}(t)=\prod_{{j\geq0}} [n_{\La^s}(j)]_t ! \, ,
\end{equation} 
with $n_{\La^s}(j)$ being the number of occurrences of $j$ in $\La^s$ and $\Lambda^R$ stands for the concatenation of $\La^a$ and $\La^s$ read in reverse order: \begin{equation}\Lambda^R=(\La_m,\ldots,\La_1,\La_N,\ldots,\La_{m+1})\, .\end{equation}
\end{definition}

 {In \eqref{PvsE}, we extended} the usual concept of inversion on a permutation to
a partition:  ${\inv}(\La^s)$ is the number of inversions in $\La^s$, the latter number being equal to
\begin{equation} \inv(\la)=\# \{n\geq i>j\, |\, \la_i<\la_j\} \, , \end{equation}
where $n$ is the number of entries in $\lambda$ (including 0's).
For instance, we have $\inv(22100)=8$.  In \eqref{cla}, we also used the following standard notation:
$$[k]_t!=[1]_t[2]_t \cdots [k]_t  \qquad {\rm with} \qquad
[m]_t=(1-t^m)/(1-t) \, .
$$ 


We first show that the stability of $E_\eta$ with respect to the number of variables can be lifted to that of $  P_{\Lambda}$.

 \begin{proposition}\label{propostable}
Suppose that $N>m$.  Then
the Macdonald superpolynomials $  P_{\Lambda}$ are stable with respect the number of variables, that is,
\begin{multline}
{  P}_{\Lambda}(x_1,\dots,x_{N-1},0,\theta_1,\dots,\theta_{N-1},0) =\\
\left \{ 
\begin{array}{ll}
{  P}_{\Lambda_-}(x_1,\dots,x_{N-1},\theta_1,\dots,\theta_{N-1}) & {\rm if~} 
\Lambda_N = 0\, , \\
0 & {\rm if~} \Lambda_N \neq 0\, ,
\end{array} \right.
\end{multline}
where $\Lambda_-=(\Lambda_1,\dots, \La_m ; \La_{m+1}, \ldots, \Lambda_{N-1})$.
\end{proposition}
\begin{proof}
From the definition of ${  P}_\Lambda$ it is immediate that it suffices to show that
\begin{equation}\label{redU}
\left[ U^+_{m^c} E_{\Lambda^R} \right]_{x_N=0} =  
\left \{
\begin{array}{ll}
\frac{c_{\Lambda_-}}{c_{\Lambda}} U^+_{m^c_-} E_{\Lambda_-^R} & {\rm if~} \Lambda_N=0 \, ,\\
0 & {\rm if~} \Lambda_N \neq 0\, ,
\end{array} \right.
\end{equation}
where 
 $m^c_-$ refers to the reduced set  of variables $x_{m+1},\ldots,x_{N-1}$ and 
 \begin{equation}
 \Lambda_-^R=(\La_-)^R=(\La_m,\ldots,\La_1,\La_{N-1},\ldots,\Lambda_{m+1}).
 \end{equation}
We stress that although $\La_-$ is a superpartition, $\La_-^R$ is a composition.
The constant $c_{\Lambda}(t)$ is the normalization constant in \eqref{PvsE}:
\begin{equation} \label{clambda}
c_{\Lambda}(t)=\frac{(-1)^{\binom{m}{2}}}{f_{\La^s}(t)\, t^{\inv(\La^s)} }\, .
\end{equation}
Now $U^+_{m^c} $ can be factorized as follows:
\begin{equation} \label{tfacto}
U^+_{m^c}=U^+_{m^c_-} (1 + T_{N-1} + T_{N-1} T_{N-2} + \dots + T_{N-1} \cdots  T_{m+1}),
\end{equation}
so that
\begin{multline}
\left[ U^+_{m^c} E_{\Lambda^R} \right]_{x_N=0} =\\  
U^+_{m^c_-} \Bigl[ 
(1 + T_{N-1} + T_{N-1} T_{N-2} + \dots + T_{N-1} \cdots T_{m+1})
E_{\Lambda^R} \Bigr]_{x_N=0}.
\end{multline}
Now, if $\Lambda^s$ has exactly $k=n_{\Lambda^s}(0)$ 
zero entries, it is easy to see
from \eqref{property1} and \eqref{property2}
that
\begin{equation}
\Bigl[ T_{N-1} \cdots T_{m+i} E_{\Lambda^R} \Bigr]_{x_N=0} =
\left \{
\begin{array}{ll}
t^{N-m-i} E_{\Lambda_-^R} & {\rm if~} i\leq k \, ,\\
0 & {\rm if~} i> k\, . 
\end{array} \right.
\end{equation}
Hence
\begin{equation}
\left[ U^+_{m^c} E_{\Lambda^R} \right]_{x_N=0} 
= \left \{
\begin{array}{ll}
t^{N-m-k} [k]_t U^+_{m^c_-} E_{\Lambda_-^R} & {\rm if~}  k > 0\, , \\
0 & {\rm if~} k=0 \, .
\end{array} \right.
\end{equation}
If $k>0$, we have that $\inv(\Lambda^s_-)=\inv(\Lambda^{s}) -N+m+k$.  Therefore,
$t^{N-m-k} [k]_t=c_{\Lambda_-}/c_{\Lambda}$ and the proposition follows. 
\end{proof}

\subsection{Two families of commuting (eigen)operators}\label{2com}

We now introduce two families of operators generalizing those introduced for the Jack superpolynomials \cite[Sect. 3]{DLM0} and defined as
\begin{equation}D^*(u;q,t)= 
\sum_{m=0}^N \sum_{\sigma \in S_N/(S_m\times S_{m^c})}\mathcal{K}_\sigma\left(\frac{\Delta_{m}}{\Delta^t_{m}}\prod_{i=1}^N(1+uY_i)
\frac{\Delta^t_{m}}{\Delta_{m}}\pi_{1,\ldots,m}\right)
\end{equation}
and
\begin{multline}
D^\circledast(u;q,t)=\\\qquad \qquad
\sum_{m=0}^N  \sum_{\sigma \in S_N/(S_m\times S_{m^c})}\mathcal{K}_\sigma\left(\frac{\Delta_{m}}{\Delta^t_{m}}\prod_{i=1}^m(1+uqY_i)\prod_{i=m+1}^N(1+uY_i)
\frac{\Delta^t_{m}}{\Delta_{m}}\pi_{1,\ldots,m}\right) \, .
\end{multline}
The operator $\pi_{1,\dots,m}$ is the  projection operator defined as
\begin{equation}\pi_{1,\dots,m}=\prod_{i= 1}^m\theta_i\partial_{\theta_i}\prod_{j=m+1}^N \partial_{\theta_j}{\theta_j}\, . \end{equation} 
In this equation,  $\partial_{\theta_i}$ denotes  the standard  derivative with respect to the Grassmann variable $\theta_i$, which is a linear operator such that,  for all polynomials $f=f(x,\theta)$ and $i,j\in\{1,\ldots, N\}$, 
\begin{equation} \partial_{\theta_i}\left(x_j f\right) = x_j\partial_{\theta_i}\left(f\right)\,\qquad   \partial_{\theta_i}\left(\theta_j f\right)=\delta_{i,j}\,f-\theta_j\,\partial_{\theta_i}\left(f\right)\, ,\end{equation}
and \beq \partial_{\theta_i}\partial_{\theta_j}\left(f\right)=-\partial_{\theta_j}\partial_{\theta_i}\left(f\right)\quad \Longrightarrow\quad \partial_{\theta_i}^2(f)=0\,.\eeq 
It is easy to see that
 \begin{equation}\pi_{1,\dots,m} \ta_{i_1}\cdots\ta_{i_k}=\begin{cases}\theta_1 \cdots \theta_m 
&{\rm if~}\{i_1,\dots,i_k \}=\{1,\dots,m \}\, ,\\0&{\rm if~}\{i_1,\dots,i_k \}\neq \{1,\dots,m \}.
\end{cases}\end{equation}

 If we let $D^*_n$ (resp. $D^\circledast_n$) 
be the coefficient of $u^n$ in $D^*(u;q,t)$ (resp.  $D^\circledast(u;q,t)$), the operators $D^*(u;q,t)$ and $D^\circledast(u;q,t)$ can be seen as the generating series
of the operators $D^*_n$ and $D^\circledast_n$ respectively.
These operators, when restricted to act on $\mathscr R_N$,
 can be considered as the generalization to superspace of 
the Macdonald operators. 
 
As will be shown below,  the superpolynomials $ {  P}_{\Lambda} $ are eigenfunctions of both $D^*(u;q,t)$ and $D^\cd(u;q,t)$.  A rationale for the
rather intricate 
structure of the operators is the following:  in order for $  P_\La$, as given by \eqref{PvsE}, to be an eigenfunction of an operator built out of $\prod(1+uY_i)$,  the factor ${\Delta^t_{m}}/{\Delta_{m}}$ needs to be inserted
to the right  to transform $A_m$ into $U^-$ via \eqref{UvsA}, 
so that $U^-$ can be commuted through the factors $\prod(1+uY_i)$.   Finally, the term  ${\Delta_{m}}/{\Delta^t_{m}}$ is added to the left to retransform 
$U^-$ into $A_m$. We now state the eigenfunction characterization and plunge into the details of the proof.

\begin{proposition} \label{LemmaSeki}
We have 
\begin{multline}
D^*(u;q,t) \, {  P}_{\Lambda} = \varepsilon_{\Lambda^*} (u;q,t) \, 
{  P}_{\Lambda}
\\ \quad \qquad {\rm and} \quad D^\circledast(u;q,t) \, 
{  P}_{\Lambda} = \varepsilon_{\Lambda^\circledast} (u;q,t) \, {  P}_{\Lambda} \, ,
\end{multline}
where $\varepsilon_{\lambda}(u;q,t)$ is given, for
any partition $\lambda$, by
\begin{equation}\label{vapu}
\varepsilon_{\lambda}(u;q,t) = \prod_{i=1}^N \bigl(1+u \, q^{\lambda_i}
t^{1-i} \bigr)\,.
\end{equation}
\end{proposition}
\begin{proof}
First observe that 
\begin{equation}
\pi_{1,\dots,l}\, {\mathcal K}_w \theta_1 \cdots \theta_m =
\left\{
\begin{array}{ll}
 {\mathcal K}_w \theta_1 \cdots \theta_m & {\rm if~} l=m {\rm ~and~} w 
\in S_m \times S_{m^c}\, , \\
0 & {\rm otherwise} .
\end{array}
\right.
\end{equation}
We thus have
\begin{multline}\label{SSpart1}
 D^*(u;q,t)  \,  {  P}_{\Lambda}
\\ \qquad = c_{\Lambda}(t)
\sum_{\sigma \in S_N/(S_m\times S_{m^c})} {\mathcal K}_{\sigma} 
\frac{\Delta_{m}}{\Delta^t_{m}}\prod_{i=1}^N(1+uY_i)
\frac{\Delta^t_{m}}{\Delta_{m}}
\theta_1\cdots \theta_m A_{m}U^+_{m^c}
E_{\Lambda^R}  \\
\qquad \qquad = c_{\Lambda}(t)
\sum_{\sigma \in S_N/(S_m\times S_{m^c})} {\mathcal K}_{\sigma} 
t^{\binom{m}{2}} \frac{\Delta_{m}}{\Delta^t_{m}} \prod_{i=1}^N(1+uY_i)
\theta_1\cdots \theta_m U^-_{m}U^+_{m^c}
E_{\Lambda^R}  .
\end{multline}
Relations \eqref{TYi} imply that the product $\prod_{i=1}^N(1+uY_i)$ can be moved beyond the factors $U$:
\begin{multline}
D^*(u;q,t)  \,  {  P}_{\Lambda}
\\ \qquad  \qquad = c_{\Lambda}(t)
\sum_{\sigma \in S_N/(S_m\times S_{m^c})} {\mathcal K}_{\sigma} t^{\binom{m}{2}}
\frac{\Delta_{m}}{\Delta^t_{m}}
\theta_1\cdots \theta_m U^-_{m}U^+_{m^c}
\prod_{i=1}^N(1+uY_i) E_{\Lambda^R}  \\ \qquad = c_{\Lambda}(t)
\sum_{\sigma \in S_N/(S_m\times S_{m^c})} {\mathcal K}_{\sigma} 
\theta_1\cdots \theta_m A_{m}U^+_{m^c}
\prod_{i=1}^N(1+uY_i) E_{\Lambda^R} \, .
 \end{multline}
To prove the first statement, it thus suffices to prove that
\begin{equation} \label{eqseki1}
\left(\prod_{i=1}^N (1+u Y_i) \right)
E_{\Lambda^R}
 = \varepsilon_{\Lambda^*}(u;q,t) \, 
E_{\Lambda^R}.
\end{equation}
Similarly, to prove the second statement, we have to prove that
\begin{equation} \label{eqseki2}
\left(\prod_{i=1}^m (1+uq Y_i)  \prod_{i=m+1}^N (1+u Y_i) \right)
E_{\Lambda^R}
 = \varepsilon_{\Lambda^\circledast}(u;q,t) \, 
E_{\Lambda^R}.
\end{equation}


Let $\eta= \La^R$ and suppose that $\eta_i=r$.
It is easy to see that the quantity $ \bar{l}_\eta(i)$ 
in the eigenvalue $\bar \eta_i  = q^r t^{- \bar{l}_\eta(i)}$ 
of $Y_i$ (see \eqref{eigenvalY}) is such that
\begin{multline}
\bar{l}_\eta(i)= \#\{{\rm rows~of~} \Lambda^* 
{\rm~of~size~larger~than~}r\} 
 \\ + \#\{{\rm rows~of~}  \La^R {\rm~of~size~}r {\rm ~above~row~}i\} \, .
\end{multline}
Therefore, letting 
\begin{multline} \label{eqji}
j_i=\#\{{\rm rows~of~} \Lambda^* {\rm~of~size~larger~than~}r\} 
\\ + \#\{{\rm rows~of~} \La^R {\rm~of~size~}r {\rm ~above~row~}i\} +1
\end{multline}
we have $\{j_1,\dots,j_N \}=\{1,\dots,N\}$, $\Lambda_{j_i}^*=r$ 
and we recover \eqref{eqseki1} in the form
\begin{equation}  
\left( \prod_{i=1}^N(1+uY_i) \right) E_{\Lambda^R} = \prod_{j_i=1}^N(1+u q^{\Lambda^*_{j_i}} t^{1-j_i})E_{\Lambda^R} \, .
\end{equation}

Continuing with the same notation for the second case \eqref{eqseki2},
 we have that if
$i$ belongs to $\{1,\dots,m\}$ then
$\eta_i=r$ is the highest row of size $r$ in $\eta$,
and thus by \eqref{eqji} $\Lambda_{j_i}^*$ is also the highest row of size
$r$ in $\Lambda^*$.
Hence, the eigenvalue in this case is
\begin{equation}
q Y_i E_{\Lambda^R} = q^{\Lambda^*_{j_i}+1}t^{1-j_i} E_{\Lambda^R}.
\end{equation}
But since $i\in\{1,\ldots,m\}$, $\Lambda_{j_i}^*\in\Lambda^a$ and therefore $\Lambda^*_{j_i}+1= \Lambda_{j_i}^\circledast$.    
Now,
if $i$ does not belong to $\{1,\dots,m\}$, then we have 
$\Lambda_{j_i}^*=\Lambda_{j_i}^\circledast$ and then
\begin{equation}
\left(\prod_{i=1}^m (1+uq Y_i)  \prod_{i=m+1}^N (1+u Y_i) \right)E_{\Lambda^R} = \prod_{j_i=1}^N(1+uq^{\Lambda_{j_i}^\circledast}t^{1-j_i})E_{\Lambda^R}.
\end{equation}
\end{proof}

\subsection{The triangular decomposition of the Macdonald superpolynomials}\label{tri}
At this point, we have established that the $P_\La$'s are eigenfunctions of the operators $D^*(u;q,t)$ and $D^\cd(u;q,t)$. We now show that this characterization of the Macdonald superpolynomials entails a triangular decomposition into monomials. The chosen normalization in \eqref{PvsE} will make this decomposition unitriangular.

It is well known that if $x^\gamma$ appears in the monomial expansion of $Y_ix^\eta$, then $\gamma\leq \eta$ \cite[Eq.\ (4.13)]{Mac1}.  This statement however is not sufficient for our purposes.  We now give a slightly more precise characterization of the triangular action of the Cherednik operators $Y_i$ on 
monomials. For this, we need to define some operations on compositions.
Suppose $i<j$. Let
\begin{equation} 
s_{i,j} (\dots, \eta_i,\dots,\eta_j,\dots) =  (\dots, \eta_j,\dots,\eta_i,\dots). 
\end{equation} 
Let also $\hat s_{i,j}$ be the following restriction of $s_{i,j}$:  
\begin{equation} 
\hat s_{i,j} (\dots, \eta_i,\dots,\eta_j,\dots) =  (\dots, \eta_j,\dots,\eta_i,\dots) \quad\text{only if}\quad \eta_i>\eta_j\, .
\end{equation}  Finally, for $\ell=1,\dots, |\eta_i -\eta_j|-1$, let
\begin{equation} 
s_{i,j}^{(\ell)} (\dots, \eta_i,\dots,\eta_j,\dots) = 
\left \{
\begin{array}{ll}
 (\dots, \eta_i-\ell,\dots,\eta_j+\ell,\dots) & {\rm if~} \eta_i> \eta_j \, ,\\
 (\dots, \eta_i+\ell,\dots,\eta_j-\ell,\dots) & {\rm if~} \eta_i < 
\eta_j \, .
\end{array} \right.
\end{equation} 
Note that we will often use $s_i$, $\hat s_i$,  and $s_i^{(\ell)}$  instead of 
$s_{i,i+1}$, $\hat{s}_{i,i+1}$  and $s_{i,i+1}^{(\ell)}$, respectively.

\begin{lemma}  \label{lemmaY}
Suppose that $x^{\gamma}$ occurs in the monomial
expansion of $Y_i x^\eta$.  Then
\begin{equation} 
\gamma = g_{1} \cdots g_r (\eta)
\end{equation} 
where $g_k$ either stands for $\hat s_{i_k j_k}$ or $s_{i_k j_k}^{(\ell_k)}$.
We stress that if $g_k=\hat s_{i_k j_k}$ then $\omega_{i_k} > \omega_{j_k}$, where
$\omega = g_{k+1} \cdots g_r(\eta)$.
\end{lemma}
\begin{proof}
One easily  shows that
\begin{equation}\label{actionTi} T_ix^\eta=a_\eta x^\eta+b_\eta x^{s_i\eta}+ (1-t)\mathrm{sgn}(\eta_i-\eta_{i+1})\sum_{\ell=1}^{|\eta_i-\eta_{i+1}|-1}x^{s_i^{(\ell)}\eta}\,,
\end{equation}
where sgn$(x)$ is the sign of $x$ and
\begin{equation}a_\eta=\begin{cases}t-1 ,& \eta_i<\eta_{i+1} \, , \\ 0,& \eta_i\geq \eta_{i+1}\, ,
\end{cases} \qquad \text{and}\qquad 
b_\eta=\begin{cases}t ,& \eta_i\leq \eta_{i+1}\, , \\ 1,& \eta_i>\eta_{i+1}\, .
\end{cases}
\end{equation}
Thus, the monomials that appear in
$T_i x^{\eta}$ are of the form $x^{s_{i} (\eta)}$, $x^{s_{i}^{(\ell)}(\eta)}$, for $\ell=1,\dots,|\eta_i-\eta_{i+1}|-1$ 
and,  possibly an extra $x^{\eta}$.    The latter extra term $x^{\eta}$ appears only if 
$\eta_i <\eta_{i+1}$.    We stress that if $\eta_i=\eta_{i+1}$, the operator $s_i$ also acts as the identity on $\eta$.  However, even in the latter case, we prefer to write $s_i\eta$ explicitly in order to avoid any confusion with the identity operator that produces the very first term on the rhs of \eqref{actionTi}.  

The action of the inverse operator $\bar T_i=t^{-1} T_i+(t^{-1}-1)$ immediately follows from  \eqref{actionTi}.
Once more, the monomials that appear in
$ \bar T_i x^{\eta}$ are of the form $x^{s_{i} (\eta)}$, $x^{s_{i}^{(\ell)}(\eta)}$, for $\ell=1,\dots,|\eta_i-\eta_{i+1}|-1$ 
and,  possibly an extra $x^{\eta}$.    However, the latter extra term $x^{\eta}$ appears this time only if 
$\eta_i > \eta_{i+1}$. 

Hence, from the definition of $Y_i$, the terms $x^\gamma$
appearing in $Y_i x^{\eta}$ 
are such that
\begin{equation} \label{eqseq}
\gamma =f_i \cdots f_{N-1} s_{N-1} \cdots s_{1} \bar f_1 \cdots \bar f_{i-1} (\eta) \, ,
\end{equation}
where $f_j$ and $\bar f_j$ correspond either to $s_{j}$, $s_{j}^{(\ell)}$ or the identity (whenever allowed).  

Let $j$ be such that $\bar f_{j}$ does not act as $s_j$,
and suppose that $j$ is the rightmost amongst such terms.
We have that (note that $j<i$)
\begin{multline} 
\gamma =f_i \cdots f_{N-1} s_{N-1} \cdots s_{1} \bar f_1\\ \cdots \bar 
f_{j-1} s_{j} \cdots s_{i-1} ( s_{i-1} \cdots s_{j+1}
 s_{j} \bar f_j
s_{j+1} \cdots s_{i-1} (\eta) )\, .
\end{multline}
Observe that $f_i \cdots f_{N-1} s_{N-1} \cdots s_{1} \bar f_1\cdots \bar 
f_{j-1} s_{j} \cdots s_{i-1}$ corresponds to the operator that appears in
\eqref{eqseq} with $\bar f_k = s_k$ for $k=j,\dots,i-1$.
If $\bar f_j$ is the identity, we have
\begin{multline} 
s_{i-1} \cdots s_{j+1}
 s_{j} \bar f_j
s_{j+1} \cdots s_{i-1} (\eta)
\\=
s_{i-1} \cdots s_{j+1}
 s_{j} s_{j+1} \cdots s_{i-1} (\eta) = s_{j,i} (\eta) \, .
\end{multline}
Note that by supposition, $\eta_j > \eta_i$ since $\bar f_j$ acted as the identity.  If $\bar f_j = s_{j}^{(\ell)}$, we have
\begin{multline} 
s_{i-1} \cdots s_{j+1}
 s_{j} \bar f_j
s_{j+1} \cdots s_{i-1} (\eta)
\\=
s_{i-1} \cdots s_{j+1} 
 s_{j}  s_{j}^{(\ell)}  s_{j+1} \cdots s_{i-1} (\eta) = s_{j,i}^{(\ell')} (\eta)\, , 
\end{multline}
for some $\ell'$ (the $\ell'$ is such that  $s_{j,i}s_{j,i}^{(\ell)}=
s_{j,i}^{(\ell')}$).
Repeating the same process, we
can get rid of all $\bar f_k$ such that $\bar f_k$ does not act as 
$s_k$.   More explicitly,  we have shown that
\begin{equation} \label{gamma2}
\gamma =f_i \cdots f_{N-1} s_{N-1} \cdots s_{1}\bar f_1\cdots \bar f_{i-1}(\mu) \, ,
\end{equation}
where all the $\bar f_k$ act as $s_k$ and where $\mu$ is the composition obtained from $\eta$ by the action of some  $s_{ji}^{\ell}$  and  $\hat s_{ji}$ with $j<i$.   

Now suppose that $f_j$ acts as the identity and suppose that
$j$ is the rightmost such terms.  In this case,
(supposing that all $\bar f_k$ act as $s_k)$ \eqref{gamma2} becomes 
(note that $i<j$ this time)
\begin{multline} 
\gamma =f_i \cdots f_{j-1} f_j s_{j+1} \cdots s_{N-1} s_{N-1} \cdots s_1 s_1 \cdots s_{i-1} (\mu)\\
= f_i \cdots f_{j-1} f_j s_{j} \cdots s_{i} (\mu)\, .
\end{multline}
Hence,
\begin{equation} 
\gamma = f_i \cdots f_{j-1} s_j s_{j} \cdots s_{i} 
(s_{i} s_{i+1} \cdots s_{j-1} f_j s_{j} s_{j-1}\cdots s_{i+1} s_i (\mu) )\, .
\end{equation}
If $f_j$ is the identity, we have
\begin{multline}
(s_{i} s_{i+1} \cdots s_{j-1} f_j s_j s_{j-1} \cdots s_{i+1} s_i (\mu) )\\
=
(s_{i} s_{i+1} \cdots s_{j-1} s_j s_{j-1} \cdots s_{i+1} s_i (\mu) )=s_{i,j}(\mu)\, .
\end{multline}
By supposition, $\mu_i > \mu_j$ since $f_j$ acted as the identity, so $s_{i,j}(\mu)=\hat s_{i,j}(\mu)$.
 If $f_j=s_j^{(\ell)}$, then for some $\ell'$,
\begin{multline} 
\gamma = f_i \cdots f_{j-1} s_j s_{j} \cdots s_{i} 
\big(s_{i} s_{i+1} \\
\cdots s_{j-1} s_j^{(\ell)} s_{j} s_{j-1}\cdots s_{i+1} s_i (\mu) \big)
=s_{i,j}^{(\ell')}(\mu)\, .
\end{multline}

Applying these operations again and
again we obtain that
\begin{equation}
 \gamma =f_i \cdots f_{N-1} s_{N-1} \cdots s_{1} \bar f_1 \cdots \bar f_{i-1} (g_1\cdots g_r (\eta))\, ,
\end{equation}
where the $g_k$'s are such as specified in the statement of the lemma, and
where all the $f_k$'s and $\bar f_k$'s act as $s_k$. But replacing
$f_k=\bar f_k=s_k$ in the previous equation, we obtain that
\begin{equation}
\gamma = g_1\cdots g_r (\eta) \, ,
\end{equation}
which completes the proof of the lemma.
\end{proof}

 \begin{proposition}\label{proptriangS} We have
\begin{equation}
D^*(u;q,t) \, m_{\Lambda} = \varepsilon_{\Lambda^*} (u;q,t) \, {  m}_{\Lambda}
+\sum_{\Omega < \Lambda} v_{\Lambda \Omega} \, m_{\Omega}
\end{equation}
and, similarly,
\begin{equation}
D^\circledast(u;q,t) \, m_{\Lambda} = \varepsilon_{\Lambda^\circledast} (u;q,t) \, {  m}_{\Lambda} +
\sum_{\Omega < \Lambda} \bar v_{\Lambda \Omega} \, m_{\Omega}\, .
\end{equation}
\end{proposition}
\begin{proof}
Let $m$ be a fixed integer.  Define $\eta^{\circledast}$ to be equal
to $(\eta+1^m)^+$, and define $\eta^*$ to be equal to $\eta^+$. 
We say that a row $i$ of $\eta^\circledast$ is fermionic if
$\eta^\circledast_i \neq \eta^*_i$ and bosonic otherwise.

We will first show that if $x^\gamma$ appears
in $Y_i x^\eta$, then $\gamma^* \leq \eta^*$ and $\gamma^\circledast \leq \eta^\circledast$.  It is immediate from the definition of the Bruhat order on compositions
and the fact that $Y_i$ acts triangularly on monomials
that $\gamma^* \leq \eta^*$.  We have left to show that 
$\gamma^\circledast \leq \eta^\circledast$.  From Lemma~\ref{lemmaY}, we only need
to show that if $\gamma=g_k (\eta)$, for $g_k$ such as specified in the 
statement of the lemma,
then $\gamma^\circledast \leq \eta^\circledast$.  Suppose that
 $g_k=s_{i,j}$.  From Lemma~\ref{lemmaY},  we have $\eta_i > \eta_j$.
The only non-trivial case is when $i \leq m$ and $j>m$.
In that case, it is easy to see that $\gamma^{\circledast}$ will be obtained
from $\eta^{\circledast}$ by interchanging a fermionic 
element and a bosonic 
one, with the fermionic one being the largest.  This is easily seen to imply
that $\gamma^\circledast \leq \eta^\circledast$.  Now suppose that
$g_k=s_{i,j}^{(\ell)}$.  Again the only non-trivial case is when 
$i \leq m$ and $j>m$.  In that case, 
given that $\ell<(\eta_i-\eta_j)$,  
$\gamma^{\circledast}$ is obtained
from $\eta^{\circledast}$ by modifying a fermionic element and a bosonic one
in such a way that none of the two modified rows is larger than the largest of the original ones.  It is then immediate that again $\gamma^\circledast \leq \eta^\circledast$. 

Suppose now that $\Lambda$ is a superpartition in the fermionic sector $m$
(we will consider that $\Lambda$ is also the composition 
$(\Lambda_1^a,\dots,\Lambda_{m}^a,\Lambda_{1}^s,\dots)$, so that 
the monomial $m_\La$ can be written as 
\begin{equation}
m_\La=
\frac{1}{f_{\La^s}(1)} \sum_{\sigma \in S_N/(S_m \times S_{m^c})} \mathcal K_\sigma \left( \theta_1 \cdots \theta_m A_m \sum_{w\in S_{m_c}} x^{w(\La)} \right).
\end{equation}
Then,
\begin{multline} \label{bigthing}
 D^*(u;q,t)  \,  m_{\Lambda}
  \\
\qquad \propto
\sum_{\sigma \in S_N/(S_m\times S_{m^c})} {\mathcal K}_{\sigma} 
\frac{\Delta_{m}}{\Delta^t_{m}}\prod_{i=1}^N(1+uY_i)
\frac{\Delta^t_{m}}{\Delta_{m}}
\theta_1\cdots \theta_m A_{m} \sum_{w \in S_{m^c}}
x^{w(\Lambda)}   \\
 \qquad  =
\sum_{\sigma \in S_N/(S_m\times S_{m^c})} {\mathcal K}_{\sigma}  t^{\binom{m}{2}}
\frac{\Delta_{m}}{\Delta^t_{m}}\prod_{i=1}^N(1+uY_i)
\theta_1\cdots \theta_m U^-_{m}
\sum_{w \in S_{m^c}}
x^{w(\Lambda)}   \\
  \propto
\sum_{\sigma \in S_N} {\mathcal K}_{\sigma} 
\theta_1\cdots \theta_m A_{m}\sum_{w \in S_{m^c}}
\prod_{i=1}^N(1+uY_i) 
x^{w(\Lambda)}\, .  
 \end{multline}
Letting $\eta=w(\Lambda)$ we have that $\eta^*=\Lambda^*$ and
$\eta^\circledast=\Lambda^\circledast$.  From our prior analysis,
we thus have that all the terms $x^\gamma$
that appear in $\prod_{i=1}^N(1+uY_i) 
x^{w(\Lambda)}$ are such that $\gamma^* \leq \Lambda^*$ and
$\gamma^\circledast\leq\Lambda^\circledast$.  The triangularity
of the action of $D^*(u;q,t)$ is then immediate.
The triangularity of the action of
$D^\circledast(u;q,t)$ is proven in the same manner.
Finally, proceeding as in the proof of  Proposition~\ref{LemmaSeki},
it can be checked that 
$$
\prod_{i=1}^N(1+uY_i)  \, x^\Lambda = \varepsilon_{\Lambda^*} (u;q,t) \, {  x}^{\Lambda}+ {\rm lower~terms}
$$
and similarly, that
$$
\prod_{i=1}^m(1+uqY_i) \prod_{i=m+1}^N(1+uY_i)
  \, x^\Lambda = \varepsilon_{\Lambda^\circledast} (u;q,t) \, {  x}^{\Lambda}+ {\rm lower~terms}\, .
$$
This completes the proof of the proposition.
\end{proof}

\begin{proposition} \label{propuni} The Macdonald superpolynomials are unitriangularly related to the monomials.  In other words,
\begin{equation}
{  P}_{\Lambda} = m_{\Lambda} + \sum_{\Omega < \Lambda } c_{\Lambda \Omega}(q,t) 
\,m_{\Om} \, ,
\end{equation}
where we observe that $c_{\Lambda \Omega}(q,t)$ does not depend on $N$ from
Proposition~\ref{propostable}.
\end{proposition}
\begin{proof}  The triangularity is almost immediate
from Propositions~\ref{LemmaSeki} and \ref{proptriangS}.  
Suppose
that there exists a term $m_\Omega$ such that $\Omega \not \leq 
\Lambda$ in ${  P}_{\Lambda}$ and suppose that $\Omega$ is maximal among
those superpartitions.  Then by Proposition~\ref{proptriangS}
the coefficient of $m_{\Omega}$ in
$D^*(u;q,t) {  P}_{\Lambda}$ and $D^\circledast(u;q,t) {  P}_{\Lambda}$ is respectively
equal to $c_{\Lambda \Omega} \varepsilon_{\Omega^*} (u;q,t) $
and $c_{\Lambda \Omega}\varepsilon_{\Omega^\circledast} (u;q,t)$.  Since
we cannot have $\varepsilon_{\Omega^*} (u;q,t) = \varepsilon_{\Lambda^*} (u;q,t)$ and $\varepsilon_{\Omega^\circledast} (u;q,t) = 
\varepsilon_{\Lambda^\circledast} (u;q,t)$ 
at the same time ($\varepsilon_{\Lambda^*} (u;q,t)$
and $\varepsilon_{\Lambda^\circledast} (u;q,t)$ uniquely determine $\Lambda$),
we have the contradiction that ${  P}_{\Lambda}$ is not an eigenfuntion
of $D^*(u;q,t)$ and $D^\circledast(u;q,t)$ with
eigenvalues $\varepsilon_{\Lambda^*} (u;q,t)$
and $\varepsilon_{\Lambda^\circledast} (u;q,t)$ respectively.

We now have to prove that the coefficient of $m_{\Lambda}$ in 
${  P}_{\Lambda}$ is equal to 1. To prove this, we follow \cite[{Lemma 5.5}]{Mar}.   
We start with the expression \eqref{PvsE} for $  P_\La$, written compactly as
\begin{equation}
  P_\La = c_\La \sum_{\sigma \in S_N/(S_m \times S_{m^c})} \mathcal K_\sigma \theta_1 \cdots \theta_m A_m U^+_{m^c} E_{\La^R}\, ,
\end{equation}
where the constant $c_\La$ given in \eqref{PvsE} (or in \eqref{clambda}).
It suffices to concentrate on the coefficient of the term in $x^{\La}$.   
 From \eqref{defEtrian}, we see that it can only arise from the dominant term
\begin{equation}
  P_\La = c_\La \sum_{\sigma \in S_N/(S_m \times S_{m^c})} \mathcal K_\sigma \left( \theta_1 \cdots \theta_m A_m \sum_{w\in S_{m^c}}T_w\, x^{\La^R} \right) + \text{lower terms}\, ,
\end{equation}
where it should be observed that
$T_w$ acts on  $x_{m+1}^{\La_N} x_{m+2}^{\La_{N-1}} \cdots x_N^{\La_{m+1}}$.  Let \begin{equation} \underline{\eta}=(\La_1, \ldots, \La_m, \La_{N}, \ldots , \La_{m+1})\,,
\end{equation} and write
\begin{equation}
A_m x^{\La^R}=(-1)^{\text{inv}(\La^a)} A_m x^{\underline{\eta}}\, ,
\end{equation}
where $\text{inv}(\La^a)=m(m-1)/2$.  We thus have
\begin{equation*}
  P_\La = (-1)^{\binom{m}{2}} c_\La \sum_{\sigma \in S_N/(S_m \times S_{m^c})} \mathcal K_\sigma \left( \theta_1 \cdots \theta_m A_m \sum_{w\in S_{m^c}}T_w x^{\underline{\eta}} \right) + \text{lower terms}\, .
\end{equation*}
It is easy to show (from the explicit action of $T_i$) that 
{\cite[Lemma 2.3]{Mar} }
\begin{equation}
T_w x^{\underline{\eta}}= d_{w(\underline{\eta})}\, x^{w(\underline{\eta})} + \sum_{\nu \prec w(\underline{\eta})}d_\nu \,x^\nu,
\end{equation}
where $d_\nu \in \mathbb Q(t)$. In the following, we denote by $[x^\La]F(x)$ the coefficient of $x^\La$ in the expression $F(x)$. The coefficient of the  term $x^\La$ in $U^+_{m^c}\,x^{\underline{\eta}}$ is given by
\begin{equation} \label{fixcoeff}
[x^{\La}] \sum_{w \in S_{m^c}} T_w \,x^{\underline{\eta}} = \sum_{w \in S_{m^c} | w(\underline{\eta})=\La} t^{\ell(w)}\, ,
\end{equation}
where $\ell(w)$ is the length of the permutation $w$.  Suppose {that all the parts of $\La^s$ are distinct. Then there is } 
only one permutation $w$ that can give $w(({\La^s})^R)=\La^s$ and its length is given by $\ell(w)= \text{inv}(\La^s)$.  
Now, when there are repeated parts in $\La^s$, $\text{inv}(\La^s)$ is the length
of the permutation of minimal length such that $w(({\La^s})^R)=\La^s$.
{However,  we must also consider the contributions resulting} from permuting these repeated parts.  So, in general we can write
\begin{equation}
[x^{\La}] \sum_{w \in S_{m^c}} T_w x^{\underline{\eta}} = t^{\text{inv}(\La^s)} \prod_{i=0}^{\La^s_1} \left( \sum_{ w^{(i)} \in S_{m(i)} } t^{\ell(w^{(i)})} \right)\, ,
\end{equation}
where $m(i)=n_{\Lambda^s}(i)$.
Using $\sum_{\sigma \in S_k}t^{\ell(\sigma)}=[k]_t!$, we then obtain
\begin{equation}
[x^{\La}] \sum_{w \in S_{m^c}} T_w x^{\underline{\eta}} = t^{\text{inv}(\La^s)} \prod_{i=0}^{\La^s_1} [n_{\La^s}(i)]_t! \; .
\end{equation} 
Now, since $U^+_{m^c}x^{\underline{\eta}}$ is symmetric in the variables $x_{m+1}, \ldots, x_N$ and $A_m U^+_{m^c}x^{\underline{\eta}}$ is antisymmetric in the variables $x_{1}, \ldots, x_m$, the monomial $m_\La$ is reconstructed with these actions and multiplication by $\ta_1\cdots\ta_m$.  Hence, we have
\begin{equation}
[m_\La] \quad \sum_{\sigma\in S_N/(S_m \times S_{m^c})}\mathcal K_\sigma  \left( \theta_{1}\cdots \theta_m A_mU^+_{m^c}x^{\underline{\eta}} \right)= t^{\text{inv}(\La^s)} \prod_{i=0}^{\La^s_1} [n_{\La^s}(i)]_t!\; ,
\end{equation}
which immediately gives
\begin{multline}
[m_\La]  \;   P_\La = (-1)^{\binom{m}{2}} \,c_\La \, t^{\text{inv}(\La^s)} \prod_{i=0}^{\La^s_1} [n_{\La^s}(i)]_t! \\
= (-1)^{\binom{m}{2}} \,c_\La \,t^{\text{inv}(\La^s)} f_{\La^s}(t) =1.
\end{multline}
\end{proof}


\begin{corollary}  The ${  P}_\Lambda$'s form a basis of the space
 $\mathscr R_N$
of symmetric polynomials in superspace. 
\end{corollary}

 Recall that $D^*_n$ (resp.  $D^\circledast_n$) 
is the coefficient of $u^n$ in $D^*(u;q,t)$ (resp.  $D^\circledast(u;q,t)$).
\begin{corollary} \label{corollaryD}
The $2N$ operators $D^*_n$ and $D^\circledast_\ell$ for $n,\ell=1,\dots,N$
are mutually commuting when their action is restricted to  $\mathscr R_{N}$.
\end{corollary}

\subsection{A simplified eigenfunction characterization}\label{simp}

We end this section with a characterization of the $P_{\Lambda}$'s as common 
eigenfuntions
of two commuting operators.  In the notation of Corollary~\ref{corollaryD}, we have
\begin{align}
D_1^*&=\sum_{m=0}^N
\sum_{\sigma \in S_N/(S_m\times S_{m^c})}\mathcal{K}_\sigma\left(\frac{\Delta_m}{\Delta^t_m}\left(Y_1 + \cdots +Y_N \right)
\frac{\Delta^t_{m}}{\Delta_{m}}\pi_{1,\ldots,m}\right), \\
D_1^\cd&=\sum_{m=0}^N
\sum_{\sigma \in S_N/(S_m\times S_{m^c})}\mathcal{K}_\sigma\left(\frac{\Delta_{m}}{\Delta^t_{m}}\left(
qY_1 + \cdots +qY_m \right. \right. \\ & \qquad \qquad \qquad\qquad    \qquad \qquad  \left. \left.+ Y_{m+1} + \cdots +Y_N \right)
\frac{\Delta^t_{m}}{\Delta_{m}}\pi_{1,\ldots,m}\right). \nonumber
\end{align} 
From the linear term in $u$ in \eqref{vapu}, we see that the eigenvalue of the above operators on $P_\La$ are
\begin{equation}\label{eigD}
D_1^\circ P_\La=\varepsilon^{(1)}_{\La^\circ}P_\La,\quad \text{where} \quad \circ\in\{*,\cd\},\quad 
\text{and}\quad \varepsilon^{(1)}_{\la}=\sum_iq^{\la_i}t^{1-i}.\end{equation}
 Given these two operators, it is natural to consider {the following} differences:
\begin{multline}
\mathcal O_1 =\frac{1}{1-q}(D_1^*-D_1^\cd) \\ 
= \sum_{m=0}^N
\sum_{\sigma \in S_N/(S_m\times S_{m^c})}\mathcal{K}_\sigma\left(\frac{\Delta_{m}}{\Delta^t_{m}}\left(
Y_1 + \cdots +Y_m  \right)
\frac{\Delta^t_{m}}{\Delta_{m}}\pi_{1,\ldots,m}\right),
\end{multline}
\begin{multline}
\mathcal O_2 = \frac{1}{q-1}(qD_1^*-D_1^\cd)\\
=\sum_{m=0}^N
 \sum_{\sigma \in S_N/(S_m\times S_{m^c})}\mathcal{K}_\sigma\left(\frac{\Delta_{m}}{\Delta^t_{m}}\left(
Y_{m+1} + \cdots +Y_N  \right)
\frac{\Delta^t_{m}}{\Delta_{m}}\pi_{1,\ldots,m}\right).
\end{multline}
 From \eqref{eigD}, we get
\begin{equation}
\mathcal O_1 P_{\Lambda} = \frac{1}{q-1}\left[\sum_{i=1}^N 
(q^{\Lambda_i^{\circledast}}-q^{\Lambda_i^*})t^{1-i} \right] P_{\Lambda}
\end{equation}
and
\begin{equation}
\mathcal O_2 P_{\Lambda} = \frac{1}{q-1}\left[\sum_{i=1}^N 
(q^{\Lambda_i^*+1}-q^{\Lambda_i^\circledast})t^{1-i} \right] P_{\Lambda}\, .
\end{equation}
Observe that the two eigenvalues are in one-to-one 
correspondence with $\Lambda$. 
We also define
\begin{equation}
\bar {\mathcal O}_1 = \sum_{m=0}^N\sum_{\sigma \in S_N/(S_m\times S_{m^c})}\mathcal{K}_\sigma\left(\frac{\Delta_{m}}{\Delta^t_{m}}\left(
\bar Y_1 + \cdots +\bar Y_m  \right)
\frac{\Delta^t_{m}}{\Delta_{m}}\pi_{1,\ldots,m}\right)\, ,
\end{equation}
where $\bar Y_i$ is the inverse of $Y_i$:
 \begin{equation}\label{invC}
 \bar Y_i=t^{N-i}T_{i-1}\cdots T_1\bar \omega \bar T_{N-1}\cdots \bar T_i,  \end{equation} 
 with $ \bar T_j=T_j^{-1}$ and  $ \bar \omega = \omega^{-1}$.
We have 
\begin{equation}
\bar{\mathcal O}_1 P_{\Lambda} = \frac{1}{1/q-1}\left[\sum_{i=1}^N 
(q^{-\Lambda_i^{\circledast}}-q^{-\Lambda_i^*})t^{i-1} \right] P_{\Lambda}.
\end{equation}

Finally, we define $E_{1,N} = \bar {\mathcal O}_1$ and
$
E_{2,N} =  {\mathcal O}_2 -\sum_{i=1}^N t^{1-i}$. 
It is easy to see that 
\begin{multline}\label{vape}
E_{1,N} {  P}_{\Lambda} = \left(\sum_{i \, : \, \Lambda^\circledast_i \neq 
\Lambda^*_i} q^{-\Lambda^*_i} t^{i-1}\right) {  P}_{\Lambda} \\\quad {\rm and} \quad
E_{2,N} {  P}_{\Lambda} = \left(\sum_{i \, : \, \Lambda^\circledast_i = 
\Lambda^*_i} (q^{\Lambda^*_i}-1) t^{1-i}\right) {  P}_{\Lambda}\, .
\end{multline}
Thus,  the eigenvalues of $E_{1,N}$ and $E_{1,N}$ do not depend on $N$. 
This property explains the substraction of $\sum_{i=1}^N t^{1-i}$  in the definition of $E_{2,N}$: it ensures that the eigenvalue does not depend upon the zeros
in $\La^s$.
 This, and the fact that 
$ {  P}_{\Lambda}$ is stable with respect to the number of variables,
allows us to define
\begin{equation}
E_1 = \lim_{\longleftarrow} E_{1,N} \qquad {\rm and} \qquad
E_2 = \lim_{\longleftarrow} E_{2,N}\, .
\end{equation}
We have the following characterization of ${  P}_{\Lambda}$.
\begin{proposition} The polynomial ${  P}_{\Lambda}$
is the unique polynomial in $\mathscr R$ such that
\begin{enumerate}
\item 
${  P}_{\Lambda} = m_{\Lambda} + \sum_{\Om <\La} v_{\La \Om} \, m_{\Om}$ \\
\item
$\displaystyle E_{1} {  P}_{\Lambda} = \Big(\! \sum_{\substack{i \\ \Lambda^\circledast_i \neq 
\Lambda^*_i}} q^{-\Lambda^*_i} t^{i-1}\Big) {  P}_{\Lambda} \quad {\rm and} \quad
E_{2} {  P}_{\Lambda} = \Big(\!\sum_{\substack{i \\ \Lambda^\circledast_i = 
\Lambda^*_i}} (q^{\Lambda^*_i}-1) t^{1-i}\Big) {  P}_{\Lambda}$.
\end{enumerate}
\end{proposition}
\begin{proof}
 From Proposition \ref{propuni} and \eqref{vape}, the superpolynomial ${  P}_{\Lambda}$ satisfies the two properties.  Since $E_1$ and $E_2$ 
have together distinct eigenvalues,
the two properties characterize ${  P}_\Lambda$.
\end{proof}

\section{Orthogonality and existence}\label{ortho}
 
 \subsection{Kernel and orthogonality}\label{ker}
Let $x=(x_1,x_2,\dots)$ and $y=(y_1,y_2,\dots)$ be two sets of commuting
variables, and let $\theta=(\theta_1,\theta_2,\dots)$ and $\phi=(\phi_1,\phi_2,\dots)$ be two sets of anticommuting variables. 
 We define the following reproducing kernel \cite{BDLM}:
\begin{equation}
K(x,\ta;y,\phi)=\prod_{i,j}\frac{\left(tx_iy_j;q\right)_\infty}{\left(x_iy_j;q\right)_\infty} 
\left(1 {+}\frac{\theta_i\phi_j}{1-q^{-1}x_iy_j}\right)
\end{equation}
where
\begin{equation}
(a;q)_\infty = \prod_{k=0}^{\infty} (1-a q^k) \, .
\end{equation}
Observe that
\begin{equation}
K(x,\ta;y,\phi)
=K^0(x;y)\prod_{i,j}\left(1 {+}\frac{\theta_i\phi_j}{1-q^{-1}x_iy_j}\right) \, ,
\end{equation}
where
\begin{equation}
K^0(x;y)= \prod_{i,j} \frac{\left(tx_iy_j;q\right)_\infty}{\left(x_iy_j;q\right)_\infty}
\end{equation}
is the usual Macdonald reproducing kernel \cite[eq. VI.2.4]{Mac}.  It is straightforward to show that
\begin{equation}\label{Ksp}
K(x,\ta;y,\phi)=\sum_\La  (-1)^{\binom{m}2}{z_\La(q,t)}^{-1} \, p_\La(x,\ta)\,p_\La(y,\phi)\, ,\end{equation}
where $z_\La(q,t)$ was defined in \eqref{newsp}.  
Recall from Theorem \ref{theo1} that the factor $(-1)^{\binom{m}2}{z_\La(q,t)} $ is the norm of the  scalar product of the power sums, i.e.,
\begin{equation}\label{newsP}
\LL {p_\La}|{p_\Om}\RR_{q,t}=(-1)^{\binom{m}2}\,z_\La(q,t)\delta_{\La\Omega}\, .
\end{equation}

The following propositions are standard and can be proven using methods
similar to those found in 
Macdonald's book \cite{Mac}.
\begin{proposition} For each $n,m$, let 
$\{ u_{\Lambda} \}$ and $\{v_{\Lambda} \}$
be bases of $\mathscr R^{n,m}$, where  $\mathscr R^{n,m}$
is the subspace of $\mathscr R$ of degree $(n|m)$.
Then the following conditions
are equivalent:
\begin{enumerate}
\item $\LL {u_\La}|{v_\Om}\RR_{q,t}=\delta_{\Lambda \Om}$ for all $\La,\Om$ ;
\item $K(x,\theta;y,\phi)=\sum_{\La} u_{\La}(x,\theta) v_{\La} (y,\phi)$.
\end{enumerate}
\end{proposition}

\begin{proposition} \label{propoortho}
 Let $E: \mathscr R \to \mathscr R$ be a $\mathbb Q(q,t)$-linear operator.  Then the following conditions are equivalent:
\begin{enumerate}
\item $\LL E f\, | \, g\RR_{q,t}= \LL  f\, |  E g\RR_{q,t}  $ for all $f,g \in \mathscr R$; 
\item $ E^{(x,\theta)} K(x,\theta;y,\phi)=  E^{(y,\phi)} K(x,\theta;y,\phi)$, where $E^{(x,\theta)}$
(resp.  $E^{(y,\phi)}$) acts on the variables $x$ and $\theta$ (resp.
$y$ and $\phi$). 
\end{enumerate}
\end{proposition}
The rest of this section will be devoted to the proof of the following proposition, whose corollary implies Theorem~\ref{theo1}.
\begin{theorem} \label{propoadj}
We have that\begin{equation}
\begin{split}
 E_1^{(x,\theta)} K(x,\theta;y,\phi)&=  E_1^{(y,\phi)} K(x,\theta;y,\phi)\, ,\\
E_2^{(x,\theta)} K(x,\theta;y,\phi) &=  E_2^{(y,\phi)} K(x,\theta;y,\phi)\, . 
\end{split}\end{equation}
\end{theorem}

\begin{corollary}\label{corotringortho} The Macdonald superpolynomial $P_{\La}$
is such that
\begin{enumerate}
\item $ P_{\Lambda} =
m_{\Lambda} + \text{lower terms}$;
\item $\LL  P_{\La}| P_{\Om} \RR_{q,t} =0\quad\text{if}\quad \La\ne\Om$.
\end{enumerate}
\end{corollary}
\begin{proof}
The triangularity was proven in Proposition~\ref{propuni}.  
By Proposition~\ref{propoortho} and Theorem~\ref{propoadj} we have
\begin{equation*}
\LL E_1 {  P}_{\La} \, | \, {  P}_{\Om} 
\RR_{q,t} =
\LL  {  P}_{\La} \, | \, E_1 {  P}_{\Om} 
\RR_{q,t} \quad {\rm and} \quad 
\LL E_2 {  P}_{\La} \, | \, {  P}_{\Om} 
\RR_{q,t} =
\LL  {  P}_{\La} \, | \, E_2 {  P}_{\Om} 
\RR_{q,t}.
\end{equation*}
Given that the two operators have together distinct eigenvalues,
this immediately gives
\begin{equation}
\LL  {  P}_{\La} \, | \, {  P}_{\Om} 
\RR_{q,t} = 0 \quad {\rm if} \quad \Lambda \neq \Omega \, .
\end{equation}
\end{proof}

\subsection{Self-adjointness of $E_1$ and $E_2$}\label{self}
The proof of Theorem~\ref{propoadj} is quite involved.  
It relies
fundamentally on Proposition~\ref{quasip}, whose  proof 
was kindly outlined to us by Alain Lascoux \cite{Laspv}.  
   Theorem~\ref{propoadj} follows from
the following proposition since $\bar{\mathcal O}_1=E_{1,N}$ and
${\mathcal O}_2$ differs from $E_{2,N}$ by a constant.
 \begin{proposition} \label{theoseki}
Let $K_N(x,\theta;y, \phi)$ be the restriction of  $K(x,\theta;y, \phi)$
to $N$ variables.  We have
\begin{equation}\label{lesO}\begin{split}
\bar{\mathcal O}_1^{(x,\theta)}
 K_N(x,\theta;y,\phi)&= 
\bar{\mathcal O}_1^{(y,\phi)} K_N(x,\theta;y,\phi)\, ,\\  
{\mathcal O}_2^{(x,\theta)} K_N(x,\theta;y,\phi)&= 
{\mathcal O}_2^{(y,\phi)} K_N(x,\theta;y,\phi)\, ,\end{split}
\end{equation}where the superscripts indicate the variables on which the operators act.
 \end{proposition}

The first step in the proof of Proposition~\ref{theoseki}
will amount to reformulate the conditions \eqref{lesO} in a more tractable form (this is Proposition  \ref{quasip}  below).

Firstly, it is not difficult to see that the coefficient of 
$
\theta_1 \cdots \theta_m \phi_{1} \cdots \phi_{m} 
$ 
in
$K_N(x,\theta;y,\phi)$ 
 is equal to
(up to a sign and a power of  $q$)
\begin{equation}\label{idena}
[\theta_1 \cdots \theta_m \phi_1 \cdots \phi_{m}] \; \; K_N(x,\theta;y,\phi)
\propto 
\frac{K^0_N(x;y)\Delta_m(x) \Delta_m(y)}{\prod_{1\leq i,j \leq m}(1-q^{-1}x_i y_j)},
\end{equation}
where we recall that $[\theta_1 \cdots \theta_m \phi_{1} \cdots \phi_{m}]\,f$ 
stands for the coefficient of the monomial  ${\theta_1 \cdots \theta_m \phi_{1} \cdots \phi_{m}}$ in $f$.
This is seen as follows: up to a sign, the coefficient of $\theta_1 \cdots \theta_m \phi_{1} \cdots \phi_{m}$
in $\prod_{i,j} \left[ 1+{\theta_i \phi_j}({1- x_i y_j})^{-1} \right]$ is 
\begin{multline} \label{eqpasrap}
\sum_{\sigma \in S_m } (-1)^{\ell(\sigma)}
\frac{1}{\prod_{i=1}^m (1-x_i y_{\sigma(i)})} = \sum_{\sigma} 
 (-1)^{\ell(\sigma)}\sum_{\eta} x^{\eta} y^{\sigma(\eta)} \\ =
\sum_{\eta} (-1)^{{\rm sign} (\eta)} x^{\eta} s_{\eta^+-\delta^{(m)}}(y) \Delta(y)\, ,
\end{multline}
where $\eta$ is a composition with distinct parts (otherwise the result is zero by antisymmetry), $\eta^+$ is the partition corresponding to 
$\eta$ and ${\rm sign} (\eta)$ is the sign of the permutation that
changes $\eta$ to $\eta^+$.  The second equality is obtained by interchanging the two summations and using the expression of the Schur polynomial $s_\mu$ as a ratio of two determinants (cf. \cite[eq. (3.1)]{Mac}). By splitting the sum over $\eta$ into a sum over $\eta^+$ and a summation over permutations of $\eta^+$,  and then by letting $\lambda=\eta^+-\delta^{(m)}$,  we can rewrite  \eqref{eqpasrap} as 
\begin{multline}
\sum_{\sigma \in S_m } (-1)^{\ell(\sigma)}
\frac{1}{\prod_{i=1}^m (1-x_i y_{\sigma(i)})} 
\\= \Delta_m(x) \Delta_m(y)
\sum_{\lambda} s_{\lambda}(x) s_{\lambda}(y) = \frac{\Delta_m(x) \Delta_m(y)}
{\prod_{1\leq i,j \leq m}(1-x_i y_j)}\, .
\end{multline}
Note that \cite[eq. (4.3)]{Mac} was used for getting the last expression on the right-hand side.
By substituting $x_i\rw x_i/q$ and multiplying the result by $K^0_N(x;y)$, we recover \eqref{idena}.

More generally, the coefficient of 
$
\theta_1 \cdots \theta_m\phi_{j_1} \cdots \phi_{j_m}
$ (with $j_1<\cdots<j_m$)
in
$K_N(x,\theta;y,\phi)$ 
 is 
\begin{multline}\label{idena2}
{[\theta_1 \cdots \theta_m \phi_{j_1} \cdots \phi_{j_m}] \; \; K_N(x,\theta;y,\phi)}
\\
\propto 
K^{(y)}_{1, j_1} \cdots K^{(y)}_{m, j_m} 
\frac{{K^0_N(x;y)}\Delta_m(x) \Delta_m(y)}{\prod_{1\leq i,j \leq m}(1-q^{-1}x_i y_j)}.
\end{multline}
Given this result, 
 the coefficient of $\theta_{i_1} \cdots \theta_{i_m} \phi_{j_1} \cdots \phi_{j_m}$ in $\bar {\mathcal O}_1^{(x,\theta)}K_N(x,\theta;y,\phi)$ is 
 proportional to {\small
\begin{equation}\label{O1x}
K^{(x)}_{1, i_1} \cdots K^{(x)}_{m, i_m} K^{(y)}_{1, j_1} \cdots K^{(y)}_{m, j_m}
\frac{\Delta_m(x)}{\Delta_m^t(x)} (\bar Y_1^{(x)} +\cdots+\bar Y_m^{(x)})
\frac{K^0_N(x;y)\Delta_m^t(x) \Delta_m(y)}{\prod_{1\leq i,j \leq m}(1-q^{-1}x_i y_j)}.
\end{equation}}\hspace{-.1cm}
Similarly, the coefficient of $\theta_{i_1} \cdots \theta_{i_m} \phi_{j_1} \cdots \phi_{j_m}$ in $\bar {\mathcal O}_1^{(y,\phi)}K_N(x,\theta;y,\phi)$ is proportional to (with the same proportionality factor as above){\small
\begin{equation}\label{O1y} 
K^{(x)}_{1, i_1} \cdots K^{(x)}_{m, i_m} K^{(y)}_{1, j_1} \cdots K^{(y)}_{m, j_m}
\frac{\Delta_m(y)}{\Delta_m^t(y)} (\bar Y_1^{(y)} +\cdots+\bar Y_m^{(y)})
\frac{K^0_N(x;y)\Delta_m(x) \Delta_m^t(y)}{\prod_{1\leq i,j \leq m}(1-q^{-1}x_i y_j)}.
\end{equation}}\hspace{-.1cm}
The first relation in \eqref{lesO} is equivalent to the equality of both coefficients. Canceling the permutation operators in the equality between \eqref{O1x} and \eqref{O1y} yields
\begin{multline}
\frac{\Delta_m(x)}{\Delta_m^t(x)} (\bar Y_1^{(x)} +\cdots+\bar Y_m^{(x)})
\frac{K^0_N(x;y)\Delta_m^t(x) \Delta_m(y)}{\prod_{1\leq i,j \leq m}(1-q^{-1}x_i y_j)}
\\= \frac{\Delta_m(y)}{\Delta_m^t(y)} (\bar Y_1^{(y)} +\cdots+\bar Y_m^{(y)})
\frac{K^0_N(x;y)\Delta_m(x) \Delta_m^t(y)}{\prod_{1\leq i,j \leq m}(1-q^{-1}x_i y_j)}.
\end{multline}
Proceeding similarly for $\mathcal O_2$,
and using commutativity of the type $ Y_i^{(x)}f(y)=f(y)Y_i^{(x)}$, we obtain the two relations appearing in the following proposition,
whose proof thus implies Proposition~\ref{theoseki}.

\begin{proposition} \label{quasip} 
We have { 
\begin{multline}\label{pro1}
(\bar Y_1^{(x)}+\cdots+\bar Y_m^{(x)}) \frac{{K_N^0(x;y)} \Delta_m^t(x) \Delta_m^t (y)}{\prod_{1\leq i, j \leq m} (1-q^{-1}x_i y_j )} \\ = (\bar Y_1^{(y)}+\cdots+\bar Y_m^{(y)}) \frac{{K_N^0(x;y)} \Delta_m^t(x) \Delta_m^t(y)}{\prod_{1\leq i, j \leq m} (1-q^{-1}x_i y_j )}
\end{multline}
\begin{multline}\label{pro2}
(Y_{m+1}^{(x)}+\cdots+Y_N^{(x)}) \frac{K_N^0(x;y) \Delta_m^t(x) \Delta_m^t (y)}{\prod_{1\leq i, j \leq m} (1-q^{-1}x_i y_j )} \\ = (Y_{m+1}^{(y)}+\cdots+Y_N^{(y)}) \frac{K_N^0(x;y) \Delta_m^t(x) \Delta_m^t(y)}{\prod_{1\leq i, j \leq m} (1-q^{-1}x_i y_j )} 
\end{multline}}
for all $m=0,\dots,N$, where $Y_i^{(x)}$ and  $Y_i^{(y)}$ act respectively on the $x$ and $y$ variables. 
\end{proposition}

The proof of the proposition relies crucially on Lemma~\ref{lemLas}, which is formulated in terms of divided differences \cite[Sec. 7.1]{Las}
\begin{equation}
\partial_i = \frac{1}{(x_i-x_{i+1})}(1-{K_{i,i+1}}).
\end{equation}
The divided differences obey the  braid relations  \cite[Sec. 7.3]{Las}, 
\begin{equation}\label{relbraid}
\partial_{i+1} \partial_i \partial_{i+1}=
\partial_{i} \partial_{i+1} \partial_{i}\, ,\qquad \partial_i \partial_j=
\partial_j \partial_i\quad\text{if} \quad |i-j|>1\, ,
\end{equation} and the nilpotent condition $\partial_i^2=0$ \cite[eq. 7.1.5]{Las}.
As is the case for the $T_i$'s, if $s_{i_1} \cdots s_{i_\ell}$ is a reduced
decomposition of $w$ then $\partial_w$ stands for 
$\partial_{i_1} \cdots \partial_{i_\ell}$.

The following lemma was stated by Alain Lascoux \cite{Laspv}.  We provide
our own proof of it.
\begin{lemma}   \label{lemLas}
Let $\partial_{\omega_{ \mathbf m}}$ be the divided difference associated to the longest permutations of $S_m$ \cite[Sect. 7.6]{Las}.   
Then
the following identity holds:
\begin{equation} \label{iden1}
(-t)^{-\binom{m}{2}} \frac{\Delta_m^t(z)}{\prod_{1\leq i, j \leq m} (1-z_i y_j )}=
\partial^{(y)}_{\omega_{ \mathbf m}}   \left(\frac{\prod_{i+j \leq m} (1-tz_i y_j)}{\prod_{i+j \leq m+1} (1-z_i y_j)} \right),
\end{equation}
where the superscript indicates that the divided differences act on the
$y$ variables.
\end{lemma}
\begin{proof} We first observe that $\prod_{1\leq i,j\leq m} (1-z_i y_j)$, being fully symmetric in the $y_j$'s,  vanishes when acted on by $\partial^{(y)}_{\omega_{ \mathbf m}}   $. Therefore, multiplying both sides of the 
identity \eqref{iden1} by $\prod_{1\leq i,j\leq m} (1-z_i y_j)$, we can then use Leibniz identity \cite[eq. (7.1.10)]{Las} to commute the product with the divided differences. We then note that the product can be factorized as follows: 
\begin{equation}
\prod_{1\leq i,j\leq m} (1-z_i y_j)=  \prod_{i+j\leq m+1} (1-z_i y_j) \; \prod_{
\substack{
i+j> m+1 \\
i,j\leq m}
} (1-z_i y_j) .
\end{equation}
Hence, the identity \eqref{iden1} is equivalent to 
\begin{multline} \label{iden2}
(-t)^{-\binom{m}{2}} \Delta_m^t(z)
=
\partial^{(y)}_{\omega_{ \mathbf m}}   \Big(
\prod_{i+j \leq m} (1-tz_i y_j)
\prod_{
\substack{
i+j> m+1 \\
i,j\leq m}
} (1-z_i y_j) \Big) \\ = \partial^{(y)}_{\omega_{ \mathbf m}}    Q(z,y).
\end{multline}
It is well known \cite[Sect. 7.6]{Las}
that 
\begin{equation}
\partial_{\omega_{ \mathbf m}}^{(y)} \left(y_1^{a_1} \cdots y_m^{a_m} \right) = 0 
\end{equation}
unless all the $a_k$'s are distinct.
Observe that after expanding 
the products on the rhs of \eqref{iden2}, the appearing monomials 
$y_1^{a_1} \cdots y_m^{a_m}$ will all be such that $a_1,\ldots,a_m \leq m-1$.
  Indeed, the power of $y_j$ equals the number of distinct factors $z_i$ that can appear in the coefficient of $y_j^{a_j}$ and the maximal value of this number is, with $j$ fixed,
\begin{equation}
\#\{ i\,|\, i+j\leq m\} + \#\{i\,|\, i+j>m+1\}=m-1.
\end{equation} 
The only option to have distinct $a_k$'s is thus for
$(a_1,\dots,a_m)$ to be a rearrangement of $(m-1,m-2,\dots,1,0)$.
Therefore,
the rhs of \eqref{iden2} is a polynomial in $y_1,\dots,y_m$ of 
degree 0, that is, the result does not depend on the variables $y$.
We now show that it is $t$-antisymmetric in the $z$ variables,
i.e., that
\begin{equation}
T_k^{(z)} \partial_{\omega_{ \mathbf m}}^{(y)} Q(z,y) = -\partial_{\omega_{ \mathbf m}}^{(y)} Q(z,y) 
\end{equation}
for all $k$. 
 We have that $\omega_{ \mathbf m}= w s_{m-k}$ for some permutation $w \in S_m$.  Hence,
 it suffices to prove
that 
\begin{equation} \label{eqsym}
T_k^{(z)} \partial_{m-k}^{(y)} Q(z,y) = - \partial_{m-k}^{(y)} Q(z,y) .
\end{equation}
 It is easy to see that $Q(z,y)$ is symmetric in both 
$y_{m-k},y_{m-k+1}$ and $z_{k},z_{k+1}$
except
for the factors $(1-tz_{k} y_{m-k})(1-z_{k+1}y_{m-k+1})$. 
A direct calculation yields
\begin{equation}
 \partial_{m-k}^{(y)} (1-tz_{k} y_{m-k})(1-z_{k+1}y_{m-k+1}) = - (tz_k -z_{k+1}),
\end{equation}
from which \eqref{eqsym} follows immediately since $T_k (tz_k-z_{k+1})=-  (tz_k-z_{k+1})$.
Finally, the rhs of \eqref{iden2} is a polynomial in $z$ of total
degree $m(m-1)/2$.  Since the only such $t$-antisymmetric polynomial
is $\Delta_m^t(z)$, \eqref{iden2} holds up to a constant.  
The coefficients of $z_1^{m-1} z_2^{m-2}\cdots z_{m-1}$ on the lhs of \eqref{iden2} 
is clearly $(-t)^{\binom{m}{2}}$.  On the rhs of \eqref{iden2}, the coefficient of
 $z_1^{m-1} z_2^{m-2}\cdots z_{m-1}$ is $\partial_{\omega_{ \mathbf m}}$ acting on a certain
polynomial $p(y)$ in $y$.  In $p(y)$, the only monomial
 whose exponent is a permutation of $(m-1,m-2,\dots,1,0)$ is
$y_1^{m-1}y_2^{m-2} \cdots y_{m-1}$.  Given that its coefficient is  $(-t)^{\binom{m}{2}}$,
the result follows.
\end{proof}

\n We now turn to the proof of Proposition \ref{quasip}. 
\begin{proof}
We will assume throughout the proof that $N>m$.  The case $N=m$
can be easily obtained as a simplified version the case $N>m$.
  Note that in the following arguments we will
never worry about constants depending on $q$ and $t$ (which are irrelevant to the symmetry).  For instance,
we write $U^{-\,(x)}_{m} \Delta_m^t(x) \propto  \Delta_m^t(x)$, meaning
that the two expressions only differ by a constant. (Recall that $ U_m^-=\sum_{\sigma\in S_m}(-t)^{-\ell(\sigma)}T_\sigma.$)

We first prove equation \eqref{pro1}. This amounts to show
that
\begin{equation}
F(x,y) = (\bar Y_1^{(x)}+\cdots+\bar Y_m^{(x)}) \frac{K^0_N(x;y) \Delta_m^t(x) \Delta_m^t (y)}{\prod_{1\leq i, j \leq m} (1-q^{-1}x_i y_j )}
\end{equation}
is symmetric in $x$ and $y$, that is, $F(x,y)=F(y,x)$.
Since $ \Delta_m^t(x) \propto U^{-\,(x)}_{m} \Delta_m^t(x)$ and because
$U^{-\,(x)}_{m}$  and $\Delta_m^t(y) $ commute with the $\bar Y_i^{(x)}$'s, we have
 \begin{equation}\label{lequationdeF}
F(x,y)
\propto
\Delta_m^t(y) U^{-\,(x)}_{m} 
(\bar Y_1^{(x)}+\cdots+\bar Y_m^{(x)})   \frac{\Delta_m^t(x)
K^0_N(x;y)}{\prod_{1\leq i, j \leq m} (1-q^{-1}x_i y_j )}.
\end{equation}
Recall the expression for the inverse Cherednik operator in \eqref{invC} and
that for $ \bar T_j$ in \eqref{Tinv}. 
Since $\bar T_i^{(x)}  \Delta_m^t(x) = - t^{-1} \Delta_m^t(x)$ and 
 $U^{-\,(x)}_{m} T_i^{(x)}=-U^{-\,(x)}_{m}$ whenever $i<m$ we have
that, up to an irrelevant $t$-power,  $U^{-\,(x)}_m \bar Y_i^{(x)}  \Delta_m^t(x)$
can be replaced in \eqref{lequationdeF}
 by $U_m^{-\, (x)} \omega^{(x)} \bar T^{(x)}_{N-1}\cdots \bar T^{(x)}_m \Delta_m^t(x)$.  Hence
  \begin{equation}
F(x,y) \propto  \Delta_m^t(y) U^{-\,(x)}_{m} \bar \omega^{(x)} \bar T^{(x)}_{N-1} \cdots \bar T^{(x)}_{m}  \frac{\Delta_m^t(x)
K^0_N(x;y)}{\prod_{1\leq i, j \leq m} (1-q^{-1}x_i y_j )}.
\end{equation}
Using (see \cite[Sect. 7.6]{Las}) $\partial_{\omega_{ \mathbf m}}^{(x)} = [\Delta_m(x)]^{-1} A_{m}$ and \eqref{UvsA}, we have
\begin{equation}\label{idenL}
U^{-\,(x)}_{m} \propto \Delta_m^t(x) \partial^{(x)}_{\omega_{ \mathbf m}} \, ,
\end{equation}
which gives
\begin{equation*}
F(x,y)  \propto  \Delta_m^t(y) \Delta_m^t(x)
\partial^{(x)}_{\omega_{ \mathbf m}} \bar \omega^{(x)} \bar T^{(x)}_{N-1} \cdots \bar T^{(x)}_{m}  \frac{\Delta_m^t(x)
K^0_N(x;y)}{\prod_{1\leq i, j \leq m} (1-q^{-1}x_i y_j )}.
\end{equation*}
Using the identity \eqref{iden1} with $x_i=q z_i$, we obtain
\begin{multline}
F(x,y)  \propto  \Delta_m^t(y) \Delta_m^t(x)\\ \times
\partial^{(y)}_{\omega_{ \mathbf m}} \partial^{(x)}_{\omega_{ \mathbf m}} \bar \omega^{(x)} \bar T^{(x)}_{N-1} \cdots \bar T^{(x)}_{m}  \frac{\prod_{i+j \leq m} (1-tq^{-1}x_i y_j)}{\prod_{i+j \leq m+1} (1-q^{-1}x_i y_j)} K^0_N(x;y).
\end{multline}
To prove the symmetry of $F(x,y)$, it thus suffices to prove the symmetry
of
\begin{equation}
G(x,y)=  
\bar \omega^{(x)} \bar T^{(x)}_{N-1} \cdots \bar T^{(x)}_{m}  \left(\frac{\prod_{i+j \leq m} (1-tq^{-1}x_i y_j)}{\prod_{i+j \leq m+1} (1-q^{-1}x_i y_j)} \right)  K^0_N(x;y).
\end{equation}
The only part of the term in parenthesis that depends on the variables $x_m,\dots,x_N$
is $(1-q^{-1}x_my_1)^{-1}$.  We have
\begin{equation}
\bar T^{(x)}_{N-1} \cdots \bar T^{(x)}_{m} \frac{1}{(1-q^{-1}x_m y_1)} =  t^{-N+m} \frac{\prod_{i=m}^{N-1} (1-tq^{-1}x_i y_1)}{\prod_{i=m}^N (1-q^{-1}x_i y_1)},
\end{equation}
which implies that
\begin{multline}
G(x,y)\propto  
\bar \omega^{(x)}  \left(\frac{\prod_{i+j \leq m} (1-tq^{-1}x_i y_j)}{\prod_{i+j \leq m+1} (1-q^{-1}x_i y_j)} \right)\\ \times \left(
 \frac{\prod_{i=m}^{N-1} (1-tq^{-1}x_i y_1)}{\prod_{i=m+1}^N (1-q^{-1}x_i y_1)} \right)
 K^0_N(x;y).
\end{multline}
We then straightforwardly compute
\begin{equation}
\bar \omega^{(x)} K^0_N(x;y) = K^0_N(x;y)  \left[
 \frac{\prod_{i=1}^{N} (1-tq^{-1}x_1 y_i)}{\prod_{i=1}^N (1-q^{-1}x_1 y_i)} \right],
\end{equation}
\begin{multline} 
\bar \omega^{(x)}  \prod_{i+j \leq m} (1-tq^{-1}x_i y_j)
=  \prod_{i+j \leq m} (1-tq^{-1}x_i y_j) \\ \times
\left(\prod_{j=2}^{m-1}(1-tq^{-1} x_{m+1-j} y_j)\right) 
\left[
\frac{(1-tq^{-1}x_m y_1)} {\prod_{j=1}^{m-1}(1-tq^{-1} x_{1} y_j)} \right],
\end{multline}
\begin{multline} 
\bar \omega^{(x)}  \prod_{i+j \leq m+1} (1-q^{-1}x_i y_j)= \prod_{i+j \leq m+1} (1-q^{-1}x_i y_j)
\\ \times \left(
\prod_{j=2}^{m}(1-q^{-1} x_{m+2-j} y_j)\right)  \left[\frac{(1-q^{-1}x_{m+1} y_1)}{\prod_{j=1}^{m}(1-q^{-1} x_{1} y_j)} \right],
\end{multline}
and
\begin{multline}
\bar \omega^{(x)} \left(
 \frac{\prod_{i=m}^{N-1} (1-tq^{-1}x_i y_1)}{\prod_{i=m+1}^N (1-q^{-1}x_i y_1)} \right) \\ 
= \left[\frac{\prod_{i=m+1}^{N} (1-tq^{-1}x_i y_1)}{\prod_{i=m+2}^{N} (1-q^{-1}x_i y_1)} \right] \frac{1}{(1-q^{-2}x_1y_1)} .
\end{multline}
All of the terms in the rhs of these expressions are symmetric 
in $x,y$ except those in square brackets. Multiplying the terms in
square brackets, the symmetry of $G(x,y)$ then depends on the symmetry of
\begin{multline}
 \left(
 \frac{\prod_{i=1}^{N} (1-tq^{-1}x_1 y_i)}{\prod_{i=1}^N (1-q^{-1}x_1 y_i)} \right)
\left(\frac{(1-tq^{-1}x_m y_1)}{\prod_{j=1}^{m-1}(1-tq^{-1} x_{1} y_j)}\right) \\ \times
\left(
\frac{\prod_{j=1}^{m}(1-q^{-1} x_{1} y_j)}{(1-q^{-1}x_{m+1} y_1)}\right) 
\left(\frac{\prod_{i=m+1}^{N} (1-tq^{-1}x_i y_1)}{\prod_{i=m+2}^{N} (1-q^{-1}x_i y_1)} \right). 
\end{multline}
But the previous expression is equal to
\begin{equation}
 \left(
 \frac{\prod_{i=m}^{N} (1-tq^{-1}x_1 y_i)}{\prod_{i=m+1}^N (1-q^{-1}x_1 y_i)} \right)
\left(\frac{\prod_{i=m}^{N} (1-tq^{-1}x_i y_1)}{\prod_{i=m+1}^{N} (1-q^{-1}x_i y_1)} \right) ,
\end{equation}
which is obviously  symmetric in $x,y$.  This proves the symmetry of $G(x,y)$ and therefore equality \eqref{pro1} holds.

We now prove eq. \eqref{pro2},  proceeding as in the proof of \eqref{pro1}. 
This amounts to proving
that
\begin{equation}
H(x,y) = (Y_{m+1}^{(x)}+\cdots+ Y_N^{(x)}) \frac{K^0_N(x;y) \Delta_m^t(x) \Delta_m^t (y)}{\prod_{1\leq i, j \leq m} (1-q^{-1}x_i y_j )}
\end{equation}
satisfies $H(x,y)=H(y,x)$.
We have
\begin{multline}
H(x,y)
\propto
\Delta_m^t(y) \\ \times U^{-\,(x)}_{m} U^{+\,(x)}_{m^c} 
(Y_{m+1}^{(x)}+\cdots+ Y_N^{(x)})   
  \frac{ \Delta_m^t(x)
K^0_N(x;y)}{\prod_{1\leq i, j \leq m} (1-q^{-1}x_i y_j )}.
\end{multline}
Since $U^{+\,(x)}_{m^c} T_i^{(x)}=  tU^{+\,(x)}_{m^c} $ 
whenever $i \geq m+1$
(and 
similarly for $\bar T_i^{(x)}$ on any function symmetric in $x_i,x_{i+1}$), 
we have
\begin{multline}
H(x,y) \propto  \Delta_m^t(y)   \\ \times U^{-\,(x)}_{m} 
U^{+\,(x)}_{m^c} \omega^{(x)} \bar T_{1}^{(x)} \cdots \bar T_{m}^{(x)}  \frac{
K^0_N(x;y)}{\prod_{1\leq i, j \leq m} (1-q^{-1}x_i y_j )}.
\end{multline}
Using \eqref{idenL},
we have thus
\begin{multline}
H(x,y)  \propto  \Delta_m^t(y) \Delta_m^t(x) \\ \times
\partial^{(x)}_{\omega_{ \mathbf m}} U^{+\,(x)}_{m^c} \omega^{(x)} \bar T^{(x)}_{1} \cdots \bar T^{(x)}_{m}  \frac{\Delta_m^t(x)
K^0_N(x;y)}{\prod_{1\leq i, j \leq m} (1-q^{-1}x_i y_j )}.
\end{multline}
Using the identity \eqref{iden1} with $x_i=q z_i$, we obtain
\begin{multline}
H(x,y)  \propto  \Delta_m^t(y) \Delta_m^t(x) \\  \times 
\partial^{(y)}_{\omega_{ \mathbf m}} \partial^{(x)}_{\omega_{ \mathbf m}} U^{+\,(x)}_{m^c} \omega^{(x)} \bar T^{(x)}_{1} \cdots \bar T^{(x)}_{m}  \frac{\prod_{i+j \leq m} (1-tq^{-1}x_i y_j)}{\prod_{i+j \leq m+1} (1-q^{-1}x_i y_j)} K^0_N(x;y).
\end{multline}
To prove the symmetry of $H(x,y)$, it thus suffices to prove the symmetry
of
\begin{equation} \label{eqproduit}
L(x,y)= U^{+\,(x)}_{m^c}  
\omega^{(x)} \bar T^{(x)}_{1} \cdots \bar T^{(x)}_{m} 
 \left(\frac{\prod_{i+j \leq m} (1-tq^{-1}x_i y_j)}{\prod_{i+j \leq m+1} (1-q^{-1}x_i y_j)} \right)  K^0_N(x;y).
\end{equation}
The only part of the product that depends on the variables $x_m,x_{m+1}$ is
 $(1-q^{-1}x_my_1)^{-1}$.  We have
\begin{equation}
\bar T^{(x)}_{m} \frac{1}{(1-q^{-1}x_m y_1)} = \frac{t^{-1}}{(1-q^{-1}x_{m+1} y_1)}
 \frac{(1-tq^{-1}x_{m} y_1)}{(1-q^{-1}x_{m} y_1)} \, .
\end{equation}
Adding the previous expression to the product in the rhs of \eqref{eqproduit}, 
the only factor that is not symmetric
in $x_{m-1},x_m$ is $(1-q^{-1}x_{m-1}y_2)^{-1}$.  We have this time
\begin{equation}
\bar T^{(x)}_{m-1} \frac{1}{(1-q^{-1}x_{m-1} y_2)} = \frac{t^{-1}}{(1-q^{-1}x_{m} y_2)}
 \frac{(1-tq^{-1}x_{m-1} y_2)}{(1-q^{-1}x_{m-1} y_2)} \, .
\end{equation}
Continuing in this manner, we get
\begin{multline}
 \bar T^{(x)}_{1} \cdots \bar T^{(x)}_{m} 
 \left(\frac{\prod_{i+j \leq m} (1-tq^{-1}x_i y_j)}{\prod_{i+j \leq m+1} (1-q^{-1}x_i y_j)} \right) \\ 
=  
 {t^{-m}}\left(\frac{\prod_{i+j \leq m} (1-tq^{-1}x_i y_j)}{\prod_{i+j \leq m+1} (1-q^{-1}x_i y_j)} \right)   \left(\frac{\prod_{i=1}^m (1-tq^{-1}x_{m+1-i} y_i)}{\prod_{i=1}^m (1-q^{-1}x_{m+2-i} y_i)} \right) \, , 
\end{multline}
which implies that
\begin{equation*}
L(x,y)\propto 
 U^{+\,(x)}_{m^c} \omega^{(x)}  \left(\frac{\prod_{i+j \leq m+1} (1-tq^{-1}x_i y_j)}{\prod_{i+j \leq m+2} (1-q^{-1}x_i y_j)} \right) (1-q^{-1}x_1 y_{m+1})
 K^0_N(x;y).
\end{equation*}
The following actions of $\omega^{(x)}$ are straightforward:
\begin{equation}
\omega^{(x)} K^0_N(x;y) = K^0_N(x;y)  \left[
 \frac{\prod_{i=1}^{N} (1-x_N y_i)}{\prod_{i=1}^N (1-tx_N y_i)} \right],
\end{equation}
\begin{equation} 
 \omega^{(x)}  \prod_{i+j \leq m+1} (1-tq^{-1}x_i y_j)
=  \prod_{i+j \leq m} (1-tq^{-1}x_i y_j)
\left[\prod_{j=1}^{m}(1-t x_{N} y_j) \right],
\end{equation}
\begin{equation} 
 \omega^{(x)}  \prod_{i+j \leq m+2} (1-q^{-1}x_i y_j)
=  \prod_{i+j \leq m+1} (1-q^{-1}x_i y_j)
\left[\prod_{j=1}^{m+1}(1-x_{N} y_j) \right],
\end{equation}
and
\begin{equation}
 \omega^{(x)} (1-q^{-1}x_1 y_{m+1}) = \left[ (1-x_Ny_{m+1}) \right].
\end{equation}
The product of the terms in square brackets of these expressions gives
\begin{multline}
 \left[
 \frac{\prod_{i=1}^{N} (1-x_N y_i)}{\prod_{i=1}^N (1-tx_N y_i)} \right] 
\left[\frac{\prod_{j=1}^{m}(1-t x_{N} y_j)}{\prod_{j=1}^{m+1} (1-x_N y_j)} \right]
\left[ (1-x_Ny_{m+1}) \right] \\ =  \frac{\prod_{i=m+1}^{N} (1-x_N y_i)}{\prod_{i=m+1}^N (1-tx_N y_i)}.
\end{multline}
Hence
\begin{multline}
L(x,y) \propto 
 \frac{\prod_{i+j \leq m} (1-tq^{-1}x_i y_j)}{\prod_{i+j \leq m+1} (1-q^{-1}x_i y_j)}  \\ K^0_N(x;y)  U^{+\,(x)}_{m^c} 
\left(   \frac{\prod_{i=m+1}^{N} (1-x_N y_i)}{\prod_{i=m+1}^N (1-tx_N y_i)} \right).
\end{multline}
All the terms in this expression are symmetric in $x,y$, except possibly
\begin{equation}
 U^{+ \, (x)}_{m^c}
\left(   \frac{\prod_{i=m+1}^{N} (1-x_N y_i)}{\prod_{i=m+1}^N (1-tx_N y_i)} \right).
\end{equation}
But the symmetry of this last expression 
follows from the well known symmetry of the
$m=0$ case (the usual Macdonald case).  Therefore $L(x,y)$ is symmetric and
\eqref{pro2} holds.
\end{proof}

\subsection{Duality}\label{dual}

 We end this section by generalizing to superspace the standard duality property that relates the  Macdonald symmetric functions $P_\la(q,t)$ and $P_{\la'}(t,q)$ \cite[Section VI.5]{Mac}.  Our method relies on the orthogonality and triangularity of both the Macdonald superpolynomials $P_\La(q,t)$ and the Jack superpolynomials $P_\La^{(\alpha)}$, respectively established in Corollary \ref{corotringortho} above and Theorem 1 of \cite{DLMadv}.  We also exploit the duality between the Jack superpolynomials $P_\La^{(\alpha)}$ and $P_{\La'}^{(\alpha^{-1})}$ given in Theorem 27 of \cite{DLMadv}.   Note that in what follows, only the special case $\alpha=1$ is relevant.   

The algebra $\mathscr{R}$ of symmetric functions in superspace is naturally equipped with two homomorphisms, {the first of which being}  
\begin{equation}
\omega_{q,t} p_r=(-1)^{r-1}\frac{1-q^r}{1-t^r}\, p_r\qquad \omega_{q,t}\tilde p_r=(-1)^{r}\tilde p_r \, .
\end{equation}
This is an extension to superspace of the standard homomorphism {defined in } \cite[Section VI.2]{Mac}. Second, we introduce the following homomorphism that affects only the fermionic {power-sums}: 
\beq
\tilde \omega_{q} p_r= p_r\qquad \tilde \omega_{q} \tilde p_r=q^{r} \tilde p_r\, .
\eeq
Combining the two homomorphisms, we get 
\beq  \Omega_{q,t}:= \tilde{\omega}_q\circ\omega_{q,t}\,,
 \eeq
which is such that
\beq \label{actionomegaqt}
\Omega_{q,t}\,p_\La=\omega_\La q^{|\La^a|}\prod_{i=1}^{\ell(\La^s)}
\frac{1-q^{\La_i^s}}{1-t^{\La^s_i}}\,p_\La,\eeq
where 
\beq \omega_\La=(-1)^{|\Lambda|-\ell(\La^s)} 
.\eeq
{When} $q=t=1$,  $\Omega_{q,t}$ reduces to the homomorphism  $\omega$ of \cite{DLMjaco}, whose action can be summarized as follows:
\beq  \label{actionomega} \omega \,p_\Lambda = \omega_\Lambda \,p_\La .\eeq
   Equations \eqref{actionomegaqt} and \eqref{actionomega} immediately imply 
that, for all $f,g\in\mathscr{R}$ (and by linearity, it suffices to verify the case  where $f=p_\La$ and $g=p_\Om$), 
\beq\label{dualprodscal} \LL \, {\Omega^{-1}_{q,t} } f\,|\,g \,\RR_{q,t}=\LL\, \omega  f\,|\, g \,\RR,
\eeq 
where the scalar product $\LL\,|\,\RR$ on the right-hand side is defined as $\LL\,|\,\RR_{q,t}$  in \eqref{newsp}, but with $q=t=1$.  Note that both $\Omega_{q,t}$ and $\omega$  are in fact automorphisms of $\mathscr{R}$. Their respective inverse  are:
\beq \Omega_{q,t}^{-1}= \Omega_{t,q}\circ \tilde{\omega}_{(qt)^{-1}}\qquad\text{and}\qquad
\omega^{-1}=\omega.\eeq 
In each  homogeneous component $\mathscr{R}^{n,m}$ of $\mathscr{R}$, we also have
\beq \label{Omegainvformula} \Omega_{q,t}^{-1}=(t^n q^{-n})\Omega_{t^{-1},q^{-1}}.\eeq 

Now let $\mathtt{ s}_\La$ be the Schur superpolynomial associated to the Jack superpolynomial $P^{(\alpha)}_\Lambda$ \cite{DLMadv}, which means  $\mathtt{s}_\La:=P_\La^{(1)}$.  We stress that $\mathtt{ s}_\La$ is not equal to the Schur function $s_\La$ defined later in the article as the  $q,t\to 0$ limit of $P_\La(q,t)$  {(cf. eq. \eqref{schur})}.   Then according to Theorem 1 of \cite{DLMadv},  
\beq \label{orthoschur} \LL \mathtt{s}_\La|\mathtt{s}_\Omega \RR= \mathtt{ b}_\La^{-1}\delta_{\La,\Om}.
\eeq
Moreover, from \cite[ Theorem 27]{DLMadv}  { applied to $\a=1$ (and recall that $\La'$ is the conjugate of $\La$)},  
\beq \label{dualityschur}
  \omega\, \mathtt{s}_\La=\mathtt{s}^*_{\La'}= {\mathtt{ b}_\La^{-1} \mathtt{s}_{\La'}} \eeq
{The expression for} the normalization  constant $\mathtt{ b}_\Lambda^{{-1}}$ is  known, being equal to $(-1)^{\binom{m}{2}}\| P^{(1)}_\La\|^2$,  where $\| P^{(\alpha)}_\La\|^2$ is given by Eq.\ (18) of \cite{DLMeva}.
When specialized to $\alpha=1$, this normalization factor reads{
\beq \label{defb}
\mathtt{ b}_\Lambda=(-1)^{\binom{m}{2}}\prod_{s\in\La}\frac{a_{\La^\cd}+\ell_{\La^*}+1}
{a_{\La^*}+\ell_{\La^\cd}+1} .
\eeq
Since upon conjugation, the role of the arm-length and leg-length (defined in \eqref{eqarms}) are exchanged, it   satisfies} 
\beq\label{inversionofb} \mathtt{ b}_{\La'}^{\phantom{-1}}=\mathtt{ b}_\La^{-1}, \eeq 
so that we can write 
\beq \label{schurdual} \mathtt{s}_{\La}^*:=\mathtt{ b}_\La\, \mathtt{s}_\La
{\quad\text{\and}\qquad
\LL \mathtt{s}_\La\,|\,\mathtt{s}^*_\Omega \RR= \delta_{\La,\Om}.}
\eeq

\begin{theorem}\label{theoduality}  Let $Q_\Lambda=\mathtt{ b}_\Lambda(q,t) P_\Lambda(q,t)\,,$ where $\mathtt{ b}_\Lambda(q,t)=\LL P_\Lambda(q,t)\,|\,P_\Lambda(q,t)\RR_{q,t}^{-1}.$
Then, the following duality holds: 
\beq
\Omega_{q,t} P_\La(q,t)= (qt^{-1})^{|\La|} \, Q_{\La'}(t^{-1},q^{-1}).\eeq
\end{theorem}
\begin{proof}  We proceed essentially as in \cite[Section VI.5]{Mac}.  Note that in what follows, we assume that all polynomials are homogeneous and belong to $\mathscr{R}^{n,m}$, which is finite dimensional.

Thanks to the orthogonality of the Macdonald superpolynomials  {established} in Corollary \ref{corotringortho} and the inversion formula \eqref{Omegainvformula}, the stated duality property is equivalent to
\beq \LL \Omega^{-1}_{q,t } \,P_{\La'}(t^{-1},q^{-1})\, |\,P_\Om(q,t)\RR_{q,t}=\delta_{\La,\Om}.
\eeq
By virtue of \eqref{dualprodscal} and $\omega^{-1}=\omega$, the last equation turns out to be equal to 
\beq \label{dual2} \LL \omega P_{\La'}(t^{-1},q^{-1})\, |\,P_\Om(q,t)\RR_{}=\delta_{\La,\Om}. 
\eeq 
Now, let ${A}(q,t)$ be the transition matrix between the $P_\Lambda(q,t)$'s and the $\mathtt{s}_{\Om}$'s, that is 
\beq \label{dual3} P_\La(q,t)=\sum_{\Omega}A_{\La\Om}(q,t)\, \mathtt{s}_\Om.
\eeq 
Let also $J$ be the involutive matrix with elements $J_{\La\Om}=\delta_{\La',\Om}$.    Then \eqref{dual2} is equivalent to the following matrix equation:
\beq  \label{eqJA}JA(t^{-1},q^{-1}) J A(q,t)'=I,
\eeq
where we stress that $A(q,t)'$ is the transpose of the matrix $A(q,t)$. 
This is the equation we {will  prove. 
Before attacking  this problem directly, we need to derive a number of auxiliary results, to which we now turn.}

Let $p$, $\mathtt{s}$, and $\mathtt{s}^*$ be the column vectors with the $\La$-element equal to  $p_\La$, $\mathtt{s}_\La$, and $\mathtt{s}^*_\La$ respectively.  Let $\mathtt{ b}$, $\zeta$, and $\zeta(q,t)$ be diagonal matrices whose non-zero elements are respectively given by   $\mathtt{ b}_\La$ (defined in \eqref{defb}),  $ \zeta_\La=\LL p_\La\,|\,p_\La\RR$ and $ \zeta_\La(q,t)=\LL p_\La \,|\, p_\La\RR_{q,t}$, {so that
\beq \label{defzeta}
 \zeta_\La=(-1)^{\binom{m}{2}} z_{\La^s}\quad \text{and}\quad 
 \zeta_\La(q,t)=(-1)^{\binom{m}{2}}z_{\La^s} 
\, q^{|\La^a|} \prod_{i=1}^{\ell(\La^s)}
\frac{1-q^{\La^s_i}}{1-t^{\La^s_i}}.
 \eeq}
We also define $X$ as the matrix with entries $X_{\La\Om}=\LL p_\La\,|\, \mathtt{s}^*_\Om\RR$.  Then by making use of the orthogonality with respect to $\LL\,|\,\RR$ of both the $s_\La$'s and the $p_\La$'s, one readily obtains 
\beq \label{eqXps}
p=X\mathtt{s}, \qquad \mathtt{s}^*=X'\zeta^{-1}p.
\eeq
{Together with $\mathtt{s}^*=\mathtt{b}\mathtt{s}$, these imply
\beq  X'=\mathtt{ b}X^{-1}\zeta.
\eeq}
Moreover, let $U(q,t)$ and  be the matrix with elements    
\beq U_{\La\Om}=\LL \mathtt{s}_\La |\mathtt{s}_\Om^*\RR_{q,t}. \eeq
The use of the {previous}  two equations then leads to
\beq XU(q,t)X^{-1}=\zeta^{-1}\zeta(q,t),
\eeq
{whose inverse version reads}
\beq  \label{Uinv}
XU(q,t)^{-1}X^{-1}=\zeta(q,t)^{-1}\zeta.
\eeq 
{From the explicit expressions of $\zeta(q,t)$ and $\zeta$ in \eqref{defzeta}, we see that}
\beq \zeta(q,t)^{-1}\zeta=(t^nq^{-n})\zeta^{-1}\zeta(t^{-1},q^{-1}).
\eeq 
{Because $X$ is independent of $q$ and $t$, the relation \eqref{Uinv} implies that}
\beq \label{eqUinv}
U(q,t)^{-1}=(t^nq^{-n})U(t^{-1},q^{-1}).
\eeq
{Furthermore, we have, using \eqref{dualityschur}, the adjoint character of $\omega$ and $\omega^2=1$, we have} 
\beq (JU(q,t)J)_{\La\Om}=\LL \mathtt{s}_{\La'}\,|\,\mathtt{s}^*_{\Om'}\RR_{q,t}=\LL \omega \mathtt{s}^*_{\La}\,|\,\omega {\mathtt{s}_{\Om}}\RR_{q,t}
= \LL \mathtt{s}^*_{\La}\,|\,\mathtt{s}_{\Om}\RR_{q,t},
\eeq so that {(translating the first and third equality above)}
\beq \label{eqUJ}
J U(q,t) J= U(q,t)'=\mathtt{ b}\, U(q,t) \mathtt{ b}^{-1}. \eeq 
One last identity concerning $U(q,t)$ is needed:
\beq \label{eqUA}
A(q,t)\mathtt{ b}^{-1} U(q,t)'A(q,t)'=\mathtt{ b}(q,t)^{-1},
\eeq 
where $\mathtt{b}(q,t)^{-1}$ is the diagonal matrix with entries 
$\LL P_\La(q,t)|P_\La(q,t)\RR_{q,t}$.  {The relation} \eqref{eqUA} follows directly from the orthogonality of the Macdonald superpolynomials and the definition of $A(q,t)$.

We {are now in position to prove} \eqref{eqJA}.  From the triangular expansions of { $P_\La(q,t)$ and $\mathtt{s}_\La$ (which are both of the form $m_\La+\text{lower terms}$), }
we know that the matrix $A(q,t)$ is strictly upper unitriangular.  We recall (see for instance {\cite[I.(6.2)]{Mac})} that a matrix $M$ is strictly upper (uni)triangular if and only if $JMJ$ is strictly lower (uni)triangular. Hence, the matrix 
\beq \label{defmatB}B=JA(t^{-1},q^{-1}) J A(q,t)' \eeq 
is strictly lower unitriangular. 
{Therefore, in order to prove that $B$ is the identity (which will prove  \eqref{eqJA}),
 it only} remains to prove that $B$ is also upper triangular.  The use of the {second equality of} \eqref{eqUA} and \eqref{eqUJ}, {under the form $  U(q,t)\mathtt{b}^{-1}J=\mathtt{b}^{-1} J U(q,t)$,} allows us to write   
\beq   \mathtt{b}(q,t)^{-1} B^{-1}= A(q,t)\mathtt{b}^{-1} J U(q,t) A(t^{-1},q^{-1})^{-1} J\, .
\eeq
{Now, by isolating $U$ from  \eqref{eqUA} and inverting the result by using   \eqref{eqUinv}}, we get 
\beq   \mathtt{b}(q,t)^{-1} B^{-1}=(t^{-n}q^n)A(q,t) \,\mathtt{b}^{-1}J\mathtt{b}^{-1} \,A(t^{-1},q^{-1})' \,\mathtt{b}(t^{-1},q^{-1})J.
\eeq
However, by exploiting \eqref{inversionofb}  and the definition of $J$, one readily shows that $\mathtt{b}^{-1}J\mathtt{b}^{-1}=J$.  Thus,
\beq   \mathtt{b}(q,t)^{-1} B^{-1}=(t^{-n}q^n)\,A(q,t) J A(t^{-1},q^{-1})' \,b(t^{-1},q^{-1})J,
\eeq
The comparison with \eqref{defmatB} yields 
\beq   \mathtt{b}(q,t)^{-1} B^{-1}=(t^{-n}q^n)\, B' \, J \mathtt{b}(t^{-1},q^{-1})J .
\eeq 
Since both $\mathtt{b}(t^{-1},q^{-1})$ and $J \mathtt{b}(t^{-1},q^{-1})J $ are diagonal matrices, and  since  $B$ is strictly lower unitriangular, the last equation implies that $B^{-1}$ is  also  
upper triangular. Consequently,  $B$ is the identity matrix, which completes the proof of \eqref{eqJA}.
\end{proof}

\section{Another scalar product} \label{asp}

The value of the norm $\LL {P_\La}|{P_\La}\RR_{q,t}$ of 
the Macdonald polynomials in superspace
was conjectured in \cite{BDLM}.   This conjecture is reproduced below. 
We will now define another
scalar product with respect to which the Macdonald polynomials in superspace
are also orthogonal.  This other scalar product is not bilinear anymore
(it is sesquilinear).  Nevertheless, we will show that, remarkably, the norm
of the  Macdonald polynomials in superspace with respect to that other 
scalar product is (up to a power of $q$) identical to the conjectural
expression for $\LL {P_\La}|{P_\La}\RR_{q,t}$. 

\subsection{The conjectured norm of the Macdonald superpolynomials}

We first present the conjectural expression for
$\LL {P_\La}|{P_\La}\RR_{q,t}$ given in \cite{BDLM}. 
It involves  the quantities 
(recall that arm- and leg-lengths were defined in \eqref{eqarms})
\begin{equation}\label{defhupdown}
h^\uparrow_\La=\prod_{s\in\mathcal{B}(\La)} (1-q^{a_{\La^*}(s)+1}t^{{l}_{\Lambda^{\circledast}}(s)})\, ,\quad h^\downarrow_\La=\prod_{s\in\mathcal{B}(\La)}(1- q^{{a}_{\Lambda^{\circledast}}(s)}t^{l_{\La^*}(s)+1}),\end{equation}
where $\B(\La)$ denotes the set of squares in the diagram of  $\La$ that do not
appear at the same time in a row containing a circle {and} in a
column containing a circle (this excludes for instance the squares $(1,1),\,(1,2)$ and $(3,1)$ of $\La$ whose diagram is found in \eqref{exdia}). 

\begin{conjecture}\label{con2}
The norm of  $P_\La$  defined in Theorem~\ref{theo1} is
\begin{equation}\label{norm}
\LL {P_\La}|{P_\La}\RR_{q,t}=   (-1)^{\binom{m}{2}} q^{ | \La^a | }\frac{h^\uparrow_\La}{h^\downarrow_\La}.
\end{equation}
\end{conjecture}

\subsection{The constant-term scalar product of the non-symmetric Macdonald polynomials}

Let $\text{C.T.}(f)$ denote the constant term of the Laurent series of
$f$ in the variables $x_1,\dots,x_N$.  Define the following scalar product on $\mathbb Q(q,t)[x_1,\dots,x_N]$:
\begin{equation}
\langle f, g \rangle_{N,q,t}:=\text{C.T.}\left\{ f(x;q,t)\, g(x^{-1};q^{-1},t^{-1})\,W(x;q,t)\right\},
\end{equation}
   where
\begin{equation}W(x;q,t)=\prod_{1\leq i<j \leq N}\frac{(x_i/x_j;q)_\y \,(qx_j/x_i;q)_{\y}}{(tx_i/x_j;q)_\y\,(qtx_j/x_i;q)_\y}.
\end{equation}
Note that this scalar product is sesquilinear, that is,
for $c=c(q,t) \in \mathbb Q(q,t)$, we have
\begin{equation}
\langle c \, f, g \rangle_{N,q,t} = c \, \langle f, g \rangle_{N,q,t} \quad
{\rm and} \quad \langle f, c \, g \rangle_{N,q,t} = \bar c \, 
\langle f, g \rangle_{N,q,t},
\end{equation}
where $\bar c=c(1/q,1/t)$.

\begin{proposition}
\cite{Mac,Mac1}  The non-symmetric Macdonald polynomials $E_\eta(x;q,t)$ form an orthogonal set with respect to $\langle \cdot, \cdot \rangle_ {N,q,t}$.
\end{proposition}

The norm is explicitly given by the expression (cf. \cite[Prop. 3.4]{Mar} and \cite{Che,Ba})
\begin{equation}\label{qtnorE}\frac{\L E_\eta,E_\eta\R_{N,q,t} }{\L 1, 1\R_{N,q,t} }=\frac{d'_\eta\,e_\eta}{d_\eta\,e'_\eta},
\end{equation}
where for a composition $\gamma$,
\begin{equation}\begin{array}{ll}\label{debhqt}\displaystyle
d_\gamma   = \prod_{s\in\gamma}[1-q^{a(s)+1}\,t^{ \l(s)+1}], &  \displaystyle
d'_\gamma = \prod_{s\in\gamma}[1-q^{a(s)+1}\,t^{ \l(s)}],  \\ \displaystyle
e_\gamma = \prod_{s\in\gamma}[1-q^{a'(s)+1}\,t^{N- \l'(s)}], \quad
& \displaystyle e'_\gamma = \prod_{s\in\gamma}[1-q^{a'(s)+1}\,t^{N-1-\l'(s)}],  \\
\displaystyle
b_\gamma = \prod_{s\in\gamma}[1-q^{ a'(s)}\, t^{N-\l'(s)}],  
& \displaystyle h_{\gamma} = \prod_{s\in\gamma}[1-q^{a(s)}\,t^{\l(s)+1}],
 \displaystyle\end{array}
\end{equation}
(we added two expressions to be needed shortly, $b_\gamma$ and $h_\gamma$).
 The {arm- and leg-(co)lengths} in these expressions are given, for $s=(i,j)$, by
\begin{equation}\begin{array}{rl}
a(s)=& \gamma_i-j \, ,  \\
 l(s) = &\# \{k=1,...,i-1 \, | \, j \leq \gamma_k+1\leq \gamma_i \} \\  
& \qquad + \# \{k=i+1,...,N \, |\,  j \leq\gamma_k \leq \gamma_i \} \, , \\
a'(s) = &j-1 \, , \\
l'(s) = & \# \{k=1,...,i-1 \, | \,  \gamma_k\geq \gamma_i \}\\
& \qquad 
+ \# \{k=i+1,...,N \, |\,  \gamma_k > \gamma_i \}.
\end{array}\end{equation}

\subsection{Another orthogonality relation for the Macdonald super\-poly\-no\-mials}

Let $f \in \mathscr R_N$ be a symmetric superpolynomial 
of fermionic degree $m$.  We define
\begin{equation}
f^{\bullet} = t^{-\binom{m}{2}} \frac{\Delta^t_m(x)}{\Delta_m(x)}
[\theta_1 \dots \theta_m] f,
\end{equation}
where we recall that, as in Section~\ref{PvsEs}, $[\theta_1 \cdots \theta_m] f$
stands for the coefficient of $\theta_1 \cdots \theta_m$ in $f$.
Note that $f^{\bullet} \in \mathbb Q(q,t)[x_1,\dots,x_N]$ since
$[\theta_1 \cdots \theta_m] f$ is antisymmetric in the variables
$x_1,\dots,x_m$ and thus divisible by $\Delta_m(x)$.  If $f$
does not have a specific fermionic degree, it can be decomposed
as $f=f_0+f_1+\cdots+f_r$, where $f_m$ is the part of $f$ of
fermionic degree $m$.  We then let
\begin{equation}
f^{\bullet} = f_0^{\bullet} + f_1^{\bullet} + \cdots + f_r^{\bullet} .
\end{equation}

\begin{definition}  Let $f$ and $g$  be superpolynomials in $\mathscr R_N$.
We define the following sesquilinear scalar product on $\mathscr R$:
\begin{equation}\label{anewsp}
\L f , g \R_{N,q,t}^{\mathscr R}= \sum_{m} \frac{1}{[m]_t!}
\L f_m^{\bullet} , g_m^{\bullet} \R_{N,q,t}\, .
\end{equation}
\end{definition}

The following proposition states that the Macdonald superpolynomials are also orthogonal with respect to this new scalar product.  Moreover, it
relates the norm of $P_{\Lambda}$ in the two scalar products.
\begin{proposition} \label{quasiP} 
We have
\begin{equation}
\L   P_\La ,   P_\Om \R_{N,q,t}^{\mathscr R}= 0  \quad {\rm if~} \Lambda \neq \Omega \, .
\end{equation}
Furthermore,
\begin{equation}\label{2normsa}
\lim_{N \rightarrow \infty} \frac{ \L   P_\La ,   P_\La \R_{N,q,t}^{\mathscr R} }{\L 1, 1\R_{N,q,t}} = (-1)^{\binom{m}{2}} \frac{h^\uparrow_\La}{h^\downarrow_\La} \stackrel{?}{=}  q^{-|\La^a |} \, \LL {P_\La}|{P_\La}\RR_{q,t}\, ,
\end{equation}
where $m$ is the fermionic degree of $\Lambda$, and where
$\stackrel{?}{=}$ 
means that the equality is only conjectural.
\end{proposition}
\begin{proof}  By definition 
$\L   P_\La ,   P_\Om \R_{N,q,t}^{\mathscr R}= 0$ if $\Lambda$
and $\Omega$ have different fermionic degrees.  We can thus suppose that
$\La$ and $\Omega$ have fermionic degree $m$.
The proof now follows the argument of \cite{Ba} up to eq. \eqref{spCT}.
 Using expression \eqref{lesSp2}, we get {\small
\begin{multline}
\L   P_\La ,   P_\Om \R_{N,q,t}^{\mathscr R}
\\ =\frac{c_\La(t)\, c_\Om(t^{-1})}{[m]_t !} \sum_{\sigma \in S_{m}} \sum_{\sigma' \in S_{m^c}} 
\L U_m^-U_{m^c}^+E_{\La^R},  (-t)^{-\ell(\sigma)} T_\sigma T_{\sigma'}E_{\Om^R}   \R_{N,q,t} \\
=\frac{c_\La(t)\, c_\Om(t^{-1})}{[m]_t !} \sum_{\sigma \in S_{m}} \sum_{\sigma' \in S_{m^c}} 
\L (-t)^{\ell(\sigma)} T_{\sigma'^{-1}}^{-1} T_{\sigma^{-1}}^{-1} U_{m}^-U_{m^c}^+E_{\La^R},   E_{\Om^R}   \R_{N,q,t} \, , 
\end{multline}}\hspace{-.1cm}
where we used the fact that $T_{\sigma^{-1}}^{-1}$ is the adjoint of $T_\sigma$ with respect to $\L \cdot, \cdot \R_{N,q,t}$.
Note that $T_\sigma$ and $T_{\sigma'}$ commute because they act on  disjoint sets of  variables.
 We then use
\begin{eqnarray}
T_i^\pm U^+=t^{\pm}U^+\quad \text{and}\quad T_i^\pm U^-=-U^-
\end{eqnarray} in order to write
\begin{equation}\label{ps444}
\L   P_\La ,   P_\Om \R_{N,q,t}^{\mathscr R}
=[N-m]_{t^{-1}}!  \: c_\Om(t^{-1})   
\L S_{(\La^a,\La^s)}  ,  E_{\Om^R}   \R_{N,q,t} ,
\end{equation}
where we used $\sum_{\sigma\in S_r } t^{\ell(\sigma)}=[r]_t !$,  and where  
\begin{equation}\label{lesSp1} 
S_{(\Lambda^a,\La^s)}=c_{\La}U_m^- U_{m^c}^+
      E_{( ({\Lambda^a})^R,(\Lambda^s)^R)},
\end{equation}
is a Macdonald polynomial with mixed symmetry
(also considered in \cite{Ba} but only in the case where 
$\la=\delta^{(m)}=(m-1,\dots,0))$. 
The analogue of \cite[Corr. 1]{BDF}, which is obtained as shown there using the generalization of
lemma 2.5 of \cite{Mar}, reads
\begin{equation}\label{lesSp2}
S_{(\lambda,\mu)}
=
\sum_{\s,\s'} (-t)^{-\ell(\s')} \frac{d'_{ ( \la^R,\mu) }  d_{(  \s'(\la^R), \s(\mu^R) )}  }{d'_{( \la^R,\s(\mu^R))} d_{(\la,\s(\mu^R))}}    E_{(\s'(\lambda^R),\s(\mu^R))}.
\end{equation}
By the orthogonality of the non-symmetric Macdonald polynomials,
the rhs 
 of \eqref{ps444} is non-vanishing only when
$\La=\Om$ , and  
$\sigma$ and $\sigma'$ are equal to the identity in 
\eqref{lesSp2}.
Hence
\begin{multline}
\L   P_\La ,   P_\Om \R_{N,q,t}^{\mathscr R}
= \delta_{\La \Om}\; (-1)^{\binom{m}{2} }\frac{ [N-m]_{t^{-1}}!  \:  t^{\text{inv}(\La^s)} }{ \; f_{\La^s}(t^{-1}) } \\ \times
 \frac{d'_{ ( (\La^a)^R , \La^s)} d_{  ((\La^a)^R , (\La^s)^R)}     }{ d'_{ ( (\La^a)^R , (\La^s)^R)}  d_{(  \La^a,(\La^s)^R  )}    }
  \L E_{\La^R}  ,  E_{\La^R}   \R_{N,q,t} \, ,
\end{multline}
which implies that $\L   P_\La ,   P_\Om \R_{N,q,t}^{\mathscr R}=0$
if $\La \neq \Om$.

Furthermore, according to the identities given in 
\cite[eqs (5.16) and (5.15)]{Mar}, 
\begin{equation}
\frac{[N-m]_{t^{-1}}!}{f_{\La^s}(t^{-1})} \; t^{\text{inv}(\La^s)}=\frac{[N-m]_{t}!}{f_{\La^s}(t)}=\frac{b_{\La^s} d_{(\La^s)^R} }{h_{\La^s} e_{(\La^s)^R}},
\end{equation}
so that, using the expression \eqref{qtnorE} for the norm of $ E_{\La^R}$, we obtain
\begin{equation}\label{spCT}
\frac{ \L   P_\La ,   P_\La \R_{N,q,t}^{\mathscr R} }{\L 1, 1\R_{N,q,t}} =
(-1)^{\binom{m}{2} }   \frac{b_{\La^s} d_{(\La^s)^R} }{h_{\La^s} e_{(\La^s)^R}} \frac{d'_{ ( (\La^a)^R , \La^s)}    }{ d_{(  \La^a,(\La^s)^R  )}    } \frac{e_{((\La^a)^R, (\La^s)^R  )} }{    e'_{ ((\La^a)^R, (\La^s)^R  ) }  }.
\end{equation}

In order to relate this expression to the conjectural expression for
$\LL {P_\La}|{P_\La}\RR_{q,t}$, we need to consider the 
 limit $N\rw \y$ of \eqref{spCT} and recall that $|t|<1$, so that all the factors $e$, $e'$ and $b$ reduce to 1.  The proposition will follow 
after establishing that
\begin{equation}\label{2norms}
\lim_{N \rightarrow \infty} \frac{ \L P_\La , P_\La \R_{N,q,t}^{\mathscr R} }{\L 1, 1\R_{N,q,t}} = 
 (-1)^{\binom{m}{2} } \lim_{N \rightarrow \infty} \frac{   d'_{ ( (\La^a)^R , \La^s)}   d_{(\La^s)^R}}{  d_{(  \La^a,(\La^s)^R  )} h_{\La^s} }
= (-1)^{\binom{m}{2} } \frac{h^\uparrow_\La}{h^\downarrow_\La}.
\end{equation}
The last equality is quite long to prove. The details are thus relegated to
the appendix. 
\end{proof}

\section{Macdonald superpolynomials and generalized $(q,t)$-Kostka coefficients}\label{Kos}

In this section, we review the super-extension of the Macdonald positivity conjecture, exhibit symmetries of the generalized 
$(q,t)$-Kostka coefficients and  present two new conjectures 
related to these coefficients.

\subsection{Generalized $(q,t)$-Kostka coefficients}\label{gkos}
The 
generalized $(q,t)$-Kostka coefficients are defined from the integral form of the Macdonald 
superpolynomials:
\begin{equation} \label{intform}
J_\La = {h^\downarrow_\La}\, P_\La
\end{equation} 
(where ${h^\downarrow_\La}$ is defined in \eqref{defhupdown}).  It was conjectured
in \cite{BDLM} that the coefficients in the monomial expansion of 
$J_\La$ are polynomials in $q$ and $t$ with integer coefficients.

We next introduce the Schur superpolynomials \cite{BDLM}
\begin{equation}\label{schur}
s_\Lambda(x,\theta)=P_\Lambda(x,\ta;0,0),
\end{equation}
and
their deformation
\begin{equation}\label{defM}
 S_\Lambda (x,\ta;t) = \, \varphi \bigl(s_\Lambda (x,\ta)\bigr),
 \end{equation}
where  $\varphi$  stands for the endomorphism of
$\mathbb Q(q,t)[p_1,p_2,p_3,\dots;\tilde p_0,\tilde p_1,\tilde p_2,\dots]$ 
defined by its action on the power-sums as
\begin{equation}
\varphi(p_n)=(1-t^n)p_n\qquad\text{and}\qquad 
\varphi(\tilde p_n)=\tilde p_n.
 \end{equation}

\begin{remark} As was commented in  \cite{BDLM}, 
the existence of the limiting case
$s_\Lambda(x,\theta)=P_\Lambda(x,\ta;0,0)$ does not follow{ from the existence 
of a solution from Theorem \ref{theo1} }
since the scalar product is degenerate when $q=t=0$.  A better approach
is to consider the limit $q=t\to 0$ in Definition~\ref{PvsE}.   This is presented in Appendix \ref{Key}.  Specifically, we prove the existence of the Schur superpolynomials by using the fact that in the limit $q=t\to 0$,  
a non-symmetric Macdonald polynomial tends to a Key polynomial (or Demazure character) \cite{Ion,LasS}.
\end{remark}

 {We now recall the following striking version of the Macdonald positivity conjecture,  formulated in \cite{BDLM}.}

\begin{conjecture}\label{ConjMac}
 The coefficients $ K_{\Omega \Lambda} (q,t)$ in the expansion of the
integral form of the Macdonald superpolynomials
\begin{equation}
 J_\Lambda (x,\ta;q,t)  = \sum_\Omega K_{\Omega \Lambda} (q,t) \,S_\Omega(x,\ta;t)
 \end{equation}
are polynomials in $q$ and $t$ with nonnegative  integer coefficients.\end{conjecture}

The following symmetries have been observed: 
\begin{equation}\label{sym1}
K_{\Om \La}(q,t)=K_{\Om' \La'}(t,q),
 \end{equation}
(cf. \cite[eq. VI (8.15)]{Mac}) and
\begin{equation}\label{sym2}
K_{\Om\La}(q,t)=q^{\bar{n}(\La')} t^{\bar{n}(\La)}\, K_{\Om' \La}(q^{-1},t^{-1}) .
\end{equation}
 In the previous equation we used
\begin{equation}
{\bar{n}(\Lambda)}=n(\S{\La})-d(\La)\qquad\text{with}\qquad n(\la)=\sum_i(i-1)\la_i \, ,
 \end{equation}
where $\S \La$ is the skew diagram 
$\S \La=\Lambda^{\circledast}/\delta^{(m+1)}$, with $\delta^{(m)}$ the staircase
partition $(m-1,m-2,\dots,1,0)$, and where
$d(\La)$ is defined as follows: fill each square
 $s\in\B\La$ (defined immediately after \eqref{defhupdown}) by the number corresponding to the number of squares above  $s$ which are both in a fermionic row and a fermionic column (i.e., both ending with a circle); $d(\La)$ is then the sum of these entries.
For example, {\small
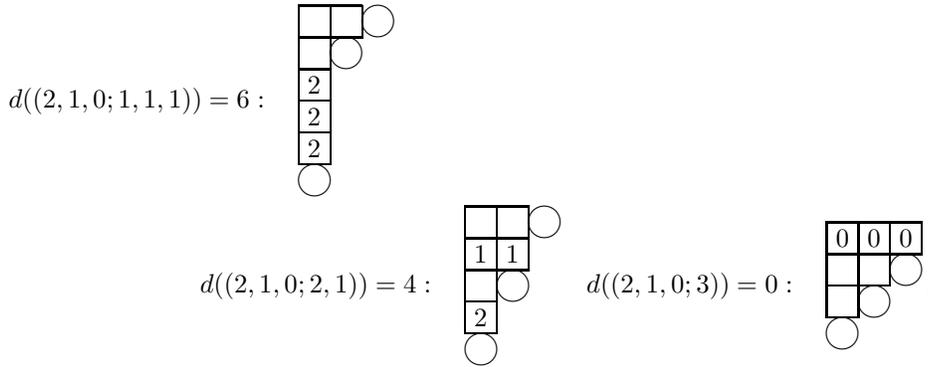
\begin{multline*} 
d((2,1,0;1,1,1))=6:\quad {\tableau[scY]{&&\bl\tcercle{}\\&\bl\tcercle{}\\2 \\2 \\2 \\ \bl\tcercle{} \\ }} \\ \quad
d((2,1,0;2,1))=4:\quad {\tableau[scY]{&&\bl\tcercle{}\\1&1  \\& \bl\tcercle{}\\2 \\ \bl\tcercle{} }} \quad
d((2,1,0;3))=0:\quad {\tableau[scY]{0&0&0\\&&\bl\tcercle{} \\& \bl\tcercle{}\\ \bl\tcercle{} }} \, .\end{multline*}}
(In a sense,  $d(\La)$ is dual to the quantity $\zeta_\La$ introduced in \cite{BDLM} to describe the specialization of $P_\La$).  The quantity $d({\La})$
vanishes for $m\leq1$, so that that 
expression \cite[eq. VI (8.14)]{Mac} is recovered when $m=0$.  Examples of $(q,t)$-Kostka coefficients are given in Tables \ref{tab11} to \ref{tab42}.  We stress that in these tables,   the prime stands for the matrix transpose, so that the matrix element found in row $\Lambda$ and column $\Omega$  gives the coefficient $K'_{\Lambda,\Omega}=K_{\Omega,\Lambda}$ of $S_\Omega$ in the modified Macdonald superpolynomial $J_\La$.  Note also that in the tables presented in \cite{BDLM}, 
the transpose symbol (prime) is missing; 
these correspond to Tables \ref{tab11} -- \ref{tab32} below.

An example illustrating the first symmetry property \eqref{sym1} is (cf. Table \ref{tab42})
\begin{equation}K_{(2,0;1,1)\,(2,0;2)}(q,t)=q+qt+q^2t=
K_{(3,0;1)\,(2,0;2)}(t,q),
\end{equation}
 since $(2,0;1,1)'=(3,0;1)$ and $(2,0;2)$ is self-conjugate. An example of the relation \eqref{sym2} is
\begin{equation} K_{(2,0;1,1)\,(1,0;2,1)}(q,t)=t+qt^2+qt^3=qt^5 K_{(3,0;1)\,(1,0;2,1)}(q^{-1},t^{-1})\, . 
\end{equation}
The factor $t^5$ can be checked as follows (the diagram of $\La$ is filled with numbers that add up to $d(\La)$ while the skew diagram $\S\La$, obtained from $\La^\cd$ by dropping the squares marked by a $\times$,   is filled by numbers that add up to $n(\mathcal S\La)$):
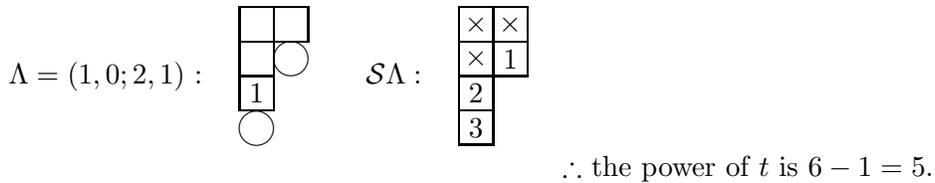
\begin{multline*}\La=(1,0;2,1):\quad {\tableau[scY]{&\\&\bl\tcercle{}\\1 \\ \bl\tcercle{} \\ }} 
\qquad\S\La:\quad {\tableau[scY]{\times &\times\\ \times&1\\ 2 \\ 3 \\ }} \\ \quad \text{$\therefore$ the power of $t$ is $6-1=5$}.\end{multline*}
Similarly, for $q$ factor, we have
\begin{multline*}\La'=(3,1;0):\quad {\tableau[scY]{0&&0&\bl\tcercle{}\\0& \bl\tcercle{} \\ }} 
\qquad\S\La':\quad {\tableau[scY]{\times&\times&0&0\\ \times&1\\ }} \\ \quad \text{$\therefore$ the power of $q$ is $1-0=1$}.\end{multline*}

\begin{table}[ht]
\caption{  $K(q,t)'$ for degree $(1|1)$.  
} 
\label{tab11}
\begin{center}
\begin{tabular}{c| c |c | } 
 & $(1;\,) $ & $(0;1)$ \\ \hline
$(1;)$ & $1$ & $q$ \\ \hline
$(0;1\,)$ & $t$ & $1$ \\ \hline
\end{tabular}
\end{center}
\end{table}
\begin{table}[ht]
\caption{ $K(q,t)'$ for  degree $(2|1)$.} 
\label{tab21}
\begin{center}
\begin{tabular}{c| c |c | c | c | } 
 & $(2;\,) $ & $(0;2 )$ & $(1;1 )$ & $(0;1,1 )$ \\ \hline
$(2;\,)$ & $1$ & $q^2$ & $q $ & $q^3 $ \\ \hline
$(0;2)$ & $t$ & $ 1$ & $qt$ & $q$\\ \hline
$(1;1)$ & $t$ & $ qt$ & $1$ & $q$\\ \hline
$(0;1,1)$ & $t^3$ & $ t$ & $t^2$ & $1$\\ \hline
\end{tabular}
\end{center}
\end{table}
\begin{table}[ht]
\caption{$K(q,t)'$ for   degree $(2|2)$.} 
\label{tab22}
\begin{center}
\begin{tabular}{c| c |c | } 
 & $(2,0;\,) $ & $(1,0;1)$ \\ \hline
$(2,0;\,)$ & $1$ & $q$ \\ \hline
$(1,0;1)$ & $t$ & $1$ \\ \hline
\end{tabular}
\end{center}
\end{table}
\begin{table}[ht]
\caption{ $K(q,t)'$ for  degree $(3|1)$.} 
\label{tab31}
\begin{center}
\begin{tabular}{c| c |c | c | c | c| c| c|} 
                  & $(3;\,) $ & $(0;3 )$ & $(2;1 )$ & $(1;2 )$ & $(0;2,1 )$ & $(1;1,1 )$ & $(0;1,1,1 )$ \\ \hline
$(3;\,)$ & $1$ & $q^3$ & $q+q^2 $ & $q^2+q^4 $  & $q^4+q^5$ & $q^3$ & $q^6$  \\ \hline
$(0;3)$ & $t$ & $ 1$ & $qt + q^2t$ & $q+q^2t$  & $q+q^2$  & $q^3t$ & $q^3$ \\ \hline
$(2;1)$ & $t$ & $ q^2t$ & $1+qt$ & $q+q^2t$  & $q^2+q^3t$  & $q$ & $q^3$\\ \hline
$(1;2)$ & $t^2$ & $ qt$ & $t+ qt^2$ & $1+q^2 t^2$ & $q+q^2t$  & $qt$  & $q^2$\\ \hline
$(0;2,1)$ & $t^3$ & $ t$ & $t^2+qt^3$ & $t+qt^2$ & $1+qt$  & $qt^2$  & $q$\\ \hline
$(1;1,1)$ & $t^3$ & $ qt^3$ & $t+t^2$ & $t+qt^2$ & $qt+qt^2$  & $1$  & $q$\\ \hline
$(0;1,1,1)$ & $t^6$ & $ t^3$ & $t^4+t^5$ & $t^2+t^4$ & $t+t^2$  & $t^3$  & $1$\\ \hline
\end{tabular}
\end{center}
\end{table}
\begin{table}[ht]
\caption{ $K (q,t)'$ for  degree $(3|2)$.} 
\label{tab32}
\begin{center}
\begin{tabular}{c| c |c | c | c |c| } 
 & $(3,0;\,) $ & $(2,1;\, )$ & $(2,0;1 )$ & $(1,0;2 )$ & $(1,0;1,1 )$ \\ \hline
$(3,0;\,)$ & $1$ & $q$ & $q+q^2 $ & $q^2 $  & $q^3$ \\ \hline
$(2,1;\,)$ & $qt$ & $ 1$ & $q+q^2t$ & $q^3t$ & $q^2$\\ \hline
$(2,0;1)$ & $t$ & $ t$ & $1+qt$ & $q$ & $q$\\ \hline
$(1,0;2)$ & $t^2$ & $ qt^3$ & $t+qt^2$ & $1$ & $qt$\\ \hline
$(1,0;1,1 )$ & $t^3$ & $ t^2$ & $t+t^2$ & $t$ & $1$\\ \hline
\end{tabular}
\end{center}
\end{table}

\footnotesize{

\begin{table}[ht]
\caption{ $K (q,t)'$ for  degree $(4|1)$.}
\label{tab41}
\begin{center}
\begin{tabular}{l  | c | c | c | c |c | c |} 
 &$(4;)$  & $(0;4)$ & $(3;1)$ & $(1;3)$  &  $(0;3, 1)$ & $(2;2)$   \\ \hline
$(4;)$  &  $1$ & $q^4 $& $q+q^2+q^3 $& $q^3+q^5+q^6 $& $q^5+q^6+q^7$ & $q^2+q^4$  \\ \hline
$(0;4)$ & $t$ & $1$ & $qt+q^2t+q^3t$ & $q+q^2+q^3t$ & $q+q^2+q^3$ & $q^2t+q^4t$ \\ \hline
$(3;1)$ & $t$ & $q^3t$ & $1+qt+q^2t$ & $q^2+q^3t+q^4t$ & $q^3+q^4t+q^5t$ & $q+q^2t$ \\ \hline
$(1;3)$  & $t^2$ & $qt$ & $t+qt^2+q^2t^2$ & $1+q^2t+q^3t^2$ & $q+q^2t+q^3t$ & $qt+q^2t^2$  \\ \hline
$(0;3, 1)$ & $t^3$ & $t$ & $t^2+qt^3+q^2t^3$ & $t+qt+q^2t^2$ & $1+qt+q^2t$ & $qt^2+q^2t^3$ \\ \hline
$(2;2)$  & $t^2$ & $q^2t^2$ & $t+qt+qt^2$ & $qt+q^2t+q^3t^2$ & $q^2t+q^3t+q^3t^2$ & $1+q^2t^2$ \\ \hline
$(0;2, 2)$ & $t^4$ & $t^2$ & $t^3+qt^3+qt^4$ & $t+qt^2+qt^3$ & $t+qt+qt^2$ & $t^2+q^2t^4$ \\ \hline
$(2;1, 1)$  &$ t^3$ & $q^2t^3$ & $t+t^2+qt^3$ & $qt+q^2t^2+q^2t^3$ & $q^2t+q^2t^2+q^3t^3$ & $t+qt^2$\\ \hline
$(1;2, 1)$  & $t^4$ & $qt^3$ & $t^2+t^3+qt^4$ & $t+qt^2+q^2t^4$ & $qt+qt^2+q^2t^3$ & $t^2+qt^3$ \\ \hline
$(0;2, 1, 1)$ & $t^6$ & $t^3$ & $t^4+t^5+qt^6$ & $t^2+t^3+qt^4$ & $t+t^2+qt^3$ & $t^4+qt^5$ \\ \hline
$(1;1^3)$  & $ t^6$ & $qt^6$ & $t^3+t^4+t^5$ & $t^3+qt^4+qt^5$ & $qt^3+qt^4+qt^5$ & $t^2+t^4$ \\ \hline
$(0;1^4)$ & $t^{10}$ & $t^6$ & $t^7+t^8+t^9$ & $t^4+t^5+t^7$ & $t^3+t^4+t^5$ & $t^6+t^8$ \\ \hline
\end{tabular}

\vspace{1cm}

\begin{tabular}{l  | c | c | c | c |c | c | } 
 & $(0;2, 2)$ &  $(2;1, 1)$  & $(1;2, 1)$  & $(0;2, 1, 1)$  & $(1;1^3)$  &  $(0;1^4)$ \\ \hline
$(4;)$  &  $q^6+q^8$ & $q^3+q^4+q^5$& $q^4+q^5+q^7$ & $q^7+q^8+q^9$ &$ q^6$& $q^{10}$ \\ \hline
$(0;4)$ & $q^2+q^4$ & $q^3t+q^4t+q^5t$ & $q^3+q^4t+q^5t$ & $q^3+q^4+q^5$ & $ q^6t$ &$ q^6$ \\ \hline
$(3;1)$ & $q^4+q^5t$ & $q+q^2+q^3t$ & $q^2+q^3+q^4t$ & $q^4+q^5+q^6t$ & $q^3 $&$ q^6$\\ \hline
$(1;3)$  & $q^2+q^3t$ & $qt+q^2t+q^3t^2$ & $q+q^2t+q^4t^2$ & $q^2+q^3+q^4t$ & $q^3t$ & $q^4$\\ \hline
$(0;3, 1)$ & $q+q^2t$ & $qt^2+q^2t^2+q^3t^3$ & $qt+q^2t^2+q^3t^2$ & $q+q^2+q^3t$ & $q^3t^2$ & $q^3$ \\ \hline
$(2;2)$  & $q^2+q^4t^2$ & $q+qt+q^2t$ & $q+q^2t+q^3t$ & $q^3+q^3t+q^4t$ & $q^2$ & $q^4$ \\ \hline
$(0;2, 2)$ & $1+q^2t^2$ & $qt^2+qt^3+q^2t^3$ & $qt+qt^2+q^2t^3$ & $q+qt+q^2t$ & $q^2t^2$ & $q^2$\\ \hline
$(2;1, 1)$  & $q^2t+q^3t^2$ & $1+qt+qt^2$ & $q+qt+q^2t^2$ & $q^2+q^3t+q^3t^2$ & $q$ & $q^3$ \\ \hline
$(1;2, 1)$  & $qt+q^2t^2$ & $t+qt^2+qt^3$ & $1+qt^2+q^2t^3$ & $q+q^2t+q^2t^2$  & $qt$ & $q^2$ \\ \hline
$(0;2, 1, 1)$ & $t+qt^2$ & $t^3+qt^4+qt^5$ & $t^2+qt^3+qt^4$ & $1+qt+qt^2$ & $qt^3$ & $q$ \\ \hline
$(1;1^3)$  & $qt^2+qt^4$ & $t+t^2+t^3$ & $t+t^2+qt^3$ & $qt+qt^2+qt^3$ & $1$ & $q$ \\ \hline
$(0;1^4)$ & $t^2+t^4$ & $t^5+t^6+t^7$ & $t^3+t^5+t^6$ & $t+t^2+t^3$ & $t^4$ & $1$\\ \hline
\end{tabular}
\vspace{0.5cm}

\end{center}
\end{table}
}

\begin{landscape}
\vspace{1cm}
{{\footnotesize
\begin{table}[ht]
\caption{ $K (q,t)'$ for  degree $(4|2)$.} 
\label{tab42}
\begin{tabular}{l | c | c | c | c | c | c | c | c | c |} 
  & $(4, 0;)$ & $(3, 1;)$ & $(3, 0; 1)$ & $(1, 0;3)$ & $(2, 0;2)$ & $(2, 1;1)$ & $(2, 0;1, 1)$ & $(1, 0;2, 1)$ & $(1, 0;1^3)$  \\ \hline
$(4, 0;)$ & $1$ &$ q+q^2$ & $q+q^2+q^3$ & $q^3$ & $q^2+q^4$ & $q^3$ & $q^3+q^4+q^5$ & $q^4+q^5$ & $q^6$ \\ \hline
$(3, 1;)$ & $qt$ & $1+q^2t$ & $q+q^2t+q^3t$ & $q^4t$ & $q^2+q^3t$ & $q$ & $q^2+q^3+q^4t$ & $q^3+q^5t$ & $q^4$  \\ \hline
$(3, 0; 1)$ &  $t$ & $t+qt$ & $1+qt+q^2t$ & $q^2$ & $q+q^2t$ & $qt$ & $q+q^2+q^3t$ & $q^2+q^3$ & $q^3$ \\ \hline
$(1, 0;3)$ & $t^2$ & $qt^2+q^2t^3$ & $t+qt^2+q^2t^2$ & $1$ & $qt+q^2t^2$ & $q^3t^3$ & $qt+q^2t+q^3t^2$ & $q+q^2t$ & $q^3t$  \\ \hline
$(2, 0;2)$ &   $t^2$ & $t+qt^2$ & $t+qt+qt^2$ & $qt$ & $1+q^2t^2$ & $qt$ & $q+qt+q^2t$ & $q+q^2t$ & $q^2$ \\ \hline
$(2, 1;1)$ &   $qt^3$ & $t+qt^2$ & $qt+qt^2+q^2t^3$ & $q^3t^3$ & $qt+q^2t^2$ & $1$ & $q+q^2t+q^2t^2$ & $q^2t+q^3t^2$ & $q^2$\\ \hline
$(2, 0;1, 1)$ &   $t^3$ & $t^2+t^3$ & $t+t^2+qt^3$ & $qt$ & $t+qt^2$ & $t^2$ & $1+qt+qt^2$ & $q+qt$ & $q$\\ \hline
$(1, 0;2, 1)$ & $t^4$ & $t^3+qt^5$ & $t^2+t^3+qt^4$ & $t$ & $t^2+qt^3$ & $qt^4$ & $t+qt^2+qt^3$ & $1+qt^2$ & $qt$  \\ \hline
$(1, 0;1^3)$ &  $t^6$ & $t^4+t^5$ & $t^3+t^4+t^5$ & $t^3$ & $t^2+t^4$ & $t^3$ & $t+t^2+t^3$ & $t+t^2$ & $1$\\ \hline
\end{tabular}
\end{table}
}}
\end{landscape}

\subsection{Two new  conjectures for the Kostka coefficients}\label{2new}

We conclude this section with the formulation of two remarkable conjectures that relate the generalized coefficients $K_{\Omega \Lambda}(q,t)$ of fermionic degree $m=1$ 
to the usual $(q,t)$-Kostka coefficients.

\begin{conjecture}\label{kos1}
Let $\Lambda$ be a superpartition of fermionic degree $m=1$, and let
$J_{\Lambda}$ be the integral form of the 
Macdonald superpolynomial.  Let also $\psi$
be the linear application that maps $S_{\Omega}$ to $S_{\Omega^{\circledast}}$.
Then
\begin{equation} 
\psi(J_{\Lambda})= J_{\Lambda^{\circledast}}.
\end{equation} 
\end{conjecture}

This conjecture implies that the usual $(q,t)$-Kostka coefficient 
$K_{\mu \la}(q,t)$ of two partitions $\mu,\la$ can be calculated 
from their lower-degree super-relatives as
\begin{equation} \label{relatekostka}
K_{\mu \la}(q,t) = \sum_{\Omega \, | \, \Om^\cd = \mu} K_{\Omega \Lambda }(q,t)\, ,
\end{equation} 
where $\Lambda$ is any superpartition that can be obtained from
$\lambda$ by replacing a square by a circle and the 
sum is over all $\Omega$'s that can be obtained from
$\mu$  by replacing a square by a circle. Moreover, the expression for the sum on the right-hand side is independent of the choice of $\La$.
We thus relate a Kostka coefficient of degree $(n|0)$ to a sum of degree $(n\!-\!1|1)$ Kostka coefficients. 
In other words, the $(n\!-\!1|1)$ Kostka coefficients provide a refinement 
of the usual $(q,t)$-Kostka coefficients.

For instance, consider $\la=\mu=(3,1)$. There are two  ways of replacing a square by a circle:
\begin{equation} 
 {\tableau[scY]{&&\\\\ }}\;\; \Rw\; \;
  {\tableau[scY]{&&\bl\tcercle{}\\ \\ }} \; \quad
  \text{or}\;\quad  {\tableau[scY]{&&\\ \bl\tcercle{} \\ }}\;.
 \end{equation}
Choosing $\La=(2;1)$, we have  \cite{BDLM}
\begin{multline} \label{ex1}
J_{(2;1)}=tS_{(3; \,)}+q^2tS_{(0;3)}+   ( 1+qt)S_{(2;1)}  + q(1+qt)S_{(1;2)}\\    +q^2(1+qt)S_{(0;2,1)}+qS_{(1;1,1)}+q^3S_{(0;1,1,1)}\, ,
\end{multline}  
so that, using \eqref{relatekostka},
\begin{equation} 
K_{(2;1)\,(2;1)}+K_{(0;3)\,(2;1)}=(1+qt)+q^2t=K_{(3,1)\,(3,1)}.
\end{equation}  
The same result follows by taking $\La=(0;3)$:
\begin{multline} \label{ex2}
J_{(0;3)}=tS_{(3;\,)}+S_{(0;3)}+  qt( 1+q)S_{(2;1)}   + q(1+qt)S_{(1;2)}   \\ +q(1+q)S_{(0;2,1)}+ q^3tS_{(1;1,1)} +q^3S_{(0;1,1,1)}\, ,
\end{multline}  
which implies that\begin{equation} 
K_{(2;1)\,(0;3)}+K_{(0;3)\,(0;3)}=(qt+q^2t)+1=K_{(3,1)\,(3,1)} \, .
\end{equation}  

To formulate the second conjecture, we need to introduce the notion of a concatenable superpartition, defined as one for which $\La^a_m\geq \La_1^s$. Such a superpartition can be transformed into a partition $\la$ by removing the semi-coma: $\la= (\Lambda^a,\Lambda^s)$. For instance, $(5,3;2,1,1)$ is concatenable and the corresponding partition is $\la=(5,3,2,1,1)$. 

\begin{conjecture}  Let $\Lambda$ be a concatenable superpartition of fermionic
degree $m=1$, and let $\lambda$ be its corresponding partition. Let also $\phi$ be the linear application that maps $S_{\Omega}$ to $S_\mu$ 
if $\Omega$ is concatenable (and $\mu$ its correspoding partition)
 and to zero otherwise. 
Then
\begin{equation} 
\phi(J_{\Lambda})= J_{\lambda}
\end{equation}  
\end{conjecture}

If $\La$ and $\Om$ are two $m=1$ concatenable superpartitions whose corresponding partitions are $\la$ and $\mu$, this implies that: 
\begin{equation} 
K_{\La \Om}=K_{\la\mu}.
\end{equation}  
For instance, considering $\phi(J_{(2;1)})$ given in \eqref{ex1}, one recovers the expression of $J_{(2,1)}$:
\begin{equation} \label{ex1p}J_{(2,1)}=qS_{(1,1,1)}+ ( qt+1)S_{(2,1)}+tS_{(3)}.
\end{equation}  

\begin{appendix}

\section{Schur polynomials in superspace}\label{Key}

{In this section, we prove that the limits of the Macdonald polynomials $P_\La(x;q,t)$, as $q=t\to0$ and $q=t\to\infty$, are well defined; they provide two new families of Schur polynomials in superspace, namely $s_\La(x)$ and $\bar s_\La(x)$, whose existence was only conjectured in \cite{BDLM}.  The proof almost immediately follows from \eqref{PvsE} (see lemma \ref{propP0} below), which expresses a Macdonald superpolynomial in terms of non-symmetric Macdonald polynomials $E_\eta(x;q,t)$, and from  Ion's work \cite{Ion,Ion2}, which shows that a non-symmetric Macdonald polynomial (for an affine root system) can be interpreted, when $t\to 0,\infty$, as the character of a certain Demazure module, and more particularly, proves the regularity of $E_\eta(x;q,t)$ as $q=t\to 0,\infty$.   

It is worth noting that  the characters of the Demazure modules were also studied from an algebro-combinatorial point of view  by Fu and Lascoux \cite{Lascoux}.   They showed that the character formulas can be written in terms of key polynomials $\Ke_\eta(x)$ and $\hat \Ke_\eta(x)$ (see \cite[Section 4]{Lascoux} for a precise definition).  Combining the results of \cite{Ion2} and \cite{Lascoux} on the irreducible root system of type $A$, one readily obtains  the explicit relation between the limiting non-symmetric Macdonald polynomials (of type $A$) and the key polynomials:  
\begin{equation} \label{Einkey}
E_\eta(x;0,0) = \widehat \Ke_\eta(x), \qquad  E_\eta(x;\infty,\infty) = \Ke_{\omega_{\mathbf n} \eta}(\omega_{\mathbf n} x)\, ,
\end{equation} 
where $\eta=(\eta_1, \eta_2, \ldots, \eta_n)$ is a composition, $  \omega_{\mathbf n}\eta = (\eta_n, \eta_{n-1}, \ldots, \eta_1)$ and $(\omega_{\mathbf n} x)=(x_n,\ldots,x_1)$.}

\begin{lemma}\label{propP0}
The Macdonald superpolynomials $P_\La(x,\theta;q,t)$ can also be written in terms of the non-symmetric Macdonald polynomials as
\begin{equation}\label{P0}
P_{\La}= c'_{\La}(t) \sum_{\sigma \in S_N/(S_m \times S_{m^c})} \mathcal{K}_\sigma \left( \theta_1 \cdots \theta_m  A_m U^{+}_{m^c} E_{( (\La^a)^R, \La^s)} \right)\,,
\end{equation}
where 
\begin{equation}
c'_\La(t)=\frac{(-1)^{\binom m 2}}{f_{\La^s}(t)}\, .
\end{equation}
\end{lemma}
\begin{proof}
This expression differs from \eqref{PvsE} in that we consider here the composition obtained from the concatenation of $((\La^a)^R, \La^s)$ instead of its (fully) reversed version. To fix the normalization coefficient, we follow exactly the proof of Proposition \ref{propuni} except that equation \eqref{fixcoeff} is now replaced by
\begin{equation}
[ x^\La ] \quad \sum_{w \in S_{m^c}} T_w x^\La =  \sum_{w \in S_{m^c} \,|\, w(\La)=\La} t^{\ell(w)}.
\end{equation}
We need to consider all those permutations that leaves $\La^s$ fixed, which gives
\begin{equation}
[x^\La] \quad \sum_{w \in S_{m^c}} T_w x^\La =  \prod_{i=0}^{\La^s_1 }  \left(  \sum_{w^{(i)} \in S_{n_{\La^s}(i)}  } t^{\ell(w^{(i)})} \right) = \prod_{i=0}^{\La^s_1} [n_{\La^s}(i)]_t! = f_{\La^s}(t).
\end{equation}
We thus have, as expected, 
\begin{equation}
[m_\La]\,  P_\La =  c'_\La(t) (-1)^{\binom m 2} f_{\La^s}(t) =1\, .
\end{equation}
\end{proof}

\begin{proposition}
The  Schur superpolynomials, 
\begin{equation}
s_\La(x,\theta) := P_\La(x,\theta;0,0), \quad \text{and} \quad \bar s_\La(x,\theta) = P_\La(x,\theta;\infty,\infty),  
\end{equation}
are well-defined symmetric polynomials in superspace.  Furthermore, the Schur superpolynomials can be expressed in terms of the  key polynomials (see \eqref{sinkey} and \eqref{sbarinkey} below).
\end{proposition}
\begin{proof}We first set $q=t$ in \eqref{P0} and take the limit $t\to 0$.  Given that  $E_\eta(x;0,0)$ is well-defined, we only need to evaluate $c'_{\La}(t)$ and $U^{+}_{m^c}$  as $t\to 0$.  Giving that  $[0]_t! =1$ and  
$\lim_{t\rightarrow 0} \, [n]_t = \lim_{t\rightarrow 0} \frac{1-t^n}{1-t}= 1 $, 
we have $f_{\La^s}(0) = 1$.  Moreover, 
\begin{equation}
\lim_{t\to 0}T_i = \frac{-x_{i+1}}{x_i-x_{i+1}}(K_{i,i+1}-1)=\widehat \pi_i\,.
\end{equation}
Consequently, 
\begin{multline*} \lim_{t\rightarrow 0} P_\La(x,\theta;t,t)\\ = \lim_{t\rightarrow 0}c'_\La(t) \sum_{\sigma \in S_N/(  S_m \times S_{m^c} )} \K_\sigma \theta_1 \cdots \theta_m A_m \sum_{w\in S_{m^c}}T_w(t) E_{((\La^a)^R,\La^s)}(x;t,t)\\=   (-1)^{\binom m 2}  \sum_{\sigma \in S_N/(  S_m \times S_{m^c} )} \K_\sigma \theta_1 \cdots \theta_m A_m \widehat \Pi_{m^c} \widehat E_{((\La^a)^R,\La^s)}(x;0,0)
\end{multline*} 
where  $\widehat \Pi_{m^c}=\sum_{w\in S_{m^c}} \widehat \pi_w$.  Thus, according to \eqref{Einkey}, 
\begin{equation} P_\La(x,\theta;0,0)=(-1)^{\binom m 2}  \sum_{\sigma \in S_N/(  S_m \times S_{m^c} )} \K_\sigma \theta_1 \cdots \theta_m A_m \widehat \Pi_{m^c} \widehat \Ke_{((\La^a)^R,\La^s)}(x)\, . \label{sinkey} \end{equation}
    This shows that the Schur superpolynomial $s_\La(x,\theta)$ is well-defined.   

For the second family of Schur superpolynomials $\bar s_\La(x,\theta)$, let us first write the  $t$-symmetrizer  $U^+$ as follows (see \cite[eq. (2.25)]{Mar}):
\begin{equation*}
U^+_{m^c}f(x) = \frac{1}{\Delta_{m^c}} A_{m^c} \Big( \prod_{m< i<j} (x_i-tx_j) f(x) \Big) = \frac{t^{n(n-1)/2}}{\Delta_{m^c}} A_{m^c} \Big( \prod_{m< i<j} (\frac{1}{t} x_i-x_j) f(x) \Big) \, ,
\end{equation*}
where $n=N-m$.   Then, we  rewrite $f_{\La^s}(t)$ as 
\begin{multline*}
f_{\La^s}(t) = \prod_{i=1}^{\La^s_1} [n_{\La^s}(i)]_t!= \prod_{i=1}^{\La^s_1} t^{n_{\La^s}(i) (n_{\La^s}(i)-1)/2}  [n_{\La^s}(i)]_{1/t}!\\ = t^{   \sum_i n_{\La^s}(i)(n_{\La^s}(i)-1)/2   } f_{\La^s}(1/t)\, . 
\end{multline*}
For $\mu=(\mu_1, \mu_2, \ldots, \mu_N)$,  with $\mu_N\geq 0$, we define
\begin{equation}
\text{inv}(\mu)=N(N-1)/2 - \sum_i n_\mu(i)(n_\mu(i)-1)/2 \,  .
\end{equation}
In words, given all possible pairs of parts of $\mu$, if we subtract all the pairs formed between repeated parts, we get the inversion number.  By combining these formulas  with \eqref{PvsE} and \eqref{clambda}, we get
\begin{multline*} \lim_{t\rightarrow \infty} P_\La(x;t,t)\\
= \lim_{t \rightarrow \infty} c_\La(t) \sum_{\sigma \in S_N/(S_m \times S_{m^c})} \K_\sigma \theta_1 \cdots \theta_m A_m U^+_{m^c}E_{\La^R}(x;t,t)\qquad \qquad \qquad 
\\=\lim_{t \rightarrow \infty}  \frac{(-1)^{\binom m 2}  t^{n(n-1)/2} }{  t^{ n(n-1)/2-\text{inv}(\La^s)  }  f_{\La^s}(1/t)  t^{\text{inv}(\La^s)}  } \qquad \qquad\qquad \qquad \qquad\qquad    \\
\qquad \qquad \times \lim_{t \rightarrow \infty} \sum_{\sigma \in S_N/(S_m \times S_{m^c})} \K_\sigma \theta_1 \cdots \theta_m A_m \frac{1}{\Delta_{m^c}} A_{m^c} \prod_{m< i<j} (\frac{1}{t} x_i-x_j) E_{\La^R}(x;t,t)\, 
\end{multline*} which leads to 
\begin{multline*} \lim_{t\rightarrow \infty} P_\La(x;t,t)=(-1)^{\binom m 2} \sum_{\sigma \in S_N/(S_m \times S_{m^c})} \K_\sigma \theta_1 \cdots \theta_m A_m \frac{1}{\Delta_{m^c}}  A_{m^c} (-1)^{\binom n 2} \\ x_{m+1}^0 x_{m+2}^1 \cdots x_N^{n-1} E_{\La^R}(x;\infty,\infty) \, .
\end{multline*}
Finally, using \eqref{Einkey} and $\delta^{(n)}=(n-1, \ldots, 1,0, \ldots)$, we obtain  
\begin{multline} P_\La(x;\infty,\infty)= \\
 (-1)^{\binom m 2+\binom n 2}\sum_{\sigma \in S_N/(S_m \times S_{m^c})} \K_\sigma \theta_1 \cdots \theta_m A_m \frac{1}{\Delta_{m^c}} A_{m^c} (\omega x)^{\delta^{(n)}} \Ke_{( \La^s, \La^a )}(\omega x)\, , \label{sbarinkey}
\end{multline}
which shows that $\bar s_\La(x,\theta)$ is well-defined. 
\end{proof}

We recall that the Macdonald superpolynomials  are stable with respect to the number $N$ of indeterminates.  The same stability holds for  $s_\La=s_\La(x,\theta)$  and $\bar s_\La=\bar s_\La(x,\theta)$, since they are limits of the Macdonald superpolynomials. This allows us to consider $s_\La $  and $\bar s_\La $ as elements of the algebra $\mathscr{R}$ of  symmetric functions  in superspace (see Section \ref{ortho}).   The next corollary shows that within this context,  there is a natural duality between  $\bar s_\La $and $s_\La $.    

\begin{corollary}   Let  $\LL\,|\,\RR$ and $\omega$ respectively  denote the scalar product  $\LL\,|\,\RR_{q,t}$  with $q=t=1$, and  the homomorphism defined in \eqref{actionomega}.  Moreover, let \beq s^*_\La=\omega \bar s_{\La'}\, .\eeq  Then, \beq  \label{dualschur} \LL s_\Lambda^*\, |\,s_\Om  \RR_{}=\delta_{\La,\Om}\, .\eeq Equivalently,
\beq \prod_{i,j}(1-x_iy_j-\theta_i\phi_j)^{-1}=\sum_\La (-1)^{\binom{m}{2}}s^*_\La(x,\theta)\, s_\La(y,\phi)\, .
\eeq 
\end{corollary}
\begin{proof} The first result is obtained by taking the limit $q,t\to0$ in \eqref{dual2}. The second result follows from \eqref{dualschur} and \cite[Lemma 37]{DLMjaco}.\end{proof}

\section{Proof of a combinatorial identity}

The equivalence, as $N$ goes to $\y$, 
between the two norms in \eqref{2norms} relies on the equality 
\begin{equation}\label{ident}
{ \lim_{N\rw\y}  \frac{d'_{({\La^a}^R, \La^s)}\,d_{{\La^s}^R} }
{d_{(\La^a,{\La^s}^R)}}=\, \frac{{h_{\La^s}}\,h^\uparrow_\La}{h^\downarrow_\La} .}
\end{equation}
The limit needs to appear on the
right-hand side since there is a residual dependence upon $N$ in the ratio
$d_{{\La^s}^R}  /{d_{(\La^a,{\La^s}^R)}} $  (the reversed  partition ${\La^s}^R$ 
contains  
$N-m-\ell(\La^s)$ zeros).

We first introduce, as in \cite{Knop},
a convenient decomposition of the leg-lengths of a composition $\g$:
\begin{equation}\label{defl}l(s) = l^\uw(s)+l^\dw(s)\end{equation}
with
\begin{align} \label{defla}
l^\uw(s)&=\# \{k=1,...,i-1 \, | \, j \leq \gamma_k+1\leq \gamma_i \}\, ,\nonumber\\
 l^\dw(s)&= \# \{k=i+1,...,N \, |\,  j \leq\gamma_k \leq \gamma_i \} \, .
\end{align}
In order to better visualize expressions 
$l^\uw(s)$ and $l^\dw(s)$, we put 
(as in \cite{LLN}) a symbol at the end of each row of $\g$, 
here trading the French hexagon for a  triangle, e.g.,
\begin{equation}\gamma=(0,0,1,3,3) \quad \longrightarrow \quad {\tableau[scY]{ \bl \triangle \\\bl \triangle \\ &\bl \triangle \\&& & \bl \triangle\\&&&\bl \triangle}}\, .
\end{equation}
 For $s=(i,j)$,  $l^\uw(s)$ is given by the number of triangles above row $i$
in columns
$j'$
for $j\leq j'\leq \g_i$ (e.g.,  $l^\uw((4,1))=3$ in the above example), while $l^\dw(s)$ 
 is given by the number of triangles below row $i$ in columns
$j'$ for $j+1\leq j'\leq \g_i+1$
(e.g., $l^\dw((4,1))=1$ in the above example). 

\begin{lemma} Identity \eqref{ident} is equivalent to
\begin{equation}\label{identa}
{ \frac{d'_{({\La^a}^R, \La^s)}}{{h_{\La^s}}}  \left[\frac{d_{{\La^s}^R} }
{d_{(\La^a,{\La^s}^R)}}\right]_{0}
=\frac{ \,h^\uparrow_\La} {h^\downarrow_\La} }\, ,
\end{equation}
where
\begin{equation}
 \left[\frac{d_{{\La^s}^R} }
{d_{(\La^a,{\La^s}^R)}}\right]_{0} =
 \frac{d_{{\La^s}^R} } 
{d_{(\La^a,{\La^s}^R)}} 
\left( \prod_{(i,1)\in\La^s} \frac{d_{{\La^s}^R}(s) } {d_{(\La^a,{\La^s}^R)}(s)} 
\right)^{-1} \, ,
\end{equation}
that is, the products corresponding the the cells in the first columns of
$\Lambda^s$ were removed from $d_{{\La^s}^R} / d_{(\La^a,{\La^s}^R)}$. 
\end{lemma}
\begin{proof}
We first isolate the part of $d_{{\La^s}^R} / d_{(\La^a,{\La^s}^R)}$
that depends upon $N$:
\begin{equation}
  \lim_{N\rw\y}  \frac{d_{{\La^s}^R} }
{d_{(\La^a,{\La^s}^R)}}=   \left[\frac{d_{{\La^s}^R} }
{d_{(\La^a,{\La^s}^R)}}\right]_0
\; \lim_{N\rw\y} \prod_{(i,1)\in\La^s} \frac{d_{{\La^s}^R}(s) } {d_{(\La^a,{\La^s}^R)}(s)}.
\end{equation}
It thus suffices to prove that
\begin{equation}\label{ratio}
\lim_{N\rw\y} \prod_{(i,1)\in\La^s} \frac{d_{{\La^s}^R}(s) } {d_{(\La^a,{\La^s}^R)}(s)}=  1\, .\end{equation}
The last equality is rather clear: it is a  ratio of terms of the form 
$1-q^{a(s)+1}t^{l(s)+1}$ for which the leg-length $l(s)$ tends to infinity
as $N$ goes to infinity (the number of triangles above $(i,1)  \in
\Lambda^s$ goes to infinity).
Take for instance the case $\La=(2,1;2,1)$ and introduce $M$ zeros represented by $\triangle^M$:
\begin{equation}
(2,1,0^M,1,2):\quad {\tableau[scY]{&&\bl\triangle\\& \bl\triangle\\ \bl \,\,\triangle^M \\ &\bl \triangle \\&&\bl \triangle\\}}\, ,
\qquad\qquad\qquad (0^M,1,2):\quad {\tableau[scY]{ \bl \,\,\triangle^M \\ &\bl \triangle \\&&\bl \triangle\\}}\, .
\end{equation}
The ratio on the left-hand side becomes
\begin{equation}
\lim_{M\rw\y}\ \prod_{(i,1)\in\La^s} \frac{d_{{\La^s}^R}(s) } {d_{(\La^a,{\La^s}^R)}(s)}=
\lim_{M\rw\y} \frac{(1-q t^{M+1})(1-q^{2}t^{M+2})}{(1-q t^{M+1})(1-q^{2}t^{M+3})} =1.
\end{equation}
\end{proof}

\begin{proposition} \label{con1} Identity \eqref{identa} holds.
\end{proposition}

\begin{proof}
First, observe that the identity \eqref{ident} is satisfied identically (for all $N$ actually) when 
$\La^a=\emptyset$ and $\La^s=\mu$, since
\begin{equation}
\text{l.h.s. \eqref{ident}}\,=  \lim_{N\rw\y}\frac{d'_{ \mu}\,d_{{\mu}^R} }{d_{{\mu}^R}}= d'_\mu \qquad {\rm and} \qquad 
\text{r.h.s. \eqref{ident}}\, =\frac{{h_{\mu}}\,h^\uparrow_\mu}{h^\downarrow_\mu}=\frac{{h_{\mu}}\,d'_\mu}{h_\mu}=d'_\mu
.\end{equation}
The result is thus true in the case $\La^a=\emptyset$. 
We thus suppose by induction that
\eqref{identa} holds for some  $\La=(\La^a;\La^s)$ and the aim is to prove that \eqref{identa} is still valid if we add a fermionic row $b>\La^a_1$ to
obtain $\tilde \Lambda=({\tilde \Lambda}^a;\Lambda^s)$, where 
$\tilde {\Lambda}^a=(b,\Lambda_1^a,\Lambda_2^a,\dots)$.
Defining
\begin{equation}
{\mathbf \Delta} F(\La) =\frac{F(\Lt)}{F(\La)}\, ,
\end{equation}
it thus suffices to demonstrate that
\begin{equation}\label{cgcd}
{\mathbf \Delta} \,\text{l.h.s.}\,\eqref{identa}= 
{\mathbf \Delta} \,\text{r.h.s.}\,\eqref{identa} \, .
 \end{equation}
Given that $\Lt^s=\La^s$, the factors $d_{{\La^s}^R} $ and $h_{{\La^s}} $ are not affected by the transformation $\La\rw\Lt$.  Hence
\begin{equation}
{\mathbf \Delta} \,\text{l.h.s.}\,\eqref{identa}= 
{\mathbf \Delta}
\left( \frac{d'_{({\La^a}^R, \La^s)}}{ \left[
{d_{(\La^a,{\La^s}^R)}}\right]_0
 } \right)
\end{equation}
Finally, \eqref{cgcd} will follow from the two relations
\begin{equation}
\label{2rel}
(1):\,{\mathbf \Delta}\,
d'_{({\La^a}^R, \La^s)}= {\mathbf \Delta} { \,\bar h^\uparrow_\La} \qquad 
{\rm and} \qquad
(2):\, {\mathbf \Delta} \left[
{d_{(\La^a,{\La^s}^R)}}\right]_0
={\mathbf \Delta} {??bar h^\downarrow_\La} 
\end{equation}
where $\bar h^{\uw \dw}$ is defined as $h^{\uw \dw}$ except that the product runs over all squares of $\La$:
\begin{equation}\bar h^\uparrow_\La=\prod_{s\in\La} (1-q^{a_{\La^*}(s)+1}t^{{l}_{\Lambda^{\circledast}}(s)})\qquad\text{and}\qquad\bar h^\downarrow_\La=\prod_{s\in\La}(1- q^{{a}_{\Lambda^{\circledast}}(s)}t^{l_{\La^*}(s)+1}).\end{equation} 
Observe that,
\begin{equation}
\frac{\bar h^\uparrow_\La}{\bar h^\downarrow_\La}=\frac{h^\uparrow_\La}{h^\downarrow_\La}
\end{equation}
and thus, as claimed, we only have to prove the two
 relations \eqref{2rel}.

There are two types of contributions to ${\mathbf \Delta}$: those corresponding to
the modifications to the leg-lengths of the squares of $\La^s$ in rows of length  larger than $b$, denoted by 
${\mathbf \Delta}_1$, and those corresponding to
the squares of the added fermionic row, denoted by ${\mathbf \Delta}_2$. They will be treated separately.

Consider first the variation
${\mathbf \Delta}_1$ 
and the relation (1). In this case, $d'_\gamma \rightarrow d'_{\tilde \gamma}$ where
  $\gamma=({\La^a}^ R,\La^s)=$ and $\tilde \gamma=({\La^a}^R,b,\La^s)
=({{(\tilde \La}^a)^R},\La^s)$. 
  The leg-lengths of the squares in columns $1\leq j\leq b+1$ and rows    $\gamma_i>b$ is increased by 1 by adding the new fermionic row.
We have thus:
\begin{equation}
\label{rel11}\,{\mathbf \Delta}_1\,
d'_{\gamma}= \prod_{\substack{s=(i,j)\in\gamma\\\gamma_i>b\\1\leq j\leq b+1}}\frac{1-q^{a_\gamma(s)+1}t^{l_\gamma(s)+1}}{1-q^{a_\gamma(s)+1}t^{l_\gamma(s)}}
\end{equation}
where we indicated explicitly with respect to which diagram the arm- and leg-lengths are calculated. 
Consider now the corresponding variation of $h^\uw$. Since the expression of $\bar h^\uw$ involves $l_{\La^\cd}$, 
the addition of the fermionic row of size $b$ increases by 1 the leg-length of  the squares in column $1\leq j\leq b+1$ of rows  $\La_i^s>b$:
\begin{equation}
\label{rel12}\,{\mathbf \Delta}_1\,
\bar h^\uw_\La= \prod_{\substack{s=(i,j)\in\La^s \\\La_i^s>b\\1\leq j\leq b+1}}\frac{1-q^{a_{\La^*}(s)+1}t^{l_{\La^\cd}(s)+1}}{1-q^{a_{\La^*}(s)+1}t^{l_{\La^\cd}(s)} }.
\end{equation}
In order to compare expressions  \eqref{rel11} and \eqref{rel12}, 
we need to clarify the meaning of the entries in each product.
Note that in the first product, the rows
  $\gamma_i>b$  are such that $\gamma_i=\Lambda_{i'}^s$ for some $i'$.
 The product is thus over the same number of squares in the two  cases.
 Now, in reordering the rows of 
   $\gamma$ to get $\La^*$ (an operation that does not affect the arm-lengths), we readily see that $a_\gamma =a_{\La^*}$ in \eqref{rel11}. 
 Next, to  compare the leg-lengths, we  note that for $s=(i,j)$, the leg-length
    $l_\gamma(s)=l_\gamma^\dw(s)+l^\uw_\gamma(s)$ in \eqref{rel11} is such that
\begin{equation}
l_\gamma^\dw=\# \{k > i' \, |\, j \leq \La^s_k  \}
\end{equation}
and
\begin{equation}\label{eqlup1}
l_\gamma^\uw =\# \{k <i \, |\, j \leq \gamma_k +1\leq \gamma_i \} = \# \{ k \, |\, j \leq {\tilde \La}_k^a+1\} \, .
\end{equation}
The sum $l^\uw_\gamma+l^\dw_\gamma$ thus corrresponds exactly to the definition 
of $l_{\La^\cd}$ in \eqref{rel12}.
This demonstrates the equivalence of  relations \eqref{rel11} and \eqref{rel12}.

Consider next relation (2), focusing again on the contribution ${\mathbf \De}_1$.  Let us first obtain the variation resulting from $d_\eta\rightarrow d_{\tilde \eta} $, where  $\eta=(\La^a,{\La^s}^R)$ and $\tilde \eta = (b,\La^a,{\La^s}^R)$.  The analysis of ${\mathbf \Delta}_1 d_\eta$ is similar to the one above for ${\mathbf \Delta}_1 d_\gamma'$, except that we need to  keep in mind that the contribution of the first column is not considered anymore. We find
\begin{equation}
\label{rel21}\,{\mathbf \Delta}_1\,
 \left[
{d_{\eta}}\right]_0
= \prod_{\substack{s=(i,j)\in\eta\\\eta_i>b\\2\leq j\leq b+1}}
\frac{1-q^{a_\eta(s)+1}t^{l_\eta(s)+2}}{1-q^{a_\eta(s)+1}t^{l_\eta(s)+1}}.
\end{equation}
Let us now turn to the corrresponding
  variation in $\bar h^\dw$. The expression for $\bar h^\dw$
involves  $a_{\La^\cd}$ which is the same as $a_{\La^*}$ for a bosonic row. Also, since the leg-length entering in $\bar h^\uw$ is $l_{\La^*}$, which does not count the possible  circle at the end of the column, the addition of the fermionic row 
of length $b$ increases by 1 the leg-lengths of the squares in columns   $1\leq j\leq b$ (and not $b+1$) of the rows $\La^s_i>b$:
\begin{equation}
\label{rel22}\,{\mathbf \Delta}_1\,\bar 
h^\dw_\La= \prod_{\substack{s=(i,j)\in\La^s\\\La_i^s>b\\1\leq j\leq b}}
\frac{1-q^{a_{\La^*}(s)}t^{l_{\La^*}(s)+2}}{1-q^{a_{\La^*}(s)}t^{l_{\La^*}(s)+1}}.
\end{equation}
 Before we can do a direct  comparison between \eqref{rel21} and \eqref{rel22}, 
we have to perform the substitution $j\rw j+1$ in \eqref{rel21}.
Since $a_\eta(i,j+1)+1=a_\eta(i,j)$, this gives
\begin{equation}
\label{rel21bis}\,{\mathbf \Delta}_1\,
 \left[
{d_{\eta}}\right]_0
= \prod_{\substack{s=(i,j+1)\in\eta\\\eta_i>b\\1\leq j\leq b}}
\frac{1-q^{a_\eta(i,j)}t^{l_\eta(s)+2}}{1-q^{a_\eta(i,j)}t^{l_\eta(s)+1}} \, .
\end{equation}
Letting $\eta_i = \Lambda_{i'}^s$, we see that the powers of $q$ are identical 
in both expressions.
Moreover, 
since the block 
 $\La^a$ still lies in the upper position of the composition, 
we can write the leg-length $l_\eta(i,j+1)$, using eq. \eqref{defla}, as 
\begin{equation}
 l_\eta(i,j+1)= \#\{ k \, | \, j \leq \La^a_k   \} + \# \{k>i' \, |\,  j \leq \La^s_k  < \Lambda_{i'}^s \} + \#\{k<i' \, |\, \La^s_k=\La^s_{i'}\}  \, , 
\end{equation}
which  corresponds, up to a reordering of the parts of $\Lambda^s$
of the same size as $\Lambda^s_{i'}$, 
precisely to the expression  of $l_{\La^*}$ in \eqref{rel22}.
We have thus verified the equivalence of \eqref{rel21} and \eqref{rel22}.

Consider now the variation ${\mathbf \De}_2$, namely, the  contribution of the squares of the added fermionic row of size $b$.
As before, we let $\gamma=({\La^a}^ R,\La^s)$, $\tilde \gamma=({\La^a}^ R,b,\La^s)$ and
$\tilde \Lambda=({\tilde \Lambda}^a;\Lambda^s)$ with 
$\tilde {\Lambda}^a=(b,\Lambda_1^a,\Lambda_2^a,\dots)$.  If the added fermionic row of size $b$ corresponds  to $\tilde \gamma_i$ (resp. 
$\tilde \Lambda_{i'}$) in $\tilde \gamma$ (resp.
$\tilde \Lambda$), the verification of  relation (1) amounts to compare
\begin{equation}
\label{re11}
{\mathbf \Delta}_2\, d'_{\gamma}= \prod_{1\leq j\leq b}
[1-q^{b-j+1}t^{l_{\tilde \gamma}(i,j)}]\quad \text{and}\quad
{\mathbf \Delta}_2\, \bar h^\uw_{\La}= \prod_{1\leq j\leq b}
[1-q^{b-j+1}t^{l_{\tilde \La^\cd}(i',j)}] \, .
\end{equation}
The powers of $q$ are manifestly the same in the two 
  contributions.
The expression for the leg-length on the l.h.s. is
\begin{equation}
l_{\tilde \gamma}(i,j)=l^\uw(i,j)+l^\dw(i,j)\, ,
\end{equation}
where 
\beq    l^\uw(i,j) =\#\{k\, |\, j \leq \La^a_k+1\}\quad \text{and}\quad l^\dw(i,j)=\#\{k\,|\, j\leq \La^s_k\leq b\}\, .\eeq
This form of $l_{\tilde \gamma}(i,j)$ is clearly the same as the expression of
 $l_{\tilde \La^\cd}(i',j)$ on the r.h.s, which demonstrates the equivalence  of the two variations in \eqref{re11}.

It only remains to establish the correctness of relation (2) under ${\mathbf \Delta}_2$.  Recall that 
$\eta=(\La^a,{\La^s}^R)$ and $\tilde \eta = (b,\La^a,{\La^s}^R)$.
If the added fermionic row of size $b$ corresponds  to $\tilde \Lambda_{i}$ in 
$\tilde \Lambda$, the verification of  relation (2) amounts to compare
\begin{equation}
\label{re12}
{\mathbf \Delta}_2\, d_{\eta}= 
\prod_{1\leq j\leq b}
[1-q^{b-j+1}t^{l_{\tilde \eta}(1,j)+1}]
\end{equation}
and \beq  
{\mathbf \Delta}_2\, h^\dw_{\La}=
 \prod_{1\leq j\leq b}[1-q^{b-j+1}t^{l_{\tilde \La^*}(i,j)+1}]\, ,\eeq
where we used the fact that $a_{{\tilde \Lambda}^\circledast}(i,j)=
a_{{\tilde \Lambda}^*}(i,j)+1$ given that row $i$ of $\tilde \Lambda$ is fermionic. 
 Again the powers of $q$ match.  We have
\begin{equation}
l_{\tilde \eta}(1,j)=l^\dw(1,j)=\#\{k \, | \, j \leq \La^a_k \leq b {\rm ~or~}  j\leq \La^s_k \leq b \}\, ,
\end{equation}
which corresponds to $l_{\tilde \La^*}(i,j)$ on the r.h.s.
This completes the demonstration of the two relations \eqref{2rel} and thus of Proposition \ref{con1}.
 \end{proof}

Let us consider an example: $\La=(0;4,1)$ and $\Lt=(2,0;4,1)$.
We have thus added a fermionic row of length 2. We ignore the zeros 
of $\La^s$ given that they contribute only to the first column 
which is removed.  
The triangle at the end of the added row is written in parenthesis and the  circle of the added row is marked by a $\times$:
\begin{equation}
\tilde \gamma=(0,2,4,1):\quad {\tableau[scY]{\bl\triangle\\&& \bl\,(\!\triangle\!)\\&&&& \bl \triangle \\ &\bl \triangle\\}}
\quad\quad
\tilde \eta= (2,0,1,4):\quad {\tableau[scY]{&&\bl\,(\!\triangle\!)\\ \bl\triangle\\&\bl \triangle \\ &&&& \bl \triangle\\}}
\quad\quad
\end{equation}
$$ \tilde \Lambda=(2,0;4,1):\quad {\tableau[scY]{&&&&\bl\\&&\bl\tcercle{$\times$}\\&\bl \\ \bl\tcercle{}\\ }}\, .$$
We have thus,  for the product of the two variations ${\mathbf \De}_1$ 
and ${\mathbf \De}_2$,
\begin{align}
\label{ex11}
&\frac{d'_{(0,2,4,1)} }{d'_{(0,4,1)}}= \frac{(1-q^4t^3)(1-q^3t)(1-q^2t)}{(1-q^4t^2)(1-q^3)(1-q^2)} \times (1-q^2t^2)(1-q)
=\frac{h^\uw_{(2,0;4,1)} }{h^\uw_{(0;4,1)}}
\\
&\frac{d_{(2,0,1,4)} }{d_{(0,1,4)}}= \frac{(1-q^3t^3)(1-q^2t^2)}{(1-q^3t^2)(1-q^2t)} \times(1-q^2t)(1-q)=
\frac{h^\dw_{(2,0;4,1)} }{h^\dw_{(0;4,1)}}. \nonumber
\end{align}

\end{appendix}

\end{document}